\documentclass[%
 aip,jmp,
 amsmath,amssymb,
 reprint,%
]{revtex4-1}
\usepackage[T1]{fontenc}
\usepackage[utf8]{inputenc}

\usepackage{amsmath}
\usepackage{amssymb}
\usepackage{amsthm}
\usepackage{stmaryrd}
\usepackage{dsfont}
\usepackage{bm}
\usepackage{slashed}

\usepackage{mathptmx}

\usepackage{graphicx}
\usepackage[export]{adjustbox}
\usepackage{xcolor}
\usepackage{tikz}
\usetikzlibrary{
  arrows,
  arrows.meta,
  backgrounds,
  decorations,
  decorations.pathmorphing,
  matrix,
  patterns,
  positioning,
  shapes,
  snakes,
  trees
}

\usepackage{dcolumn}
\usepackage{enumitem}
\usepackage{subcaption}
\usepackage{tocloft}

\usepackage{fancyhdr}
\usepackage{hyperref}

\tikzset{
  bag/.style={align=center}
}

\newtheorem{theorem}{Theorem}
\newtheorem{proposition}{Proposition}
\newtheorem{corollary}{ Corollary}
\newtheorem{lemma}{Lemma}
\newtheorem{remark}{Remark}

\theoremstyle{definition}
\newtheorem{definition}{Definition}

\begin{document}

\preprint{AIP/123-QED}

\title[Triviality in a Non-Perturbative Second-Order
Mean-Field Theory for $\phi^4_4$]
{Triviality in a Non-Perturbative Second-Order
Mean-Field Theory for $\phi^4_4$}
% Force line breaks with \\

\author{Majdouline Borji}
\email{majdouline.borji@kfupm.edu.sa}
\affiliation{King Fahd Univeristy of Petroleum and Minerals, Department of Mathematics, Saudi Arabia}%

\date{\today}% It is always \today, today,



\begin{abstract}
We introduce a second-order mean-field description of the
four-dimensional Euclidean $\phi^4$ model within the
Wilson--Polchinski renormalization-group framework. The construction
is based on the connected amputated Schwinger functions evaluated at
the symmetric momentum configurations $
(p,-p,\ldots,p,-p)$,
which are decomposed into a momentum-independent component, a component
quadratic in $p$, and a higher-order remainder. The first two components are chosen to satisfy a closed nonlinear
hierarchy. We prove the existence
of solutions to this hierarchy for arbitrary positive bare coupling
and establish their convergence to the Gaussian fixed point as the
ultraviolet cutoff is removed. In particular, both the
momentum-independent and the quadratic momentum sectors are
asymptotically trivial.
\end{abstract} 

\maketitle

\section{Introduction}

The ultraviolet behaviour of the four-dimensional Euclidean
$\phi^4$ model is one of the central problems of constructive quantum
field theory. Perturbation theory suggests that a continuum limit with
positive bare coupling can only be obtained if the renormalized
coupling vanishes logarithmically as the ultraviolet cutoff is removed.
The limiting theory should therefore be governed by a Gaussian fixed
point, which is known as \emph{triviality}; see,
for instance, \cite{WilsonKogut1974,Gallavotti1985} for the
renormalization-group background.

The first rigorous results in this direction were obtained by Aizenman
and Fröhlich. Aizenman proved triviality of continuum limits of lattice
$\phi^4_d$ models in dimensions $d>4$ and established the corresponding
mean-field behaviour of the Ising model \cite{Aizenman1981}. Fröhlich subsequently proved that the one and two component lattice
$\lambda|\phi|^4$ models have only free scaling
limits in dimensions $d\ge5$. In the marginal dimension $d=4$, working in the symmetric phase and under a
uniform decay assumption on the two-point function, he established
triviality whenever the continuum limit exhibits infinite
field-strength renormalization \cite{Frohlich1982}. More recently, Aizenman and Duminil-Copin proved
triviality of the scaling limits of the critical Ising model and of a
class of nearest-neighbour reflection-positive lattice $\phi^4_4$
models \cite{AizenmanDuminilCopin2021,AizenmanDuminilCopin2024}. Their
argument combines the random-current representation with a logarithmic
improvement of the tree-diagram bound. These results give a rigorous
description of triviality at the level of lattice scaling limits. They
do not, however, directly identify the corresponding mechanism in
terms of the continuum Wilsonian flow of the effective interaction.

The purpose of the present paper is to study this mechanism within a
non-perturbative reduction of the Wilson--Polchinski flow equations that
retains the quadratic momentum dependence of the Schwinger functions.
We refer to the resulting system as the \emph{second-order mean-field
model}, where "second order" refers to the order of the momentum
expansion.

The Wilson--Polchinski flow equation provides an exact differential
formulation of the renormalization group for the effective interaction
\cite{Polchinski1984}. Although it is a fundamental tool in
perturbative renormalization theory, its direct non-perturbative
analysis remains substantially more difficult. Kopper
\cite{Kopper2022} introduced a mean-field approximation in which the
connected amputated Schwinger functions are replaced by their
zero-momentum values. The resulting hierarchy can be studied without
expanding in the coupling constant, leading to a rigorous construction
of a trivial solution for the one-component $\phi^4_4$ model under the assumption of a small coupling.

This restriction was subsequently removed by Kopper and Wang
\cite{KopperWang2025}. Their construction applies to arbitrary positive
bare coupling, replaces the auxiliary infrared regularization used in
\cite{Kopper2022} by a physical mass, and extends the analysis to the
$O(N)$ model. The relation between the resulting non-perturbative
solution and the perturbative expansion of the mean-field flow
equations is further investigated in
\cite{KopperWangPerturbation}. These works demonstrate that suitable
reductions of the Polchinski hierarchy retain enough of its
combinatorial structure to yield rigorous non-perturbative information
about the ultraviolet behaviour of four-dimensional scalar models.

The mean-field systems studied in
\cite{Kopper2022,KopperWang2025} retain only the momentum-independent
part of the connected amputated Schwinger functions. They therefore
describe the mass and coupling sectors, but do not capture the
quadratic momentum dependence associated with wave-function
renormalization. Our aim is to extend this framework to the quadratic
momentum sector while preserving both closure of the resulting
hierarchy and its non-perturbative solvability.

To this end, we evaluate the connected amputated Schwinger functions
at the distinguished momentum configurations
\[
(p,-p,\ldots,p,-p)
\]
and set
\[
T_{2n}^{\alpha,\alpha_0}(p)
:=
\mathcal L_{2n}^{\alpha,\alpha_0}
(p,-p,\ldots,p,-p),
\qquad n\ge1.
\]
Here, \(\mathcal L_{2n}^{\alpha,\alpha_0}\) denotes the connected
amputated \(2n\)-point Schwinger function with ultraviolet cutoff
\(\alpha_0\) and flow parameter \(\alpha\). These alternating
configurations satisfy momentum conservation and provide a natural
one-momentum restriction of the Wilson--Polchinski hierarchy.

We introduce the second-order mean-field ansatz
\[
T_{2n,\mathrm{MF2}}^{\alpha,\alpha_0}(p)
:=
A_{2n}^{\alpha,\alpha_0}
+
2n|p|^2B_{2n}^{\alpha,\alpha_0}
\]
and define the corresponding error through the exact decomposition
\begin{equation}\label{eq:IntroDecomposition}
T_{2n}^{\alpha,\alpha_0}(p)
=
A_{2n}^{\alpha,\alpha_0}
+
2n|p|^2B_{2n}^{\alpha,\alpha_0}
+
R_{2n}^{\alpha,\alpha_0}(p).
\end{equation}
The family \(A_{2n}^{\alpha,\alpha_0}\) describes the
momentum-independent part retained by the mean-field closure, whereas
\(B_{2n}^{\alpha,\alpha_0}\) introduces its first nontrivial
momentum-dependent correction. The term
\(R_{2n}^{\alpha,\alpha_0}\) is the exact error of this
approximation.

It is important to stress that the decomposition
\eqref{eq:IntroDecomposition} is not a Taylor decomposition of the
exact correlator. In particular,
\(A_{2n}^{\alpha,\alpha_0}\) and
\(B_{2n}^{\alpha,\alpha_0}\) need not coincide with the constant and
quadratic Taylor coefficients of
\(T_{2n}^{\alpha,\alpha_0}\), and the error
\(R_{2n}^{\alpha,\alpha_0}\) is not required to have a vanishing
quadratic jet at the origin. The terminology ``second-order'' refers
instead to the momentum dependence retained by the closure.

The main structural observation is that the families
\[
\bigl(A_{2n}^{\alpha,\alpha_0}\bigr)_{n\ge1}
\qquad\text{and}\qquad
\bigl(B_{2n}^{\alpha,\alpha_0}\bigr)_{n\ge1}
\]
can be chosen to satisfy an autonomous nonlinear hierarchy. This
hierarchy is obtained directly from the Wilson--Polchinski equation
without expanding in the coupling constant and retains the
combinatorial structure of both its linear and bilinear terms. Since
it incorporates momentum dependence up to quadratic order, we call it
the \emph{second-order mean-field hierarchy}.

This construction extends the momentum-independent mean-field systems
of \cite{Kopper2022,KopperWang2025} by introducing a quadratic
momentum sector. In particular, the two-point coefficient \(B_2\)
provides within the closure the degree of freedom associated with
wave-function renormalization, which is absent from a purely
momentum-independent approximation. The resulting hierarchy therefore
retains, at the level of the approximation, the mass, coupling, and
wave-function parameters.

After a suitable rescaling, the second-order mean-field hierarchy
becomes an infinite-dimensional dynamical system for dimensionless
variables. We construct non-perturbative solutions to this system for
arbitrary positive bare coupling and establish uniform bounds on their
Taylor coefficients. These estimates control the complete coefficient hierarchy and allow
the corresponding jets to be realised by smooth functions of the
rescaled flow parameter. At the fixed infrared scale, we prove that both families
\(A_{2n}^{\alpha,\alpha_0}\) and
\(B_{2n}^{\alpha,\alpha_0}\) converge to zero as
\(\alpha_0\downarrow0\). Thus the second-order mean-field closure is
asymptotically trivial.

The construction also gives a non-perturbative characterization of the
parameters entering the closure. Its mass and coupling parameters are
encoded by the corresponding two- and four-point components of the
family \(\bigl(A_{2n}\bigr)_{n\ge1}\), whereas its wave-function
parameter is encoded by \(B_2\). More generally, the ultraviolet
boundary values of the families
\(\bigl(A_{2n}\bigr)_{n\ge1}\) and
\(\bigl(B_{2n}\bigr)_{n\ge1}\) are determined by an explicit ansatz
compatible with the autonomous hierarchy. These parameters are
therefore obtained as part of the non-perturbative solution rather
than through an order-by-order expansion in the renormalized coupling.

Although the mean-field closure is approximate, its error is described
exactly. Substituting \eqref{eq:IntroDecomposition} into the restricted
Wilson--Polchinski equation yields an induced evolution equation for
\(R_{2n}^{\alpha,\alpha_0}\). This equation determines the discrepancy between the exact restricted
correlator and its second-order mean-field approximation. The coupled
family
\[
\bigl(
A_{2n}^{\alpha,\alpha_0},
B_{2n}^{\alpha,\alpha_0},
R_{2n}^{\alpha,\alpha_0}
\bigr)_{n\ge1}
\]
then reconstructs the original restricted hierarchy. No smallness,
ultraviolet decay, or renormalization-group irrelevance of
\(R_{2n}^{\Lambda,\Lambda_0}\) is asserted in the present work.

A natural continuation of this analysis is to study the perturbative
expansions of
\(\bigl(A_{2n}\bigr)_{n\ge1}\),
\(\bigl(B_{2n}\bigr)_{n\ge1}\), and
\(\bigl(R_{2n}\bigr)_{n\ge1}\). In particular, one may ask whether the
perturbative coefficients of the error can be controlled in terms of
those of the mean-field families and whether suitable summability
estimates can be established. Such estimates, combined with the
non-perturbative construction developed here, could provide a route
towards reconstructing and analysing the full correlators on the
symmetric momentum configurations
\[
(p,-p,\ldots,p,-p).
\]
This problem lies beyond the scope of the present paper.

The main contributions of this work can therefore be summarized as
follows. We derive the exact flow equation for the correlators
restricted to symmetric momentum configurations and introduce an
autonomous second-order mean-field equations retaining both constant and
quadratic momentum contributions. We construct this hierarchy
non-perturbatively and prove its ultraviolet triviality for arbitrary
positive bare coupling. Finally, we derive an exact evolution equation
for the error of the expansion and establish the reconstruction of the
restricted Wilson--Polchinski flow equations from the mean-field solution
and its complementary error.

The paper is organized as follows.
In Section~\ref{sec2}, we restrict the Wilson--Polchinski hierarchy to
the symmetric momentum configurations, introduce the decomposition into
momentum-independent, quadratic, and complementary components, and
derive the autonomous flow equations for the mean-field coefficients
\(A_n^{\alpha,\alpha_0}\) and \(B_n^{\alpha,\alpha_0}\), together with
the exact evolution equation for the remainder. We conclude the section
with a reconstruction theorem showing that the three components recover
the restricted Wilson--Polchinski hierarchy.
In Section~\ref{sec3}, we pass to dimensionless variables and derive
the second-order mean-field hierarchy whose non-perturbative solutions
are studied in the remainder of the paper.
Sections~\ref{sec4}--\ref{SecDiscreteHierarchy0} develop the framework required for
this construction: we introduce suitable weighted sequence classes,
derive the discrete hierarchy satisfied by the Taylor coefficients,
reformulate it in terms of linear and bilinear operators, and establish
the corresponding mapping properties.
Finally, in Section~\ref{sec7}, we use these estimates to construct
smooth non-perturbative solutions of the mean-field hierarchy and prove
their ultraviolet triviality.

\section{Exact decomposition of the Wilson--Polchinski hierarchy}\label{sec2}

In this section, we restrict the Wilson--Polchinski flow equations to a
symmetric momentum manifold and decompose the resulting correlators
into their momentum-independent, quadratic, and higher-order
components. This yields a closed nonlinear hierarchy for the relevant
sector and an exact flow equation for the complementary remainder. We
conclude by showing that this decomposition reconstructs the original
flow on the symmetric momentum manifold.

\subsection{Flow equations for the connected amputated Green functions}
We formulate the theory with an ultraviolet (UV) cutoff $\alpha_0$ and an
infrared (IR) cutoff $\alpha$ with $0\le \alpha_0\le \alpha<+\infty$, in the standard path-integral formalism. We define the regularized momentum space propagator as
\begin{equation}\label{eq:regularised_propagator}
C^{\alpha,\alpha_0}(p;m)
=
\frac{1}{p^2+m^2}
\left[
\exp\!\left(-\alpha_0(p^2+m^2)\right)
-
\exp\!\left(-\alpha(p^2+m^2)\right)
\right].
\end{equation}
The derivative with respect to $\alpha$ is denoted by $\dot{C}^{\alpha}$.
Upon removal of the cutoffs, i.e.\ in the limit
\(\alpha_0\to0\),
\(\alpha\to\infty\), we indeed recover the free propagator
\[
C(p)=\frac{1}{p^2+m^2}.
\]
For the Fourier transform, we use the convention
\begin{equation}\label{eq:FourierConvention}
\phi(x)
=
\int_p
\hat \phi(p)e^{ipx}
:=
\int_{\mathbb R^4}
\frac{d^4p}{(2\pi)^4}
\,e^{ipx}\,
\hat \phi(p).
\end{equation}
The effective action is assumed to be of the form
\begin{equation}\label{EqBareInteraction}
L^{\alpha_0}(\varphi)
=
\int_{\mathbb R^4}
\left(c_{0,4}(\alpha_0)\,
\varphi^4(x)+c_{0,2}(\alpha_0)\,
\varphi^2(x)+c'_{0,2}(\alpha_0)\,
\varphi(x)\Delta\varphi(x)\right)~
dx.
\end{equation}
We assume that there exists a constant \(K'\in\mathbb R\),
independent of \(\alpha_0\), such that
\begin{equation}\label{EqBareInteractionLowerBound}
L^{\alpha_0}(\varphi)
\ge
K',
\qquad
\mu^{\alpha,\alpha_0}\text{-a.s.}
\end{equation}
where $\mu^{\alpha,\alpha_0}$ is the Gaussian measure with mean zero and covariance $C^{\alpha,\alpha_0}$.
This lower bound guarantees the integrability of the Gibbs weight with
respect to the Gaussian measure.

The coefficients \(c_{0,n}(\alpha_0)\) must be chosen as suitable
functions of the ultraviolet cutoff \(\alpha_0\). We begin by recalling the exact Wilson--Polchinski flow equations satisfied by
the connected amputated Green functions. Throughout this work we adopt the
following definition of the $n$-point Schwinger functions
\begin{equation}\label{def:Ln}
\frac{\delta^n L^{\alpha,\alpha_0}}
{\delta\phi(p_1)\cdots\delta\phi(p_n)}
\Big|_{\phi\equiv0}
=
n!~
\delta^{(4)}
\left(
\sum_{i=1}^n p_i
\right)
\mathcal L_n^{\alpha,\alpha_0}(p_1,\ldots,p_n),
\end{equation}
where $\delta^{(4)}$ denotes the four-dimensional Dirac distribution enforcing
momentum conservation.

The effective interaction satisfies the Wilson--Polchinski flow equation
\begin{equation}\label{Polchinski}
\partial_{\alpha}
\bigl(
L^{\alpha,\alpha_0}
+
I^{\alpha,\alpha_0}
\bigr)
=
\frac12
\Big\langle
\frac{\delta}{\delta\phi},
\dot C^\alpha
\frac{\delta}{\delta\phi}
\Big\rangle
L^{\alpha,\alpha_0}
-
\frac12
\Big\langle
\frac{\delta L^{\alpha,\alpha_0}}{\delta\phi},
\dot C^\alpha
\frac{\delta L^{\alpha,\alpha_0}}{\delta\phi}
\Big\rangle,
\end{equation}
where $\langle\cdot,\cdot\rangle$ denotes the usual $L^2$ pairing and
$\delta/\delta\phi$ the functional derivative.

Applying functional derivatives to \eqref{Polchinski} and evaluating at
$\phi\equiv0$ yields the flow equations for the connected
amputated Green functions. The linear contribution is immediate. The quadratic contribution is
obtained by applying the Leibniz rule and grouping together derivatives acting
on each factor. Using the normalization
\eqref{def:Ln} yields the combinatorial coefficients
$\binom{n+2}{2}$ and $(n_1+1)!(n_2+1)!$. Hence the flow equations verified by the connected amputated Green functions is given by
\begin{multline}
\partial_{\alpha}
\mathcal L_n^{\alpha,\alpha_0}
(p_1,\ldots,p_n)
=
\binom{n+2}{2}
\int_k
\dot C^\alpha(k)
\mathcal L_{n+2}^{\alpha,\alpha_0}
(p_1,\ldots,p_n,k,-k)\\
-
\frac12
\sum_{{\pi_1,\pi_2}}
\frac{(n_1+1)!(n_2+1)!}{n!}~
\dot C^\alpha(p_{\pi})~
\mathcal L_{n_1+1}^{\alpha,\alpha_0}
(\vec p_{\pi_1},p_{\pi})
\mathcal L_{n_2+1}^{\alpha,\alpha_0}
(\vec p_{\pi_2},-p_{\pi})~.
\label{flow-Ln}
\end{multline}
Here the sum runs over ordered partitions $
(\pi_1,\pi_2)$
of \(\{1,\dots,n\}\) into two disjoint subsets, with
$$\pi_1\cup\pi_2=\left\{1,\cdots,n\right\},\qquad |\pi_i|=n_i,\qquad n_1+n_2=n$$
and 
\[
p_\pi
:=
-\sum_{i\in\pi_1}p_i
=
\sum_{i\in\pi_2}p_i.
\]

The ordering of the elements within each subset is not important, whereas
the two subsets are regarded as distinguished. The factor \(1/2\)
compensates for the simultaneous exchange
\[
(\pi_1,\pi_2)\longleftrightarrow(\pi_2,\pi_1)
\]
of the two vertex factors.
\noindent In the sequel, we use the following notation: for $q\in\mathbb R^4$ and $r\in\mathbb N$, we write
\[
q^{[r]}
:=
(\underbrace{q,\ldots,q}_{r\text{ entries}}).
\]
and 
$$
(q,-q)^{[r]}
:=
(\underbrace{q,-q,\ldots,q,-q}_{2r\text{ entries}}).
$$
\begin{proposition}[Flow equation on the symmetric momentum configuration]
\label{prop:symmetric-flow}
Let \(n\ge2\) be even and let \(p\in\mathbb R^4\). For
\(n_1+n_2=n\) and \(0\le r\le n_1\), define
\begin{equation}\label{DefAlphaRnOne}
\alpha_{r,n_1}
:=
\binom{n/2}{r}
\binom{n/2}{n_1-r},
\end{equation}
with the convention
\[
\binom{i}{j}=0,
\qquad
j\notin\{0,\dots,i\}.
\]
Then
\begin{multline}\label{FE+-}
\partial_{\alpha}
\mathcal{L}_n^{\alpha,\alpha_0}
\bigl((p,-p)^{[n/2]}\bigr)
=
\binom{n+2}{2}
\int_k
\dot C^\alpha(k)\,
\mathcal{L}_{n+2}^{\alpha,\alpha_0}
\bigl((p,-p)^{[n/2]},k,-k\bigr)
\\
-\frac12
\sum_{\substack{n_1+n_2=n\\ n_1,n_2\ge1,~n_i\in2\mathbb{N}+1}}
\frac{(n_1+1)!(n_2+1)!}{n!}
\sum_{r=0}^{n_1}
\alpha_{r,n_1}\,
\dot C^\alpha\bigl((n_1-2r)p\bigr)
\mathcal{L}_{n_1+1}^{\alpha,\alpha_0}
\Bigl(
p^{[r]},
(-p)^{[n_1-r]},
(n_1-2r)p
\Bigr)
\\
\times
\mathcal{L}_{n_2+1}^{\alpha,\alpha_0}
\Bigl(
p^{[n/2-r]},
(-p)^{[n/2-n_1+r]},
-(n_1-2r)p
\Bigr).
\end{multline}
Equivalently, the inner summation may be restricted to
\[
\max\left\{0,n_1-\frac n2\right\}
\le r\le
\min\left\{n_1,\frac n2\right\},
\]
since \(\alpha_{r,n_1}=0\) outside this range.
\end{proposition}
\begin{proof}
We start from the exact flow equation
\begin{multline}\label{EqGeneralFlowProof}
\partial_{\alpha}
\mathcal L_n^{\alpha,\alpha_0}(p_1,\ldots,p_n)
=
\binom{n+2}{2}
\int_k
\dot C^\alpha(k)\,
\mathcal L_{n+2}^{\alpha,\alpha_0}
(p_1,\ldots,p_n,k,-k)
\\
-\frac12
\sum_{\substack{n_1+n_2=n\\ n_1,n_2\ge1,~n_i\in2\mathbb{N}+1}}
\frac{(n_1+1)!(n_2+1)!}{n!}
\sum_{\substack{\pi\subset\{1,\ldots,n\}\\ |\pi|=n_1}}
\dot C^\alpha(p_\pi)
\mathcal L_{n_1+1}^{\alpha,\alpha_0}
\bigl((p_i)_{i\in\pi},-p_\pi\bigr)
\mathcal L_{n_2+1}^{\alpha,\alpha_0}
\bigl((p_i)_{i\in\pi^c},p_\pi\bigr).
\end{multline}
We evaluate this equation at the special momenta configuration
\[
(p_1,\ldots,p_n)=(p,-p)^{[n/2]}.
\]
The linear term immediately gives
\[
\binom{n+2}{2}
\int_k
\dot C^\alpha(k)\,
\mathcal L_{n+2}^{\alpha,\alpha_0}
\bigl((p,-p)^{[n/2]},k,-k\bigr).
\]
It remains to reorganise the quadratic term. Fix $
n_1+n_2=n$
and a subset $
\pi\subset\{1,\ldots,n\}$ with $
|\pi|=n_1$. Let \(r\) denote the number of entries equal to \(p\) contained in
\(\pi\). Since \(\pi\) contains \(n_1-r\) entries equal to \(-p\), one
has
\begin{equation}\label{EqMomentumPiProof}
p_\pi
=
rp-(n_1-r)p
=
(2r-n_1)p.
\end{equation}
The additional momentum entering the first correlator is therefore
\[
-p_\pi=(n_1-2r)p,
\]
so that momentum is conserved at the first vertex. Since the full external configuration satisfies
\[
\sum_{i=1}^n p_i=0,
\]
one also has
\[
p_{\pi^c}
=
\sum_{i\in\pi^c}p_i
=
-p_\pi
=
(n_1-2r)p.
\]
Thus the additional momentum entering the second correlator is
\[
-p_{\pi^c}
=
p_\pi
=
-(n_1-2r)p,
\]
which ensures momentum conservation at the second vertex as well. Moreover, the complement \(\pi^c\) contains $
\frac n2-r$ entries equal to \(p\) and
\[
\frac n2-(n_1-r)
=
\frac n2-n_1+r
\]
entries equal to \(-p\). By permutation symmetry of the correlators,
the contribution associated with any such subset \(\pi\) is therefore
equal to

\begin{multline}\label{EqClassContributionProof}
\dot C^\alpha\bigl((n_1-2r)p\bigr)
\mathcal L_{n_1+1}^{\alpha,\alpha_0}
\Bigl(
p^{[r]},
(-p)^{[n_1-r]},
(n_1-2r)p
\Bigr)
\\
\times
\mathcal L_{n_2+1}^{\alpha,\alpha_0}
\Bigl(
p^{[n/2-r]},
(-p)^{[n/2-n_1+r]},
-(n_1-2r)p
\Bigr).
\end{multline}
Here we have used that the covariance derivative is even w.r.t. $p$. It remains only to count the subsets \(\pi\) yielding a fixed value of
\(r\). There are \(n/2\) entries equal to \(p\) and \(n/2\) entries
equal to \(-p\). Hence the number of subsets of cardinality \(n_1\)
containing exactly \(r\) positive entries is
\[
\binom{n/2}{r}
\binom{n/2}{n_1-r}
=
\alpha_{r,n_1}.
\]
Grouping the subsets \(\pi\) according to \(r\) therefore transforms
the quadratic term in \eqref{EqGeneralFlowProof} into
\begin{multline*}
-\frac12
\sum_{\substack{n_1+n_2=n\\ n_1,n_2\ge1,~n_i\in2\mathbb{N}+1}}
\frac{(n_1+1)!(n_2+1)!}{n!}
\sum_{r=0}^{n_1}
\alpha_{r,n_1}\,
\dot C^\alpha\bigl((n_1-2r)p\bigr)
\mathcal L_{n_1+1}^{\alpha,\alpha_0}
\Bigl(
p^{[r]},
(-p)^{[n_1-r]},
(n_1-2r)p
\Bigr)
\\
\times
\mathcal L_{n_2+1}^{\alpha,\alpha_0}
\Bigl(
p^{[n/2-r]},
(-p)^{[n/2-n_1+r]},
-(n_1-2r)p
\Bigr),
\end{multline*}
which is precisely \eqref{FE+-}. Finally,
$
\alpha_{r,n_1}\neq0
$
only if
\[
0\le r\le\frac n2,\qquad
0\le n_1-r\le\frac n2.
\]
Equivalently,
\[
\max\left\{0,n_1-\frac n2\right\}
\le r\le
\min\left\{n_1,\frac n2\right\}
\]
and this proves the final assertion.
\end{proof}
\subsection{Restriction to the symmetric momentum manifold}

The Wilson--Polchinski hierarchy \eqref{flow-Ln} is exact, but it couples
infinitely many momentum-dependent correlators. Our aim is to extract from
this hierarchy a reduced system which retains the relevant momentum
dependence and is sufficiently explicit to admit a non-perturbative
analysis.

We consider first the symmetric momentum configurations
\[
(p,-p,\ldots,p,-p),
\]
with \(n\) even. This family contains the zero-momentum configuration and
is invariant under the exchange \(p\leftrightarrow -p\). It therefore
provides a natural setting in which to separate the constant and quadratic
momentum contributions from the remaining momentum dependence.

For every even integer \(n\), we introduce the family
\begin{equation}\label{DefDecomposition}
D_n^{\alpha,\alpha_0}(p)
:=
A_n^{\alpha,\alpha_0}
+
np^2 B_n^{\alpha,\alpha_0}
+
R_n^{\alpha,\alpha_0}(p).
\end{equation}
The first two terms define the second-order mean-field approximation:
\(A_n^{\alpha,\alpha_0}\) is its momentum-independent coefficient,
whereas \(B_n^{\alpha,\alpha_0}\) determines its quadratic momentum
dependence. The function \(R_n^{\alpha,\alpha_0}\) represents the exact
error of this approximation. In particular, it is not assumed to have
a vanishing constant or quadratic jet at the origin.
The families
\[
\bigl(A_n^{\alpha,\alpha_0}\bigr)_{n\ge2}
\qquad\text{and}\qquad
\bigl(B_n^{\alpha,\alpha_0}\bigr)_{n\ge2}
\]
will be chosen to satisfy an autonomous nonlinear hierarchy, while
\(R_n^{\alpha,\alpha_0}\) will be defined through an exact evolution
equation measuring the discrepancy between this closure and the
restricted Wilson--Polchinski flow.

The three families on the right-hand side will be defined through
a system of flow equations chosen so that their sum satisfies the
Wilson--Polchinski hierarchy restricted to the symmetric momentum
manifold. 

A direct restriction of the flow equations to
\[
(p,-p,\ldots,p,-p)
\]
does not yield a closed system. There are two difficulties: first, the linear term in the flow equation for the \(n\)-point correlator
contains the \((n+2)\)-point correlator in which the arguments $(k,-k)$ are integrated over \(\mathbb R^4\). Thus the linear part of the
flow immediately leaves the one-parameter family of symmetric
configurations.

Second, the quadratic term involves a sum over partitions of the external
momenta. If a block \(\pi\) contains \(r\) entries equal to \(p\) and
\(n_1-r\) entries equal to \(-p\), then momentum conservation gives
\[
p_\pi=(2r-n_1)p.
\]
The internal momentum carried by the complementary correlator is therefore
\[
-p_\pi=(n_1-2r)p.
\]
Except in special cases, the two correlators occurring in the quadratic
term are consequently evaluated at momentum configurations which are no
longer of the form
\[
(p,-p,\ldots,p,-p).
\]
The enlargement of the ansatz below is designed precisely to accommodate
these two effects. For every even integer \(n\ge2\) and every momentum configuration
\[
\mathbf p=(p_1,\ldots,p_n)\in(\mathbb R^4)^n,
\qquad
\sum_{i=1}^n p_i=0,
\]
we define
\begin{equation}\label{DefGeneralTn}
T_n^{\alpha,\alpha_0}(\mathbf p)
:=
A_n^{\alpha,\alpha_0}
+
B_n^{\alpha,\alpha_0}
\sum_{i=1}^n |p_i|^2
+
R_n^{\alpha,\alpha_0}(\mathbf p).
\end{equation}
For the symmetric momentum configuration $
\mathbf p=(p,-p)^{[n/2]}$
one has
\[
T_n^{\alpha,\alpha_0}\bigl((p,-p)^{[n/2]}\bigr)
=
A_n^{\alpha,\alpha_0}
+
np^2B_n^{\alpha,\alpha_0}
+
R_n^{\alpha,\alpha_0}\bigl((p,-p)^{[n/2]}\bigr),
\]
so that \eqref{DefGeneralTn} reduces to
\eqref{DefDecomposition} on the symmetric momentum manifold.
The families
\[
A_n^{\alpha,\alpha_0},
\qquad
B_n^{\alpha,\alpha_0},
\qquad
R_n^{\alpha,\alpha_0}
\]
will now be defined through their respective flow equations. For later convenience, we introduce the notation
\[
\int_k
:=
\int_{\mathbb R^4}
\frac{d^4k}{(2\pi)^4},
\]
together with
\begin{align}
\mathfrak{a}_m^\alpha(1)
&:=
\int_k
\dot C^\alpha(k;m),
&
\mathfrak{a}_m^\alpha(2)
&:=
2\int_k
k^2
\dot C^\alpha(k;m),\label{12}
\\
\mathfrak{b}_m^\alpha(0)
&:=
\dot C^\alpha(0;m),
&
\mathfrak{b}_m^\alpha(2)
&:=
\partial^2
\dot C^\alpha(0;m).\label{13}
\end{align}
The coefficients
\(
A_n^{\alpha,\alpha_0}
\)
are defined by the flow equations
\begin{multline}\label{flow-An}
\partial_\alpha
A_n^{\alpha,\alpha_0}
=
\binom{n+2}{2}
\Big(
A_{n+2}^{\alpha,\alpha_0}
\mathfrak{a}_m^\alpha(1)
+
B_{n+2}^{\alpha,\alpha_0}
\mathfrak{a}_m^\alpha(2)
\Big)\\
-
\frac12
\sum_{n_1+n_2=n}
(n_1+1)(n_2+1)
A_{n_1+1}^{\alpha,\alpha_0}
A_{n_2+1}^{\alpha,\alpha_0}
\mathfrak{b}_m^\alpha(0).
\end{multline}
The coefficients
\(
B_n^{\alpha,\alpha_0}
\)
are defined by
\begin{align}\label{flow-Bn}
\partial_\alpha
B_n^{\alpha,\alpha_0}
&=
\binom{n+2}{2}
B_{n+2}^{\alpha,\alpha_0}
\mathfrak{a}_m^\alpha(1)
-
\frac{1}{2n}
\sum_{n_1+n_2=n}
(n_1+1)(n_2+1)
\Bigg[
n_1
\Big(
1+\frac{n_2}{n-1}
\Big)
B_{n_1+1}^{\alpha,\alpha_0}
A_{n_2+1}^{\alpha,\alpha_0}
\nonumber\\
&\hspace{3.8cm}
+
n_2
\Big(
1+\frac{n_1}{n-1}
\Big)
A_{n_1+1}^{\alpha,\alpha_0}
B_{n_2+1}^{\alpha,\alpha_0}
\Bigg]
\mathfrak{b}_m^\alpha(0)
\nonumber\\
&\quad
-
\frac{1}{2n}
\sum_{n_1+n_2=n}
\frac{
n_1(n_1+1)n_2(n_2+1)
}{n-1}
A_{n_1+1}^{\alpha,\alpha_0}
A_{n_2+1}^{\alpha,\alpha_0}
\mathfrak{b}_m^\alpha(2).
\end{align}
The flow equations \eqref{flow-An} and \eqref{flow-Bn} form a closed
nonlinear subsystem governing the relevant momentum sector.

It remains to
construct the remainder. Contrary to the coefficients
\(
A_n^{\alpha,\alpha_0}
\)
and
\(
B_n^{\alpha,\alpha_0},
\)
whose definition does not depend on the external momenta, the linear term of the
Wilson--Polchinski flow equation couples
\(
\mathcal{L}_n^{\alpha,\alpha_0}
\)
to
\(
\mathcal{L}_{n+2}^{\alpha,\alpha_0}(\cdot,k,-k),
\)
where the additional momenta
\(
k,-k
\)
are integrated over all of
\(
\mathbb R^4
\).
For this reason, the remainder must be defined as a function of arbitrary
external momenta rather than only on the symmetric momentum manifold. We introduce the following definition that we need in stating the flow equations of the remainder:
\begin{definition}
For every sufficiently smooth function \(f:\mathbb R^4\to\mathbb R\), we define the first- and second-order Taylor remainders at the origin by
\begin{align}
\mathcal R_{(1)}(f)(q)
&:=
f(q)-f(0),
\\
\mathcal R_{(2)}(f)(q)
&:=
f(q)-f(0)-q^2\partial^2f(0),
\end{align}
where
\[
\partial^2f(0)
:=
\frac{1}{8}
\sum_{\mu=1}^{4}
\frac{\partial^2f}{\partial q_\mu^2}(0).
\]
\end{definition}
\noindent The remainder
\(
R_n^{\alpha,\alpha_0}
\)
is defined by the following exact flow equation.
For every \(n\ge2\),
\begin{multline}\label{flow-Rn}
\partial_\alpha
R_n^{\alpha,\alpha_0}
(p_1,\dots,p_n)
=
\binom{n+2}{2}
\int_k
\dot C^\alpha(k;m)\,
R_{n+2}^{\alpha,\alpha_0}
(p_1,\dots,p_n,k,-k)
\\
-
\frac12
\sum_{\substack{
n_1+n_2=n\\
(\pi_1,\pi_2)
}}
\gamma_{n_1,n_2}\,
A_{n_1+1}^{\alpha,\alpha_0}
A_{n_2+1}^{\alpha,\alpha_0}
\,
\mathcal R_{(2)}
(\dot C^\alpha)
(p_\pi;m)
\\
-
\frac12
\sum_{\substack{
n_1+n_2=n\\
(\pi_1,\pi_2)
}}
\gamma_{n_1,n_2}
\Bigl(
\|p\|_{\pi_1}^2
B_{n_1+1}^{\alpha,\alpha_0}
A_{n_2+1}^{\alpha,\alpha_0}
+
\|p\|_{\pi_2}^2
A_{n_1+1}^{\alpha,\alpha_0}
B_{n_2+1}^{\alpha,\alpha_0}
\Bigr)
\mathcal R_{(1)}
(\dot C^\alpha)
(p_\pi;m)
\\
-
\frac12
\sum_{\substack{
n_1+n_2=n\\
(\pi_1,\pi_2)
}}
\gamma_{n_1,n_2}
\Bigl(
\|p\|_{\pi_1}^2
\|p\|_{\pi_2}^2
B_{n_1+1}^{\alpha,\alpha_0}
B_{n_2+1}^{\alpha,\alpha_0}
\\
\qquad\qquad
+
A_{n_1+1}^{\alpha,\alpha_0}
R_{n_2+1}^{\alpha,\alpha_0}
(-p_\pi,\vec p_{\pi_2})
+
R_{n_1+1}^{\alpha,\alpha_0}
(\vec p_{\pi_1},p_\pi)
A_{n_2+1}^{\alpha,\alpha_0}
\Bigr)
\dot C^\alpha(p_\pi;m)
\\
-
\frac12
\sum_{\substack{
n_1+n_2=n\\
(\pi_1,\pi_2)
}}
\gamma_{n_1,n_2}
\Bigl(
\|p\|_{\pi_1}^2
B_{n_1+1}^{\alpha,\alpha_0}
R_{n_2+1}^{\alpha,\alpha_0}
(-p_\pi,\vec p_{\pi_2})
\\
\qquad\qquad
+
R_{n_1+1}^{\alpha,\alpha_0}
(\vec p_{\pi_1},p_\pi)
\|p\|_{\pi_2}^2
B_{n_2+1}^{\alpha,\alpha_0}
+
R_{n_1+1}^{\alpha,\alpha_0}
(\vec p_{\pi_1},p_\pi)
R_{n_2+1}^{\alpha,\alpha_0}
(-p_\pi,\vec p_{\pi_2})
\Bigr)
\dot C^\alpha(p_\pi;m).
\end{multline}
Here 
$$\gamma_{n_1,n_2}:=\frac{(n_1+1)!(n_2+1)!}{n!}$$
and we use the same convention for partitions as in
\eqref{flow-Ln}: the sum runs over ordered partitions
\(
(\pi_1,\pi_2)
\)
of \(\{1,\ldots,n\}\), with
$
|\pi_i|=n_i$ and $
n_1+n_2=n$. The factor \(1/2\) compensates for the simultaneous exchange
\[
(n_1,\pi_1)\longleftrightarrow(n_2,\pi_2).
\]
We use the shorthand notation
\[
\|p\|_{\pi_i}^2
:=
\sum_{j\in\pi_i}|p_j|^2
+
\left|
\sum_{j\in\pi_i}p_j
\right|^2.
\]
\begin{remark}[Structure of the error equation]
\label{RemRemainderStructure}
The error \(R_n^{\alpha,\alpha_0}\) cannot in general be regarded as
irrelevant in the renormalization-group sense. Although the Taylor
remainders appearing in the source terms provide additional powers of
the external momenta, this gain is not preserved by the linear term
\[
\binom{n+2}{2}
\int_k
\dot C^\alpha(k;m)\,
R_{n+2}^{\alpha,\alpha_0}
(p_1,\ldots,p_n,k,-k).
\]
Indeed, powers of the internal momentum are converted by Gaussian
integration into negative powers of the flow parameter; in dimension
four, for example,
\[
\int_k e^{-\alpha k^2}\asymp\alpha^{-2},
\qquad
\int_k k^2e^{-\alpha k^2}\asymp\alpha^{-3}.
\]
Consequently, the flow need not preserve either
\(R_n^{\alpha,\alpha_0}(0)=0\) or the vanishing of its quadratic
momentum component.

Nevertheless, \eqref{flow-Rn} retains a useful algebraic structure.
The terms independent of \(R\) are generated by products of the
families \(A\) and \(B\), of the types \(AA\), \(AB\), and \(BB\),
whereas the remaining nonlinear terms involve \(AR\), \(BR\), and
\(RR\). This suggests seeking \(R\) in a suitable completion of the
algebra generated by \(A\) and \(B\), for instance in the form
\[
R_n^{\alpha,\alpha_0}(p)
=
\sum_{\mathfrak m\in\mathfrak M_n}
c_{n,\mathfrak m}^{\alpha,\alpha_0}(p)\,
\mathfrak m
\bigl(
A^{\alpha,\alpha_0},
B^{\alpha,\alpha_0}
\bigr),
\]
where every monomial \(\mathfrak m\) has strictly positive degree. If
the coefficients could be constructed through a closed inductive
scheme and shown to satisfy uniform summability estimates independent
of \(\alpha_0\), then the ultraviolet triviality of \(A\) and \(B\)
would imply that of \(R\). 
\end{remark}
The following result establishes that  the family
\(
T_n^{\alpha,\alpha_0}
\)
is a solution of the flow equations restricted to the special momentum sector $\left(p,-p\right)^{[n/2]}$:
\begin{theorem}[Reconstruction of a solution to the restricted hierarchy]
\label{thm:reconstruction}
Let
\[
\bigl(
A_n^{\alpha,\alpha_0},
B_n^{\alpha,\alpha_0},
R_n^{\alpha,\alpha_0}
\bigr)_{n\ge2}
\]
satisfy the flow equations
\eqref{flow-An}--\eqref{flow-Rn}, and define
\(T_n^{\alpha,\alpha_0}\) by \eqref{DefGeneralTn}. Then, for every even
integer \(n\ge2\) and every \(p\in\mathbb R^4\), the family
\[
T_n^{\alpha,\alpha_0}
\bigl((p,-p)^{[n/2]}\bigr)
\]
satisfies the Wilson--Polchinski flow equation restricted to the
symmetric momentum configuration, namely \eqref{FE+-}.
\end{theorem}
\begin{proof}
For notational simplicity, we suppress the superscripts
\(\alpha,\alpha_0\) of the unknown functions throughout the proof and
write
\[
\dot C(q):=\dot C^\alpha(q;m).
\]
Recall the kernel coefficients
\begin{align}
\mathfrak a_m^\alpha(1)
&=
\int_k\dot C^\alpha(k;m),
&
\mathfrak a_m^\alpha(2)
&=
2\int_k k^2\dot C^\alpha(k;m),
\label{EqKernelMomentsProof}
\\
\mathfrak b_m^\alpha(0)
&=
\dot C^\alpha(0;m),
&
\mathfrak b_m^\alpha(2)
&=
\partial^2\dot C^\alpha(0;m).
\label{EqKernelTaylorProof}
\end{align}
The Taylor remainders of the covariance derivative therefore satisfy
\begin{align}
\dot C(q)
&=
\mathfrak b_m^\alpha(0)
+
\mathcal R_{(1)}(\dot C)(q),
\label{EqTaylorOneProof}
\\
\dot C(q)
&=
\mathfrak b_m^\alpha(0)
+
q^2\mathfrak b_m^\alpha(2)
+
\mathcal R_{(2)}(\dot C)(q).
\label{EqTaylorTwoProof}
\end{align}
On the symmetric momentum configuration, the definition of \(T_n\)
gives
\begin{equation}\label{EqTSymmetricProof}
T_n\bigl((p,-p)^{[n/2]}\bigr)
=
A_n+np^2B_n
+
R_n\bigl((p,-p)^{[n/2]}\bigr).
\end{equation}
We differentiate \eqref{EqTSymmetricProof} with respect to \(\alpha\)
and insert the flow equations
\eqref{flow-An}--\eqref{flow-Rn}.

\medskip

\noindent\emph{Linear term.}
By \eqref{DefGeneralTn},
\begin{equation}\label{EqLinearExpansionT}
T_{n+2}
\bigl((p,-p)^{[n/2]},k,-k\bigr)
=
A_{n+2}
+
\bigl(np^2+2k^2\bigr)B_{n+2}
+
R_{n+2}
\bigl((p,-p)^{[n/2]},k,-k\bigr).
\end{equation}
Consequently,
\begin{align}
\int_k\dot C(k)\,
T_{n+2}
\bigl((p,-p)^{[n/2]},k,-k\bigr)
&=
\mathfrak a_m^\alpha(1)A_{n+2}
+
\mathfrak a_m^\alpha(2)B_{n+2}
\nonumber\\
&\quad
+
np^2\mathfrak a_m^\alpha(1)B_{n+2}
\nonumber\\
&\quad
+
\int_k\dot C(k)\,
R_{n+2}
\bigl((p,-p)^{[n/2]},k,-k\bigr).
\label{EqLinearRecombinationProof}
\end{align}
The first two terms on the right-hand side are precisely the linear
contributions to \(\partial_\alpha A_n\), the third is the linear
contribution to \(np^2\partial_\alpha B_n\), and the last is the
linear contribution to \(\partial_\alpha R_n\). After multiplication
by \(\binom{n+2}{2}\), they recombine into
\[
\binom{n+2}{2}
\int_k\dot C(k)\,
T_{n+2}
\bigl((p,-p)^{[n/2]},k,-k\bigr).
\]

\medskip

\noindent\emph{Nonlinear term.}
Fix odd integers \(n_1,n_2\ge1\) such that \(n_1+n_2=n\), and consider
an ordered partition \((\pi_1,\pi_2)\) of the \(n\) external momenta
with
\[
|\pi_i|=n_i.
\]
Suppose that \(\pi_1\) contains \(r\) entries equal to \(p\). It then
contains \(n_1-r\) entries equal to \(-p\), and hence
\[
\sum_{q\in\pi_1}q
=
rp-(n_1-r)p
=
(2r-n_1)p.
\]
The momentum inserted into the first factor is therefore
\begin{equation}\label{EqPartitionMomentumTProof}
p_\pi
:=
-\sum_{q\in\pi_1}q
=
(n_1-2r)p.
\end{equation}
Set $
s_r:=n_1-2r$
so that \(p_\pi=s_rp\), while the momentum inserted into the second
factor is \(-s_rp\).

The number of choices of \(\pi_1\) containing exactly \(r\) entries
equal to \(p\) is
\begin{equation}\label{EqAlphaMultiplicityProof}
\alpha_{r,n_1}
=
\binom{n/2}{r}
\binom{n/2}{n_1-r}.
\end{equation}
As usual, the binomial coefficient is understood to vanish when its
lower index lies outside its natural range.

The quadratic momentum weights of the two factors are
\begin{align}
Q_1(r)
&:=
\|p\|_{\pi_1}^2
=
\bigl(n_1+s_r^2\bigr)p^2,
\label{EqQOne}
\\
Q_2(r)
&:=
\|p\|_{\pi_2}^2
=
\bigl(n_2+s_r^2\bigr)p^2.
\label{EqQTwo}
\end{align}
Indeed, each block contains \(n_i\) external momenta of squared norm
\(p^2\), together with one additional momentum of squared norm
\(s_r^2p^2\).

It follows from \eqref{DefGeneralTn} that
\begin{align}
T_{n_1+1}
\Bigl(
p^{[r]},
(-p)^{[n_1-r]},
s_rp
\Bigr)
&=
A_{n_1+1}
+
Q_1(r)B_{n_1+1}
+
R_{n_1+1}^{(r)},
\label{EqFirstFactorT}
\\
T_{n_2+1}
\Bigl(
p^{[n/2-r]},
(-p)^{[n/2-n_1+r]},
-s_rp
\Bigr)
&=
A_{n_2+1}
+
Q_2(r)B_{n_2+1}
+
R_{n_2+1}^{(r)},
\label{EqSecondFactorT}
\end{align}
where
\begin{align*}
R_{n_1+1}^{(r)}
&:=
R_{n_1+1}
\Bigl(
p^{[r]},
(-p)^{[n_1-r]},
s_rp
\Bigr),
\\
R_{n_2+1}^{(r)}
&:=
R_{n_2+1}
\Bigl(
p^{[n/2-r]},
(-p)^{[n/2-n_1+r]},
-s_rp
\Bigr).
\end{align*}
We shall use the combinatorial identities
\begin{align}
\sum_{r=0}^{n_1}\alpha_{r,n_1}
&=
\binom{n}{n_1},
\label{EqVandermondeZero}
\\
\sum_{r=0}^{n_1}\alpha_{r,n_1}s_r^2
&=
\binom{n}{n_1}
\frac{n_1n_2}{n-1},
\label{EqVandermondeSecond}
\\
\sum_{r=0}^{n_1}\alpha_{r,n_1}Q_1(r)
&=
p^2\binom{n}{n_1}
n_1
\left(
1+\frac{n_2}{n-1}
\right),
\label{EqVandermondeQOne}
\\
\sum_{r=0}^{n_1}\alpha_{r,n_1}Q_2(r)
&=
p^2\binom{n}{n_1}
n_2
\left(
1+\frac{n_1}{n-1}
\right).
\label{EqVandermondeQTwo}
\end{align}
The first is Vandermonde's identity. For the second, expand
\(s_r^2=(n_1-2r)^2\) and use
\begin{align}
\sum_{r=0}^{n_1}
r\,\alpha_{r,n_1}
&=
\frac{n_1}{2}
\binom{n}{n_1},
\label{EqFirstWeightedVandermonde}
\\
\sum_{r=0}^{n_1}
r(r-1)\alpha_{r,n_1}
&=
\frac{n_1(n_1-1)}{4}
\frac{n-2}{n-1}
\binom{n}{n_1}.
\label{EqSecondWeightedVandermonde}
\end{align}
The last two identities then follow from
\eqref{EqQOne}--\eqref{EqQTwo}. Recall that
\[
\gamma_{n_1,n_2}
:=
\frac{(n_1+1)!(n_2+1)!}{n!}.
\]
By \eqref{EqVandermondeZero}, we have
\begin{equation}\label{EqGammaVandermonde}
\gamma_{n_1,n_2}
\sum_{r=0}^{n_1}\alpha_{r,n_1}
=
(n_1+1)(n_2+1).
\end{equation}
We now expand the product in
\eqref{EqFirstFactorT}--\eqref{EqSecondFactorT}. For the \(AA\)-term,
we use
\begin{equation}\label{EqAATaylorProof}
\dot C(s_rp)
=
\mathfrak b_m^\alpha(0)
+
s_r^2p^2\mathfrak b_m^\alpha(2)
+
\mathcal R_{(2)}(\dot C)(s_rp).
\end{equation}
By \eqref{EqGammaVandermonde}, the first contribution is the nonlinear
\(AA\)-term in \(\partial_\alpha A_n\). By
\eqref{EqVandermondeSecond}, the second is the corresponding
contribution to \(np^2\partial_\alpha B_n\), while the last is the
\(AA\)-source term in \(\partial_\alpha R_n\). These three
contributions therefore recombine into
\[
-\frac12
\gamma_{n_1,n_2}
\sum_{r=0}^{n_1}
\alpha_{r,n_1}
A_{n_1+1}A_{n_2+1}
\dot C(s_rp).
\]
For the \(BA\)- and \(AB\)-terms, we use
\[
\dot C(s_rp)
=
\mathfrak b_m^\alpha(0)
+
\mathcal R_{(1)}(\dot C)(s_rp).
\]
By \eqref{EqVandermondeQOne}--\eqref{EqVandermondeQTwo}, the
contributions proportional to \(\mathfrak b_m^\alpha(0)\) are exactly
the corresponding terms in \(np^2\partial_\alpha B_n\). The remaining
terms are the \(BA\)- and \(AB\)-source terms in
\(\partial_\alpha R_n\). Hence they recombine into
\begin{equation}\label{EqABRecombinationProof}
-\frac12
\gamma_{n_1,n_2}
\sum_{r=0}^{n_1}
\alpha_{r,n_1}
\Bigl(
Q_1(r)B_{n_1+1}A_{n_2+1}
+
Q_2(r)A_{n_1+1}B_{n_2+1}
\Bigr)
\dot C(s_rp).
\end{equation}
The use of \(\mathcal R_{(1)}\), rather than
\(\mathcal R_{(2)}\), is essential here, since the factors \(Q_i(r)\)
are already quadratic in the external momentum.

Finally, the terms of types \(BB\), \(AR\), \(BR\), and \(RR\) occur
entirely in \(\partial_\alpha R_n\), multiplied by the full covariance
\(\dot C(s_rp)\). Combining all contributions gives
\begin{multline}\label{EqFullQuadraticRecombinationProof}
-\frac12
\sum_{\substack{
n_1+n_2=n\\
n_1,n_2\ge1\ {\rm odd}
}}
\gamma_{n_1,n_2}
\sum_{r=0}^{n_1}
\alpha_{r,n_1}\,
\dot C(s_rp)
\\
\times
\Bigl[
A_{n_1+1}
+
Q_1(r)B_{n_1+1}
+
R_{n_1+1}^{(r)}
\Bigr]
\Bigl[
A_{n_2+1}
+
Q_2(r)B_{n_2+1}
+
R_{n_2+1}^{(r)}
\Bigr].
\end{multline}
By \eqref{EqFirstFactorT} and \eqref{EqSecondFactorT}, this is equal to
\begin{multline}\label{EqQuadraticTProof}
-\frac12
\sum_{\substack{
n_1+n_2=n\\
n_1,n_2\ge1\ {\rm odd}
}}
\frac{(n_1+1)!(n_2+1)!}{n!}
\sum_{r=0}^{n_1}
\alpha_{r,n_1}\,
\dot C\bigl((n_1-2r)p\bigr)
\\
\times
T_{n_1+1}
\Bigl(
p^{[r]},
(-p)^{[n_1-r]},
(n_1-2r)p
\Bigr)
\\
\times
T_{n_2+1}
\Bigl(
p^{[n/2-r]},
(-p)^{[n/2-n_1+r]},
-(n_1-2r)p
\Bigr).
\end{multline}
Combining \eqref{EqQuadraticTProof} with the linear contribution shows
that
\[
T_n^{\alpha,\alpha_0}
\bigl((p,-p)^{[n/2]}\bigr)
\]
satisfies the restricted Wilson--Polchinski flow equation
\eqref{FE+-}. This concludes the proof.
\end{proof}
\section{Non-perturbative analysis of the second-order mean-field hierarchy}
\label{sec3}

In this section, we study the second-order mean-field hierarchy formed
by the families
\[
\bigl(
A_n^{\alpha,\alpha_0},
B_n^{\alpha,\alpha_0}
\bigr)_{n\ge2}.
\]
Our main objective is to construct non-perturbative solutions to the
coupled flow equations \eqref{flow-An} and \eqref{flow-Bn} and to
analyse their ultraviolet behaviour. More precisely, we prove that,
for every even integer \(n\ge2\),
\begin{equation}\label{EqRelevantTriviality}
\lim_{\alpha_0\downarrow0}
\bigl(
A_n^{1,\alpha_0},
B_n^{1,\alpha_0}
\bigr)
=
(0,0).
\end{equation}
Thus, the constant and quadratic families retained by the
second-order mean-field closure are asymptotically trivial when the
ultraviolet cutoff is removed. No corresponding assertion concerning
the error \(R_n^{\alpha,\alpha_0}\), and hence no triviality statement
for the full restricted correlators, is made in this paper.

The argument proceeds in several steps. We first integrate the coupled
flow equations and reduce them to a system of equations for
dimensionless sequences. We then introduce the algebraic framework in
which this system is naturally posed and establish the mapping
properties of its linear and bilinear operators. Finally, these
estimates are combined to construct solutions and prove
\eqref{EqRelevantTriviality}.

Throughout this section, we work in the massless case \(m=0\). This
choice simplifies the flow coefficients and facilitates comparison
with the mean-field analysis of
\cite{Kopper2022,KopperWang2025}. The massive case is not considered
here.

The mean-field system studied by Kopper in \cite{Kopper2022} is
obtained by replacing the momentum-dependent Wilson--Polchinski
hierarchy with a momentum-independent ansatz. The connected amputated
Schwinger functions are thereby represented by scalar quantities
\(A_n^{\alpha,\alpha_0}\), leading to a closed
infinite-dimensional dynamical system.

The hierarchy considered here extends this approximation by retaining
a quadratic momentum contribution. More precisely, the
second-order mean-field ansatz takes the form
\[
T_{n,\mathrm{MF2}}^{\alpha,\alpha_0}(p_1,\ldots,p_n)
=
A_n^{\alpha,\alpha_0}
+
\left(
\sum_{i=1}^{n}p_i^2
\right)
B_n^{\alpha,\alpha_0}.
\]
The family \(A_n^{\alpha,\alpha_0}\) describes the
momentum-independent part of the closure, while
\(B_n^{\alpha,\alpha_0}\) determines its quadratic momentum
dependence. These quantities are fixed by the autonomous system of flow equations
\eqref{flow-An}--\eqref{flow-Bn}. Note that they are not asserted to coincide
with the constant and quadratic Taylor coefficients of the exact
correlators.

The resulting system may therefore be viewed as a non-perturbative
second-order mean-field approximation to the Wilson--Polchinski
hierarchy. It incorporates the first nontrivial momentum dependence
and, in its two-point sector, the parameter associated with
wave-function renormalization. The main result of this section extends
the mean-field construction of \cite{Kopper2022,KopperWang2025} to
this enlarged closure by proving the existence and ultraviolet
triviality of both families \(A_n\) and \(B_n\). The exact discrepancy
between this approximation and the restricted Wilson--Polchinski
hierarchy remains encoded by the separately defined error
\(R_n^{\alpha,\alpha_0}\).

Since \(m=0\), we consider the flow on the finite interval
\[
\alpha\in[\alpha_0,1],
\]
whose upper endpoint provides a fixed infrared cutoff. Removal of the
ultraviolet cutoff corresponds to the limit
\[
\alpha_0\downarrow0.
\]

In the massless case, all coefficients appearing in the second-order
mean-field equations are explicit. More precisely, the quantities
defined in \eqref{12}--\eqref{13} are given by
\begin{equation}\label{EqExplicitFlowCoefficients}
\mathfrak a_0^\alpha(1)
=
\frac{c}{\alpha^2},
\qquad
\mathfrak a_0^\alpha(2)
=
\frac{4c}{\alpha^3},
\qquad
\mathfrak b_0^\alpha(0)
=
1,
\qquad
\mathfrak b_0^\alpha(2)
=
-\alpha,
\end{equation}
where $
c:=\frac{1}{16\pi^2}$. Consequently, the flow equations for
\(\bigl(A_n,B_n\bigr)\) form a closed system with explicit
coefficients:
\begin{multline}\label{EqFlowAnExplicit}
\partial_\alpha A_n^{\alpha,\alpha_0}
=
\binom{n+2}{2}
\frac{c}{\alpha^2}
A_{n+2}^{\alpha,\alpha_0}
+
\binom{n+2}{2}
\frac{4c}{\alpha^3}
B_{n+2}^{\alpha,\alpha_0}
\\
-
\frac12
\sum_{\substack{
n_1+n_2=n\\
n_1,n_2\ge1\ {\rm odd}
}}
(n_1+1)(n_2+1)
A_{n_1+1}^{\alpha,\alpha_0}
A_{n_2+1}^{\alpha,\alpha_0},
\end{multline}
whereas
\begin{multline}\label{EqFlowBnExplicit}
\partial_\alpha B_n^{\alpha,\alpha_0}
=
\binom{n+2}{2}
\frac{c}{\alpha^2}
B_{n+2}^{\alpha,\alpha_0}
-
\frac{1}{2n}
\sum_{\substack{
n_1+n_2=n\\
n_1,n_2\ge1\ {\rm odd}
}}
(n_1+1)(n_2+1)
\Biggl[
n_1
\left(
1+\frac{n_2}{n-1}
\right)
B_{n_1+1}^{\alpha,\alpha_0}
A_{n_2+1}^{\alpha,\alpha_0}
\\
\hspace{4cm}
+
n_2
\left(
1+\frac{n_1}{n-1}
\right)
A_{n_1+1}^{\alpha,\alpha_0}
B_{n_2+1}^{\alpha,\alpha_0}
\Biggr]
\\
+
\frac{\alpha}{2n(n-1)}
\sum_{\substack{
n_1+n_2=n\\
n_1,n_2\ge1\ {\rm odd}
}}
n_1(n_1+1)n_2(n_2+1)
A_{n_1+1}^{\alpha,\alpha_0}
A_{n_2+1}^{\alpha,\alpha_0}.
\end{multline}

We next factor out the canonical scaling in \(\alpha\), together with
the combinatorial and \(c\)-dependent factors, by setting
\begin{equation}\label{EqRescalingAB}
A_n^{\alpha,\alpha_0}
=
\alpha^{\frac n2-2}
\frac{c^{\frac{2-n}{2}}}{n}
f_n(\alpha),
\qquad
B_n^{\alpha,\alpha_0}
=
\alpha^{\frac n2-1}
\frac{c^{\frac{2-n}{2}}}{n}
h_n(\alpha).
\end{equation}
The dependence of \(f_n\) and \(h_n\) on \(\alpha_0\) is suppressed
from the notation.

The notation \(f_n\) agrees with that of \cite{Kopper2022}. Formally
setting \(h_n=0\) in the \(f_n\)-equations recovers the mean-field
sequence equation studied there. The variables \(h_n\) describe the
additional quadratic momentum contribution retained by the present
closure. We stress, however, that \(h_n=0\) does not define an
invariant sector of the coupled system, since the equations for \(h_n\)
contain terms depending only on \(f_n\). Substituting \eqref{EqRescalingAB} into
\eqref{EqFlowAnExplicit} gives
\begin{multline}\label{EqSequenceFn}
f_{n+2}(\alpha)
+
4h_{n+2}(\alpha)
=
\frac{n-4}{n(n+1)}
f_n(\alpha)
+
\frac{2}{n(n+1)}
\alpha\partial_\alpha f_n(\alpha)
\\
+
\frac{1}{n+1}
\sum_{\substack{
n_1+n_2=n\\
n_1,n_2\ge1\ {\rm odd}
}}
f_{n_1+1}(\alpha)
f_{n_2+1}(\alpha).
\end{multline}
Similarly, \eqref{EqFlowBnExplicit} becomes
\begin{multline}\label{EqSequenceHn}
h_{n+2}(\alpha)
=
\frac{n-2}{n(n+1)}
h_n(\alpha)
+
\frac{2}{n(n+1)}
\alpha\partial_\alpha h_n(\alpha)
\\
+
\frac{1}{n(n+1)}
\sum_{\substack{
n_1+n_2=n\\
n_1,n_2\ge1\ {\rm odd}
}}
\Biggl[
n_1
\left(
1+\frac{n_2}{n-1}
\right)
h_{n_1+1}(\alpha)
f_{n_2+1}(\alpha)
\\
\hspace{4cm}
+
n_2
\left(
1+\frac{n_1}{n-1}
\right)
f_{n_1+1}(\alpha)
h_{n_2+1}(\alpha)
\Biggr]
\\
-
\frac{1}{n(n-1)(n+1)}
\sum_{\substack{
n_1+n_2=n\\
n_1,n_2\ge1\ {\rm odd}
}}
n_1n_2
f_{n_1+1}(\alpha)
f_{n_2+1}(\alpha).
\end{multline}
To recover the form used in \cite{Kopper2022}, we introduce the
logarithmic flow parameter
\begin{equation}\label{EqLogarithmicScale}
\mu
:=
\log\left(\frac{\alpha}{\alpha_0}\right).
\end{equation}
Then
\[
\partial_\mu=\alpha\partial_\alpha,
\qquad
0\le\mu\le\mu_{\max},
\qquad
\mu_{\max}:=\log\left(\frac1{\alpha_0}\right).
\]
With a slight abuse of notation, we set
\[
f_n(\mu)
:=
f_n(\alpha_0e^\mu),
\qquad
h_n(\mu)
:=
h_n(\alpha_0e^\mu),
\]
and continue to suppress their dependence on \(\alpha_0\). For \(n=2\), \eqref{EqSequenceFn} gives
\begin{equation}\label{EqSequenceF2Mu}
f_4+4h_4
=
\frac13f_2(f_2-1)
+
\frac13\partial_\mu f_2,
\end{equation}
while \eqref{EqSequenceHn} gives
\begin{equation}\label{EqSequenceH2Mu}
h_4
=
\frac16f_2(4h_2-f_2)
+
\frac13\partial_\mu h_2.
\end{equation}
For every even integer \(n\ge4\), separating the terms corresponding
to \(n_1=1\) or \(n_2=1\) in \eqref{EqSequenceFn} yields
\begin{multline}\label{EqSequenceFnMu}
f_{n+2}
+
4h_{n+2}
=
\frac{1}{n+1}
\sum_{\substack{
n_1+n_2=n\\
n_1,n_2\ge3\ {\rm odd}
}}
f_{n_1+1}f_{n_2+1}
+
\frac{1}{n+1}
\left(
2f_2+1-\frac4n
\right)
f_n
+
\frac{2}{n(n+1)}
\partial_\mu f_n.
\end{multline}
Similarly,
\begin{multline}\label{EqSequenceHnMu}
h_{n+2}
=
\frac{1}{n(n+1)}
\sum_{\substack{
n_1+n_2=n\\
n_1,n_2\ge3\ {\rm odd}
}}
\Biggl[
n_1
\left(
1+\frac{n_2}{n-1}
\right)
h_{n_1+1}f_{n_2+1}
\\
\hspace{4cm}
+
n_2
\left(
1+\frac{n_1}{n-1}
\right)
f_{n_1+1}h_{n_2+1}
\Biggr]
\\
-
\frac{1}{n(n-1)(n+1)}
\sum_{\substack{
n_1+n_2=n\\
n_1,n_2\ge3\ {\rm odd}
}}
n_1n_2
f_{n_1+1}f_{n_2+1}
\\
+
\frac{2}{n(n+1)}
f_n(2h_2-f_2)
+
\left[
\frac{n-2}{n(n+1)}
+
\frac{2}{n+1}f_2
\right]
h_n
+
\frac{2}{n(n+1)}
\partial_\mu h_n.
\end{multline}
All functions in these equations are evaluated at the common flow
parameter \(\mu\). In the remainder of the paper, we refer to
\eqref{EqSequenceF2Mu}--\eqref{EqSequenceHnMu} as the
\emph{second-order mean-field hierarchy}. 

\section{The second-order mean-field hierarchy}
\label{sec4}

The purpose of this section is to develop the analytic framework used
to construct solutions of the second-order mean-field hierarchy
\eqref{EqSequenceF2Mu}--\eqref{EqSequenceHnMu} and to establish their
ultraviolet triviality. The hierarchy is supplemented with the initial
conditions
\begin{equation}\label{EqMeanFieldBC}
f_2(0)
=
2(2\pi)^4\alpha_0c_2,
\qquad
f_4(0)
=
4\pi^2c_4,
\qquad
f_n(0)
=
0,
\quad
n\ge6,
\end{equation}
and
\begin{equation}\label{EqMeanFieldBC2}
h_2(0)
=
2(2\pi)^4\alpha_0c'_2,
\qquad
h_n(0)
=
0,
\quad
n\ge4,
\end{equation}
where the indices \(n\) are even and
\[
c_2,c'_2\in\mathbb R,
\qquad
c_4>0
\]
are fixed. Since the upper endpoint of the flow is normalized to
\(\alpha=1\), one has
\[
\mu_{\max}
=
\log\left(\frac1{\alpha_0}\right),
\]
so that the ultraviolet limit \(\alpha_0\rightarrow0\) is equivalent to
\(\mu_{\max}\to\infty\).

The construction is based on weighted sequence spaces controlling the
Taylor coefficients of the unknown functions at the initial scale
\(\mu=0\). Their weights are chosen to provide uniform estimates on
the coefficient recursions and sufficient control to realise the
resulting jets as smooth functions on bounded intervals. These
estimates also provide the control required to establish triviality in
the ultraviolet limit.

We can now state the main result of the paper. It establishes
the existence of smooth solutions for every finite interval $[0,\mu_{\max}]$ and
their triviality as the ultraviolet cutoff is removed:
\begin{theorem}[Triviality of the second-order mean-field model]\label{ThmSecondOrderTriviality}
Consider the \(\phi_4^4\)-theory with bare interaction
Lagrangian~\eqref{EqBareInteraction}, and impose the mean-field
boundary conditions~\eqref{EqMeanFieldBC}--\eqref{EqMeanFieldBC2},
with
\begin{equation}\label{EqBoundaryParametersTriviality}
0<c_4<+\infty,
\qquad
|c_2|<+\infty,
\qquad
|c'_2|<+\infty.
\end{equation}
Then there exist smooth solutions \(f_n\) and \(h_n\) of the
second-order mean-field flow equations
\eqref{EqSequenceFn}--\eqref{EqSequenceHn}, satisfying
\[
f_n,h_n\in C^\infty([0,\mu_{\max}]),
\qquad n\ge2,
\]
which vanish in the ultraviolet limit. More precisely,
\begin{equation}\label{EqSecondOrderTriviality}
\lim_{\mu_{\max}\to+\infty} f_n(\mu_{\max})
=
\lim_{\mu_{\max}\to+\infty} h_n(\mu_{\max})
=
0,
\qquad n\ge2.
\end{equation}
\end{theorem}
Next we introduce the weighted sequence classes. We then reformulate the coefficient hierarchy in
terms of suitable linear and bilinear operators and establish their
mapping properties. Finally, these estimates are used to construct a
family of non-perturbative solutions satisfying
\eqref{EqMeanFieldBC}--\eqref{EqMeanFieldBC2} and to prove that this
family becomes trivial as \(\mu_{\max}\to\infty\). No uniqueness claim
is made in the present work.
\subsection{Weighted sequence classes}
Throughout this section, we use the convention
\[
x!
:=
\Gamma(x+1),
\qquad
x\ge0.
\]
In particular, expressions of the form
\(\lvert n/4+i-3\rvert!\) are understood through the Gamma function. 

Let \(K>1\). For every even integer \(n\ge6\), define
\begin{multline}\label{Sn}
\mathcal S_n(K)
:=
\Biggl\{
(u_{n,i})_{i\in\mathbb N_0}\in\mathbb R^{\mathbb N_0}
:
|u_{n,0}|
\le
\frac{K^{\frac n2-\frac32}}{2n^2},
\quad
|u_{n,1}|
\le
\frac{K^{\frac n2-\frac32}}{n^2}
\left(
1+\frac{nK}{2}
\right),
\\
|u_{n,i}|
\le
K^{\frac n2+i-\frac32}
\frac{
\left|\frac n4+i-3\right|!
}{
(i!)^{1/4}
},
\qquad
i\ge2
\Biggr\}.
\end{multline}
For \(n=4\), set
\begin{equation}\label{S4}
\mathcal S_4(K)
:=
\Biggl\{
(u_{4,i})_{i\in\mathbb N_0}\in\mathbb R^{\mathbb N_0}
:
|u_{4,0}|
\le
\frac{\sqrt K}{32},
\quad
|u_{4,1}|
\le
\frac{K}{32},
\quad
|u_{4,i}|
\le
K^{i+\frac12}
\frac{|i-2|!}{(i!)^{1/4}},
\quad
i\ge2
\Biggr\},
\end{equation}
and, for \(n=2\), set
\begin{multline}\label{S2}
\mathcal S_2(K)
:=
\left\{
(u_{2,i})_{i\in\mathbb N_0}\in\mathbb R^{\mathbb N_0}
:
|u_{2,0}|
\le
\frac{\sqrt K}{4},
\quad
|u_{2,1}|
\le
\frac{K}{2},\right.\\\left.
\quad
|u_{2,i}|
\le
K^{i+\frac12}
\frac{|i-3|!}{((i-1)!)^{1/4}},
\quad
i\ge2
\right\}.
\end{multline}

For later use, we also introduce the corresponding component classes.
For every even integer \(n\ge2\) and every \(i\in\mathbb N_0\), we
denote by
\begin{equation}\label{DefSnComponent}
\mathcal S_n^{(i)}(K)
:=
\left\{
u_{n,i}\in\mathbb R
:
u_n\in\mathcal S_n(K)
\right\}
\subset\mathbb R
\end{equation}
the \(i\)-th component class associated with
\(\mathcal S_n(K)\). Thus,
\[
u_n\in\mathcal S_n(K)
\quad\Longleftrightarrow\quad
u_{n,i}\in\mathcal S_n^{(i)}(K)
\quad\text{for every }i\in\mathbb N_0.
\]
\subsection{Strategy of the proof}

We now describe the proof of
Theorem~\ref{ThmMeanFieldTriviality} and clarify the respective roles
of the discrete and functional formulations of the second-order
mean-field hierarchy. The argument proceeds in three stages: a
simultaneous construction of the coefficient hierarchy, a smooth
realisation of the two-point coefficient sequences, and a functional
construction of the higher sectors.

The starting point is that the hierarchy is to be generated from the
three parameters
\[
f_{2,0},
\qquad
h_{2,0},
\qquad
g_{4,0},
\]
which represent, within the second-order mean-field closure, the mass,
wave-function, and coupling data. No higher Taylor coefficient is
prescribed independently.

To uncover the resulting recursion, we first treat the unknown
functions formally as power series in \(\mu\):
\[
f_2(\mu)
=
\sum_{i\ge0}f_{2,i}\mu^i,
\qquad
h_2(\mu)
=
\sum_{i\ge0}h_{2,i}\mu^i,
\]
and similarly for the higher families \(g_n\) and \(k_n\).
Substitution into the mean-field equations and identification of equal
powers of \(\mu\) yield a coupled discrete hierarchy for
\[
f_{2,i},
\qquad
h_{2,i},
\qquad
g_{n,i},
\qquad
k_{n,i}.
\]
At this stage, the expansions are used only to derive the coefficient
recursions; no convergence or analyticity is assumed.

The discrete hierarchy must be analysed simultaneously. The
recursions for the two-point coefficients involve the four-point
sector, while the equations for \(g_{n,i}\) and \(k_{n,i}\) are coupled
both in the correlation order \(n\) and in the coefficient index
\(i\). The weighted classes \(\mathcal S_n(K)\) are designed so that
the shift, multiplication, and bilinear convolution operations
appearing in these recursions preserve the required component bounds.
Their mapping properties allow the complete coefficient family to be
constructed recursively from the three initial parameters while
retaining uniform control in both \(n\) and \(i\).

The outcome of this first stage is purely discrete: it produces and
controls the infinite jets of the two-point sector, but does not yet
construct functions of \(\mu\). For the second stage, following
\cite{Kopper2022,KopperWang2025}, we introduce auxiliary sequences
\((b_n)_{n\ge1}\) and \((d_n)_{n\ge1}\) and set
\begin{align}
f_2(\mu)
&=
\sum_{n\ge1}
b_n
\frac{(n\mu)^{n-1}}{1+(n\mu)^n},
\label{EqAnsatzf2}
\\
h_2(\mu)
&=
\sum_{n\ge1}
d_n
\frac{(n\mu)^{n-1}}{1+(n\mu)^n}.
\label{EqAnsatzh2}
\end{align}
The coefficients \(b_n\) and \(d_n\) are determined through triangular
relations requiring
\[
\frac{\partial_\mu^i f_2(0)}{i!}
=
f_{2,i},
\qquad
\frac{\partial_\mu^i h_2(0)}{i!}
=
h_{2,i}.
\]
The estimates obtained from the discrete hierarchy imply convergence
of these series, together with all their derivatives, on bounded
subsets of \(\mathbb R_+\). They therefore realise the prescribed jets
as genuine smooth functions. No analyticity at \(\mu=0\) is required.

The particular form of the two-point ansatz also yields the estimates
needed to prove the triviality of \(f_2\), \(h_2\), and all
their derivatives. This is the only stage at which the large-\(\mu\)
structure of the chosen smooth realisation is used.

In the third stage, we return to the original, unexpanded mean-field
hierarchy. Once \(f_2\) and \(h_2\) have been constructed, its
triangular structure in the correlation order allows the higher
functions \(g_n\) and \(k_n\) to be constructed recursively. Repeated
differentiation of the functional equations propagates smoothness and
triviality from the two-point sector to every higher sector.

Finally, differentiating the functional hierarchy at \(\mu=0\) shows
that the jets of the constructed higher functions satisfy the same
discrete recursions and initial conditions as the coefficient family
obtained in the first stage. The recursive determination of the
discrete hierarchy therefore gives
\[
\frac{\partial_\mu^i g_n(0)}{i!}
=
g_{n,i},
\qquad
\frac{\partial_\mu^i k_n(0)}{i!}
=
k_{n,i}.
\]
This closes the argument and proves that the smooth
solution realizes the complete coefficient hierarchy generated from
the three prescribed parameters.

We now turn to the discrete mean-field hierarchy and establish the
operator estimates required in the first stage of the proof.
\section{The discrete mean-field hierarchy}
\label{SecDiscreteHierarchy0}
Before deriving the coefficient equations, we determine the first
order at which each higher-point function may have a nonzero Taylor
coefficient at \(\mu=0\). The initial conditions and the recursive
structure of the hierarchy imply that the number of vanishing
low-order derivatives increases with $n$.
\begin{proposition}[Vanishing of the low-order \(\mu\)-jets]
\label{PropVanishingJets}
Let
$
(f_n,h_n)_{n\ge2,n\in2\mathbb{N}}
$
be a smooth solution of the second-order mean-field hierarchy
\eqref{EqSequenceF2Mu}--\eqref{EqSequenceHnMu}, subject to the initial
conditions \eqref{EqMeanFieldBC}--\eqref{EqMeanFieldBC2}. Then, for
every even integer \(n\ge6\),
\begin{align}
\partial_\mu^\ell f_n(0)
&=0,
&
0\le \ell\le \frac n2-3,
\label{EqVanishingJetsFn}
\\
\partial_\mu^\ell h_n(0)
&=0,
&
0\le \ell\le \frac n2-3.
\label{EqVanishingJetsHn}
\end{align}
Equivalently,
\[
f_n(\mu)
=
O\left(\mu^{\frac n2-2}\right),
\qquad
h_n(\mu)
=
O\left(\mu^{\frac n2-2}\right),
\qquad
\mu\to0,
\]
for every even integer \(n\ge6\).
\end{proposition}
\begin{proof}
The initial conditions
\eqref{EqMeanFieldBC}--\eqref{EqMeanFieldBC2} imply
\begin{equation}\label{EqZerothOrderVanishing}
f_n(0)=h_n(0)=0
\qquad
\text{for every even integer }n\ge6.
\end{equation}
This proves \eqref{EqVanishingJetsFn}--\eqref{EqVanishingJetsHn}
for \(\ell=0\).

We proceed by induction on the order of the jet, imposing the induction
hypothesis simultaneously for every even \(n\). Fix \(r\in\mathbb N_0\)
and assume that
\begin{equation}\label{EqJetsInductionHypothesis}
\partial_\mu^j f_m(0)
=
\partial_\mu^j h_m(0)
=
0
\end{equation}
for every even integer \(m\ge6\) and every
\[
0\le j\le r
\qquad\text{such that}\qquad
j\le\frac m2-3.
\]
We prove the statement at order \(r+1\). Let \(n\ge6\) be even and suppose that
\begin{equation}\label{EqAdmissibleR}
r+1\le\frac n2-3.
\end{equation}
In particular,
\begin{equation}\label{EqRUpperBound}
r\le\frac n2-4.
\end{equation}
After
differentiating \(r\) times and applying Leibniz' rule, every such term
is a linear combination of products of the form
\[
\partial_\mu^j u_a(0)\,
\partial_\mu^{r-j}v_b(0),
\qquad
a+b=n+2,
\qquad
u,v\in\{f,h\},
\]
where \(a,b\ge4\) are even. If neither factor vanished by the induction
hypothesis, one would necessarily have
\[
j\ge\frac a2-2,
\qquad
r-j\ge\frac b2-2.
\]
Adding these inequalities would give
\[
r
\ge
\frac{a+b}{2}-4
=
\frac n2-3,
\]
contradicting \eqref{EqRUpperBound}. Thus every product contains at
least one vanishing factor.

We now consider the \(h\)-sector. Differentiating
\eqref{EqSequenceHnMu} \(r\) times and evaluating at \(\mu=0\) gives
\begin{align*}
\frac{2}{n(n+1)}
\partial_\mu^{r+1}h_n(0)
&=
\partial_\mu^r h_{n+2}(0)
-
\frac{2}{n(n+1)}
\partial_\mu^r
\bigl[
f_n(2h_2-f_2)
\bigr](0)
\\
&\quad
-
\frac{2}{n+1}
\partial_\mu^r
\bigl[
h_nf_2
\bigr](0)
-
\frac{n-2}{n(n+1)}
\partial_\mu^r h_n(0),
\end{align*}
since all interior nonlinear contributions vanish by the preceding
argument. The induction hypothesis gives
\[
\partial_\mu^r h_{n+2}(0)=0,
\]
because
\[
r\le\frac n2-4
<
\frac{n+2}{2}-3.
\]
Moreover, every term in the Leibniz expansions of
\[
\partial_\mu^r
\bigl[
f_n(2h_2-f_2)
\bigr](0)
\quad\text{and}\quad
\partial_\mu^r
\bigl[
h_nf_2
\bigr](0)
\]
contains a derivative of \(f_n\) or \(h_n\) of order at most \(r\).
Since
\[
r\le\frac n2-4<\frac n2-3,
\]
these derivatives vanish by the induction hypothesis. The same
hypothesis gives
\[
\partial_\mu^r h_n(0)=0.
\]
It follows that
\[
\partial_\mu^{r+1}h_n(0)=0.
\]
We next consider the \(f\)-equation. Differentiating
\eqref{EqSequenceFnMu} \(r\) times and evaluating at \(\mu=0\), the
induction hypothesis gives
\[
\partial_\mu^r f_{n+2}(0)
=
\partial_\mu^r h_{n+2}(0)
=
0.
\]
The interior nonlinear terms vanish by the same product argument as
above. Finally, every term obtained by differentiating
\[
\left(
2f_2+1-\frac4n
\right)f_n
\]
contains a derivative of \(f_n\) of order at most \(r\), and therefore
vanishes by the induction hypothesis. Consequently,
\[
\partial_\mu^{r+1}f_n(0)=0.
\]
The induction is complete and proves
\eqref{EqVanishingJetsFn}--\eqref{EqVanishingJetsHn}.
\end{proof}
\noindent By Proposition~\ref{PropVanishingJets}, for every even integer
\(n\ge4\) there exist smooth functions \(g_n\) and \(k_n\) such that
\begin{equation}\label{EqRescaledFunctions}
f_n(\mu)
=
\mu^{\frac n2-2}g_n(\mu),
\qquad
h_n(\mu)
=
\mu^{\frac n2-2}k_n(\mu).
\end{equation}
For \(n=4\), this simply means that \(g_4=f_4\) and \(k_4=h_4\).
In particular,
\[
k_4(0)=0
\]
by the boundary condition \eqref{EqMeanFieldBC2}. Substituting \eqref{EqRescaledFunctions} into
\eqref{EqSequenceF2Mu} gives
\begin{equation}\label{EqSequenceG4Mu}
g_4
+
4k_4
=
\frac13 f_2(f_2-1)
+
\frac13\partial_\mu f_2.
\end{equation}
Likewise, \eqref{EqSequenceH2Mu} becomes
\begin{equation}\label{EqSequenceK4Mu}
k_4
=
\frac16 f_2(4h_2-f_2)
+
\frac13\partial_\mu h_2.
\end{equation}
For every even integer \(n\ge4\), one has
\begin{multline}\label{EqSequenceGnMu}
\mu^2
\bigl(
g_{n+2}
+
4k_{n+2}
\bigr)
=
\frac{1}{n+1}
\sum_{\substack{n_1+n_2=n\\ n_1,n_2\ge3}}
g_{n_1+1}g_{n_2+1}
\\
+
\frac{\mu}{n+1}
\left(
2f_2+1-\frac4n
\right)g_n
+
\frac{n-4}{n(n+1)}g_n
+
\frac{2\mu}{n(n+1)}
\partial_\mu g_n.
\end{multline}

Similarly, \eqref{EqSequenceHnMu} gives
\begin{multline}\label{EqSequenceKnMu}
\mu^2k_{n+2}
=
\frac{1}{n(n+1)}
\sum_{\substack{
n_1+n_2=n\\
n_1,n_2\ge3\ {\rm odd}
}}
\Biggl[
n_1
\left(
1+\frac{n_2}{n-1}
\right)
k_{n_1+1}g_{n_2+1}
\\
\hspace{3.5cm}
+
n_2
\left(
1+\frac{n_1}{n-1}
\right)
g_{n_1+1}k_{n_2+1}
\Biggr]
\\
-
\frac{1}{n(n-1)(n+1)}
\sum_{\substack{
n_1+n_2=n\\
n_1,n_2\ge3\ {\rm odd}
}}
n_1n_2
g_{n_1+1}g_{n_2+1}
\\
+
\frac{2\mu}{n(n+1)}
g_n(2h_2-f_2)
+
\frac{2\mu}{n+1}
k_nf_2
+
\frac{n-4}{n(n+1)}k_n
+
\frac{(n-2)\mu}{n(n+1)}k_n
+
\frac{2\mu}{n(n+1)}
\partial_\mu k_n.
\end{multline}
Motivated by the previous system, we formally expand the unknown
functions into power series,
\begin{equation}\label{EqFormalExpansion}
g_n(\mu)
=
\sum_{i\ge0}\mu^i g_{n,i},
\qquad
k_n(\mu)
=
\sum_{i\ge0}\mu^i k_{n,i},
\qquad
f_2(\mu)
=
\sum_{i\ge0}\mu^i f_{2,i},
\qquad
h_2(\mu)
=
\sum_{i\ge0}\mu^i h_{2,i}.
\end{equation}
Accordingly,
\[
\partial_\mu g_n(\mu)
=
\sum_{i\ge1}
i\,\mu^{i-1}g_{n,i},
\qquad
\partial_\mu k_n(\mu)
=
\sum_{i\ge1}
i\,\mu^{i-1}k_{n,i}.
\]

Substituting these formal expansions into
\eqref{EqSequenceG4Mu}--\eqref{EqSequenceKnMu} and identifying the
coefficients of equal powers of \(\mu\) yields a recursive system for
the coefficients
\(
(g_{n,i},k_{n,i},f_{2,i},h_{2,i})
\).
The remainder of this section is devoted to proving that, for a
sufficiently large constant \(K\), these coefficients satisfy the
bounds defining the weighted classes \(\mathcal S_n(K)\). These
estimates control the complete formal jets and provide the input for the smooth realization
constructed in the next stage. No convergence of the formal series
\eqref{EqFormalExpansion} is asserted at this point.
\subsection{Recursion for the Taylor coefficients}
\label{SecDiscreteHierarchy}

Identifying equal powers of \(\mu\) in
\eqref{EqSequenceG4Mu} and \eqref{EqSequenceK4Mu} gives, for every
\(i\in\mathbb N_0\),
\begin{equation}\label{eq:f2_recursion}
f_{2,i+1}
=
\frac{1}{i+1}
\left[
3g_{4,i}
+
12k_{4,i}
-
\sum_{j=0}^{i}
f_{2,j}f_{2,i-j}
+
f_{2,i}
\right],
\end{equation}
and
\begin{equation}\label{eq:h2_recursion}
h_{2,i+1}
=
\frac{1}{2(i+1)}
\left[
6k_{4,i}
-
4\sum_{j=0}^{i}
f_{2,j}h_{2,i-j}
+
\sum_{j=0}^{i}
f_{2,j}f_{2,i-j}
\right].
\end{equation}
Let \(n\ge4\) be even and \(i\in\mathbb N_0\). Identifying the
coefficient of \(\mu^{i+2}\) in \eqref{EqSequenceGnMu} yields
\begin{align}
g_{n,i+2}
&=
\frac{n(n+1)}{n+2i}
\left(
g_{n+2,i}+4k_{n+2,i}
\right)
-
\frac{n}{n+2i}
\sum_{\substack{
n_1+n_2=n+2\\
n_1,n_2\ge4,\ {n\in 2\mathbb{N}}
}}
\sum_{j=0}^{i+2}
g_{n_1,j}g_{n_2,i+2-j}
\nonumber\\
&\quad
-
\frac{2n}{n+2i}
\sum_{j=0}^{i+1}
g_{n,j}f_{2,i+1-j}
-
\frac{n-4}{n+2i}
g_{n,i+1}.
\label{eq:g_recursion}
\end{align}
Similarly, \eqref{EqSequenceKnMu} gives
\begin{align}
k_{n,i+2}
&=
\frac{n(n+1)}{n+2i}
k_{n+2,i}
-
\sum_{\substack{
n_1+n_2=n+2\\
n_1,n_2\ge4,~{n\in 2\mathbb{N}}
}}
\frac{2(n_1-1)}{n+2i}
\left(
1+\frac{n_2-1}{n-1}
\right)
\sum_{j=0}^{i+2}
k_{n_1,j}g_{n_2,i+2-j}
\nonumber\\
&\quad
+
\frac{1}{(n-1)(n+2i)}
\sum_{\substack{
n_1+n_2=n+2\\
n_1,n_2\ge4,~{n\in 2\mathbb{N}}
}}
(n_1-1)(n_2-1)
\sum_{j=0}^{i+2}
g_{n_1,j}g_{n_2,i+2-j}
\nonumber\\
&\quad
-
\frac{2}{n+2i}
\sum_{j=0}^{i+1}
g_{n,j}
\left(
2h_{2,i+1-j}
-
f_{2,i+1-j}
\right)
-
\frac{2n}{n+2i}
\sum_{j=0}^{i+1}
k_{n,j}f_{2,i+1-j}
-
\frac{n-2}{n+2i}
k_{n,i+1}.
\label{eq:k_recursion}
\end{align}

The smoothness of solutions at $\mu=0$ together with
\eqref{EqSequenceGnMu}--\eqref{EqSequenceKnMu} imposes
compatibility conditions on the first two coefficient levels. At
order zero, one obtains
\begin{equation}\label{eq:g_compatibility_zero}
\sum_{\substack{
n_1+n_2=n+2\\
n_1,n_2\ge4\ {\rm even}
}}
g_{n_1,0}g_{n_2,0}
+
\frac{n-4}{n}g_{n,0}
=
0,
\qquad n\ge4.
\end{equation}
For the \(k\)-sector, this gives, for every even \(n\ge6\),
\begin{multline}\label{eq:k_compatibility_zero}
k_{n,0}
=
\frac{2}{4-n}
\sum_{\substack{
n_1+n_2=n+2\\
n_1,n_2\ge4,~n\in2\mathbb{N}
}}
(n_1-1)
\left(
1+\frac{n_2-1}{n-1}
\right)
k_{n_1,0}g_{n_2,0}
\\
+
\sum_{\substack{
n_1+n_2=n+2\\
n_1,n_2\ge4,~n\in2\mathbb{N}
}}
\frac{(n_1-1)(n_2-1)}{(n-1)(n-4)}~
g_{n_1,0}g_{n_2,0}.
\end{multline}
At first order, \eqref{EqSequenceGnMu} gives
\begin{equation}\label{eq:g_compatibility_one}
2
\sum_{\substack{
n_1+n_2=n+2\\
n_1,n_2\ge4,~n\in 2\mathbb{N}
}}
g_{n_1,0}g_{n_2,1}
+
g_{n,0}
\left(
2f_{2,0}+1-\frac4n
\right)
+
\frac{n-2}{n}g_{n,1}
=
0.
\end{equation}
Finally, the coefficient of \(\mu\) in
\eqref{EqSequenceKnMu} yields
\begin{multline}\label{eq:k_compatibility_one}
2
\sum_{\substack{
n_1+n_2=n+2\\
n_1,n_2\ge4,~n\in 2\mathbb{N
}}}
\left[
n+
\frac{2(n_1-1)(n_2-1)}{n-1}
\right]
k_{n_1,1}g_{n_2,0}
-
\frac{2}{n-1}
\sum_{\substack{
n_1+n_2=n+2\\
n_1,n_2\ge4,~n\in 2\mathbb{N}
}}
(n_1-1)(n_2-1)
g_{n_1,0}g_{n_2,1}
\\
+
2g_{n,0}
\left(
2h_{2,0}-f_{2,0}
\right)
+
2nk_{n,0}f_{2,0}
+
(n-2)k_{n,0}
+
(n-2)k_{n,1}
=
0.
\end{multline}
The boundary conditions
\eqref{EqMeanFieldBC}--\eqref{EqMeanFieldBC2} determine the zeroth-order
coefficients
\begin{equation}\label{EqDiscreteInitialData}
f_{2,0}
=
2(2\pi)^4\alpha_0c_2,
\qquad
h_{2,0}
=
2(2\pi)^4\alpha_0c_2',
\qquad
g_{4,0}
=
4\pi^2c_4,
\qquad
k_{4,0}
=
0.
\end{equation}
Thus \(f_{2,0}\), \(h_{2,0}\), and \(g_{4,0}\) encode the three
prescribed parameters of the model, whereas \(k_{4,0}\) is fixed by
the vanishing boundary condition for \(h_4\). The compatibility
conditions
\eqref{eq:g_compatibility_zero}--\eqref{eq:k_compatibility_one},
together with the recursions
\eqref{eq:f2_recursion}--\eqref{eq:k_recursion}, then determine the
remaining coefficients recursively. The weighted sequence classes
introduced above are designed to control this construction
simultaneously in $n$ and $i$.
\begin{theorem}[Construction and stability of the coefficient hierarchy]
\label{ThmCoefficientHierarchy}
There exists a constant \(K_0>1\) such that the following holds. Let
the initial coefficients be given by \eqref{EqDiscreteInitialData},
where
\[
c_4>0,
\qquad
c_2,c_2'\in\mathbb R.
\]
Assume that \(K\ge K_0\) is sufficiently large that
\begin{equation}\label{EqInitialBounds0}
|f_{2,0}|
\le
\frac{\sqrt K}{4},
\qquad
|h_{2,0}|
\le
\frac{\sqrt K}{16},
\qquad
|g_{4,0}|
\le
\frac{\sqrt K}{32}.
\end{equation}
Then the recursions
\eqref{eq:f2_recursion}--\eqref{eq:k_recursion}, together with the
compatibility conditions
\eqref{eq:g_compatibility_zero}--\eqref{eq:k_compatibility_one},
recursively determine a coefficient family
\[
(f_{2,i})_{i\in\mathbb N_0},
\qquad
(h_{2,i})_{i\in\mathbb N_0},
\qquad
(g_{n,i})_{\substack{n\ge4,\;n\in2\mathbb N\\ i\in\mathbb N_0}},
\qquad
(k_{n,i})_{\substack{n\ge4,\;n\in2\mathbb N\\ i\in\mathbb N_0}}.
\]
Moreover, for every \(i\in\mathbb N_0\),
\[
f_{2,i},h_{2,i}
\in
\mathcal S_2^{(i)}(K),
\]
and, for every even integer \(n\ge4\),
\[
g_{n,i},k_{n,i}
\in
\mathcal S_n^{(i)}(K).
\]
\end{theorem}

\subsection{Operator formulation of the coefficient hierarchy}\label{sec6}
We now recast the coefficient hierarchy as a system of equations on
the weighted sequence classes introduced above. The coefficient
recursions involve only three elementary operations: shifts in the
Taylor index, multiplication by rational functions of this index, and
discrete convolution. Isolating these operations will allow the
estimates required for the construction to be proved independently of
the particular equation in which they occur. We begin with the linear operators.

\begin{definition}[Elementary linear operators]
For every even integer \(n\ge2\), define the forward shift operator
\[
\mathbb T:
\mathcal S_n(K)
\longrightarrow
\mathbb R^{\mathbb N_0},
\]
by
\[
(\mathbb Tu_n)_i
=
u_{n,i+1},
\qquad
i\in\mathbb N_0.
\]
For every even integer \(n\ge2\), define the multiplication operators
\[
\mathbb M^{(1)}_{p}:
\mathcal S_n(K)
\longrightarrow
\mathbb R^{\mathbb N_0},\qquad \tilde{\mathbb M}^{(1)}_{p}:
\mathcal S_n(K)
\longrightarrow
\mathbb R^{\mathbb N_0},
\]
by
\[
(\mathbb M^{(1)}_{p}u_n)_i
=
\frac{p-4}{p+2i}\,
u_{n,i},\qquad (\tilde{\mathbb M}^{(1)}_{p}u_n)_i
=
\frac{p-2}{p+2i}\,
u_{n,i}
\qquad
i\in\mathbb N_0.
\]
For every even integer \(n\ge4\), define the multiplication operator
\[
\mathbb M^{(2)}_{p}:
\mathcal S_n(K)
\longrightarrow
\mathbb R^{\mathbb N_0},
\]
by
\[
(\mathbb M^{(2)}_{p}u_n)_i
=
\frac{p(p+1)}{p+2i}\,
u_{n,i},
\qquad
i\in\mathbb N_0.
\]
\end{definition}
\noindent The linear terms appearing in the coefficient hierarchy are precisely
the compositions
\[
\tilde{\mathbb M}^{(1)}_{n}\circ\mathbb T,\qquad\mathbb M^{(1)}_{n}\circ\mathbb T,
\qquad
\mathbb M^{(2)}_{n}
\]
so that the forward shift and the multiplication operators constitute the elementary building blocks from which the linear part of the
hierarchy is assembled.

The nonlinear part of the hierarchy is generated by the discrete
convolution of sequences.

\begin{definition}[Discrete convolution]
For every pair of even integers \(n,n'\ge2\), define the bilinear map
\[
\mathbb C_{n,n'}:
\mathcal S_n(K)\times\mathcal S_{n'}(K)
\longrightarrow
\mathbb R^{\mathbb N_0},
\]
by
\[
\bigl(\mathbb C_{n,n'}(u_n,w_{n'})\bigr)_i
:=
\sum_{j=0}^{i}
u_{n,j}\,
w_{n',i-j},
\qquad
i\in\mathbb N_0.
\]
\end{definition}
For every even integer \(n\ge4\), set
\[
\mathfrak I_n
:=
\left\{
n_1\in2\mathbb N:
4\le n_1\le n-2
\right\},
\]
with the convention that \(\mathfrak I_4=\varnothing\). For
\(n_1\in\mathfrak I_n\), we write
\[
n_2:=n+2-n_1.
\]
Then \(n_1,n_2\ge4\) and \(n_1+n_2=n+2\).
The coefficients multiplying the convolution terms may depend on the
particular decomposition \(n+2=n_1+n_2\). We therefore allow the
coefficient to depend on \(n_1\), while imposing a bound that is
uniform over all admissible decompositions. This uniform bound is the only
property of the coefficients used later in the space estimates. Keeping
the ordered dependence on \(n_1\) allows the same operator to cover the
asymmetric terms appearing in the \(k_n\)-hierarchy.
\begin{definition}\label{Admissible}[Admissible coefficients]
A family
\[
\alpha
=
\left(
\alpha_{n_1,n}^{(i)}
\right)_{
\substack{
n\in2\mathbb N,\; n\ge4,\\
n_1\in\mathfrak I_n,\; i\in\mathbb N_0
}}
\]
is called admissible if
\[
\alpha_{n_1,n}^{(0)}
=
\alpha_{n_1,n}^{(1)}
=
0
\]
and
\begin{equation}\label{EqAlphaBound}
\left|
\alpha_{n_1,n}^{(i)}
\right|
\le
\frac{n}{n+2i-4},
\qquad
n\ge4,\quad
n_1\in\mathfrak I_n,\quad
i\ge2.
\end{equation}
\end{definition}
\noindent Since the bilinear operators simultaneously involve all lower-order
sequences, it is convenient to collect them into the product space
\begin{equation}
\widehat{\mathcal S}_n(K)
:=
\prod_{\substack{m\in2\mathbb N\\2\le m\le n}}
\mathcal S_m(K),\quad n\ge 4~.
\end{equation}

\begin{definition}[Bilinear operators]
Let $\alpha$ be an admissible family of coefficients. For every even
integer \(n\ge6\) and 
\[
u=(u_m)_{2\le m\le n+2},
\qquad
w=(w_m)_{2\le m\le n+2},
\]
define
\[
\mathbb B_{n}^{\alpha}:
\widehat{\mathcal S}_{n+2}(K)
\times
\widehat{\mathcal S}_{n+2}(K)
\longrightarrow
\mathbb R^{\mathbb N_0}
\]
by
\begin{equation}\label{EqBilinearOperator}
\bigl(
\mathbb B_{n}^{\alpha}(u,w)
\bigr)_i
:=
\sum_{\substack{4\le n_1\le n-2\\n_1\in2\mathbb N}}
\alpha_{n_1,n}^{(i)}
\bigl(
\mathbb C_{n_1,n+2-n_1}
(u_{n_1},w_{n+2-n_1})
\bigr)_i,
\qquad
i\in\mathbb N_0.
\end{equation}
\end{definition}
The two-point sector is treated separately, since it satisfies different
componentwise bounds.
\begin{definition}[Two-point sector bilinear operator]
Let \(n\ge2\) be even and let
$
\alpha_{2,n}
=
\bigl(\alpha_{2,n}^{(i)}\bigr)_{i\in\mathbb N_0}$
be a family of real coefficients. We define
\[
\mathbb B_{2,n}^{\alpha}:
\mathcal S_2(K)\times\mathcal S_n(K)
\longrightarrow
\mathbb R^{\mathbb N_0}
\]
by
\begin{equation}\label{EqBilinearOperator2n}
\bigl(\mathbb B_{2,n}^{\alpha}(u,w)\bigr)_i
:=
\alpha_{2,n}^{(i)}
\bigl(\mathbb C_{2,n}(u,w)\bigr)_i,
\qquad i\in\mathbb N_0.
\end{equation}
\end{definition}
\indent In terms of these operators, the equations defining the hierarchy can be written compactly as follows 
\begin{equation}\label{f2OP}
\mathbb{T}\left(f_2\right)=\mathbb{M}^{(2)}_{2}\left(g_4+4k_4\right)-\mathbb{M}^{(1)}_{2}(f_2)-\mathbb{B}_{2,2}^{\alpha_{2,2}}\left(f_2,f_2\right)
\end{equation}
and 
\begin{equation}\label{h2OP}
\mathbb{T}\left(h_2\right)=\mathbb{M}^{(2)}_{2}\left(k_4\right)-\frac{1}{2}\mathbb{B}_{2,2}^{\alpha_{2,2}}\left(f_2,f_2-4h_2\right)~.
\end{equation}
For every even integer \(n\ge4\), set
\[
n_2=n+2-n_1
\]
and define the coefficient families
\(\alpha_{gg}\), \(\alpha_{kg}\), \(\alpha_g\), and \(\alpha_k\) by
\begin{align}
\alpha_{gg;n_1,n}^{(r)}
&:=
\frac{(n_1-1)(n_2-1)}
{(n-1)(n+2r-4)},
\qquad
r\ge2,
\label{EqAlphaGg}
\\
\alpha_{kg;n_1,n}^{(r)}
&:=
-
\frac{n_1-1}{n+2r-4}
\left(
1+\frac{n_2-1}{n-1}
\right),
\qquad
r\ge2,
\label{EqAlphaKg}
\\
\alpha_{g;2,n}^{(r)}
&:=
-
\frac{1}{n+2r-2},
\qquad
r\ge1,
\label{EqAlphaG}
\\
\alpha_{k;2,n}^{(r)}
&:=
-
\frac{n}{n+2r-2},
\qquad
r\ge1.
\label{EqAlphaK}
\end{align}
Here, in \eqref{EqAlphaGg} and \eqref{EqAlphaKg}, the indices satisfy
\[
n_1+n_2=n+2,
\qquad
n_1,n_2\ge4.
\]
With these definitions, the recursion for \(k_n\) takes the operator form
\begin{multline}\label{EqOperatorK}
\mathbb T^2(k_n)
=
\mathbb M_n^{(2)}(k_{n+2})
-
\tilde{\mathbb M}_n^{(1)}\circ\mathbb T(k_n)
+
\mathbb T^2\circ
\mathbb B_n^{\alpha_{gg}}(g,g)
+
2\,
\mathbb T^2\circ
\mathbb B_n^{\alpha_{kg}}(k,g)
\\
+
2\,
\mathbb T\circ
\mathbb B_{2,n}^{\alpha_g}
\bigl(2h_2-f_2,g_n\bigr)
+
2\,
\mathbb T\circ
\mathbb B_{2,n}^{\alpha_k}
\bigl(f_2,k_n\bigr),
\qquad n\ge4.
\end{multline}
For every even integer \(n\ge4\), define the coefficient families
\(\beta_{gg}\) and \(\beta_{gf}\) by
\begin{align}
\beta_{gg;n_1,n}^{(r)}
&:=
-\frac{n}{n+2r-4},
\qquad
r\ge2,
\label{EqAlphaGgForG}
\\
\beta_{gf;2,n}^{(r)}
&:=
-\frac{n}{n+2r-2},
\qquad
r\ge1.
\label{EqAlphaGfForG}
\end{align}
In \eqref{EqAlphaGgForG}, the indices satisfy
\[
n_1+n_2=n+2,
\qquad
n_1,n_2\ge4,
\]
with \(n_2=n+2-n_1\). Notice that
\(\beta_{gg;n_1,n}^{(r)}\) is independent of the particular
decomposition \(n+2=n_1+n_2\). With these definitions, the recursion \eqref{eq:g_recursion} takes the
operator form
\begin{multline}\label{EqOperatorG}
\mathbb T^2(g_n)
=
\mathbb M_n^{(2)}
\bigl(g_{n+2}+4k_{n+2}\bigr)
-
\mathbb M_n^{(1)}\circ\mathbb T(g_n)
\\
+
\mathbb T^2\circ
\mathbb B_n^{\beta_{gg}}(g,g)
+
2\,
\mathbb T\circ
\mathbb B_{2,n}^{\beta_{gf}}
\bigl(f_2,g_n\bigr),
\qquad n\ge4.
\end{multline}
\subsection{Mapping properties of the multiplication operators}

For every even integer \(n\ge2\), every \(K\ge1\), and every
\(i\in\mathbb N_0\), recall that
\[
\mathcal S_n^{(i)}(K)
:=
\left\{
u_{n,i}
:\;
u_n\in\mathcal S_n(K)
\right\}
\subset\mathbb R
\]
which denotes the set of all possible \(i\)-th components of sequences belonging to
\(\mathcal S_n(K)\). We next establish the mapping properties of the operators introduced in
the previous subsection. We begin with the multiplication operators,
which encode the linear part of the hierarchy. In the sequel, we shall repeatedly use the functional equation of the Gamma function
\begin{equation}\label{EqGammaFunctionalEquation}
\Gamma(x+1)=x\Gamma(x),
\qquad x>0,
\end{equation}
and its iterated form
\begin{equation}\label{EqGammaIterated}
\frac{\Gamma(x)}{\Gamma(x+m)}
=
\frac{1}{x(x+1)\cdots(x+m-1)},
\qquad x>0,\quad m\in\mathbb N.
\end{equation}
See, for instance, \cite[Eq.~5.5.1]{NISTHandbook}.
\begin{proposition}\label{PropMultiplicationOperators}
There exists a constant \(C_{\mathbb M}>1\) such that, for every
\(K\ge C_{\mathbb M}\), every even integer \(n\ge2\), and every
\(i\in\mathbb N_0\), one has
\begin{equation}\label{EqMapM2}
\mathbb M_n^{(2)}
:
\mathcal S_{n+2}^{(i)}(K)
\longrightarrow
\mathcal S_n^{(i+2)}(K).
\end{equation}
The constant \(C_{\mathbb M}\) is independent of \(n\) and \(i\).
\end{proposition}

\begin{proof}
Let \(u_{n+2}\in\mathcal S_{n+2}(K)\). We shall prove that, for every \(i\in\mathbb N_0\),
\begin{equation}\label{EqM2Goal}
\frac{n(n+1)}{n+2i}u_{n+2,i}
\end{equation}
satisfies the estimate prescribed for the \((i+2)\)-th component of
\(\mathcal S_n(K)\). Since the definition of \(\mathcal S_n(K)\) depends on the value of $n$, we distinguish the three cases
\[
n=2,\qquad n=4,\qquad n\ge6.
\]

\begin{itemize}
\item First, we consider \(n\ge6\): for \(i=0\), the definition \eqref{Sn} gives
\[
|u_{n+2,0}|
\le
\frac{K^{\frac n2-\frac12}}{2(n+2)^2}.
\]
Consequently,
\begin{equation}\label{b0}
\left|
\frac{n(n+1)}{n}\,u_{n+2,0}
\right|
\le
\frac{n+1}{2(n+2)^2}
K^{\frac n2-\frac12}.
\end{equation}
On the other hand, the bound defining the second component $i=2$ of
\(\mathcal S_n(K)\) is
\[
K^{\frac n2+\frac12}
\frac{\left|\frac n4-1\right|!}{(2!)^{1/4}}.
\]
Since \(n\ge6\), one has
\begin{equation}\label{b1}
\left|\frac n4-1\right|!
=
\Gamma\left(\frac{n}{4}\right)\ge\Gamma\left(\frac{3}{2}\right)=\frac{\sqrt{\pi}}{2}.
\end{equation}
We also have for all $n\ge 6$ 
\begin{equation}\label{b2}
    \frac{n+1}{2(n+2)^2}\le \frac{7}{128}
\end{equation}
Combining \eqref{b1} and \eqref{b2} it follows for all $K\ge 1$ that
\begin{equation}\label{b00}
\frac{n+1}{2(n+2)^2}
\le
K~
\frac{\left|\frac n4-1\right|!}{(2!)^{1/4}}
\end{equation}
This together with \eqref{b0} implies that $$\left(\mathbb{M}_n^{(2)}\left(u_{n+2}\right)\right)_0\in\mathcal{S}_{n}^{(2)}(K)~.$$
For \(i=1\), we use \eqref{Sn} to write 
\[
|u_{n+2,1}|
\le
\frac{K^{\frac n2-\frac12}}{(n+2)^2}
\left(
1+\frac{(n+2)K}{2}
\right).
\]
Therefore,
\[
\left|
\frac{n(n+1)}{n+2}\,u_{n+2,1}
\right|
\le
\frac{n(n+1)}{(n+2)^3}
K^{\frac n2-\frac12}
\left(
1+\frac{(n+2)K}{2}
\right).
\]
Since \(K>1\), we  have
\[
1+\frac{(n+2)K}{2}
\le
\frac{n+4}{2}K~,
\]
and hence
\[
\left|
\frac{n(n+1)}{n+2}\,u_{n+2,1}
\right|
\le
\frac{n(n+1)(n+4)}{2(n+2)^3}
K^{\frac n2+\frac12}.
\]
The bound defining the third component of \(\mathcal S_n(K)\) is
\[
K^{\frac n2+\frac32}
\frac{\left|\frac n4\right|!}{(3!)^{1/4}}.
\]
Following similar steps that led to \eqref{b00} we deduce that 
$$\left(\mathbb{M}_n^{(2)}\left(u_{n+2}\right)\right)_1\in\mathcal{S}_{n}^{(3)}(K)~$$
holds for $K\ge 1$.

We now let \(i\ge2\). By definition,
\[
|u_{n+2,i}|
\le
K^{\frac n2+i-\frac12}
\frac{
\left|
\frac{n+2}{4}+i-3
\right|!
}
{(i!)^{1/4}}.
\]
Since \(n\ge6\) and \(i\ge2\),
\[
\frac{n+2}{4}+i-3
=
\frac n4+i-\frac52
\ge1,
\]
so the absolute value may be omitted. We obtain
\begin{align}
\left|
\frac{n(n+1)}{n+2i}u_{n+2,i}
\right|
&\le
\frac{n(n+1)}{n+2i}
K^{\frac n2+i-\frac12}
\frac{
\left(
\frac n4+i-\frac52
\right)!
}
{(i!)^{1/4}}
\nonumber\\
&=
\frac{1}{K}
\frac{n(n+1)}{n+2i}
\left(
\frac{(i+2)!}{i!}
\right)^{1/4}
\frac{
\left(
\frac n4+i-\frac52
\right)!
}{
\left(
\frac n4+i-1
\right)!
}
\nonumber\\
&\qquad\qquad\times
K^{\frac n2+i+\frac12}
\frac{
\left(
\frac n4+i-1
\right)!
}
{((i+2)!)^{1/4}}.
\label{EqM2GeneralRatio}
\end{align}
The numerical factor in the second line,
\[
\frac{n(n+1)}{n+2i}
\bigl((i+1)(i+2)\bigr)^{1/4}
\frac{
\Gamma\!\left(\frac n4+i-\frac32\right)
}{
\Gamma\!\left(\frac n4+i\right)
},
\]
is bounded uniformly over all even \(n\ge6\) and all \(i\ge2\). Since 
\[
\frac n4+i-\frac32\ge2,
\]
there exists\cite{NISTHandbook} a universal constant $C_{\Gamma}$ such that
\begin{equation}\label{GammaRatio}
\frac{
\Gamma\!\left(\frac n4+i-\frac32\right)
}{
\Gamma\!\left(\frac n4+i\right)
}
\le
C_\Gamma
\left(\frac n4+i\right)^{-3/2}.
\end{equation}
Furthermore,
\[
\frac{n(n+1)}{n+2i}
\le
\frac{2n^2}{n+i},\quad
\bigl((i+1)(i+2)\bigr)^{1/4}
\le
(n+i)^{1/2},\quad\frac n4+i
\gtrsim
n+i.
\]
Consequently, there exists a universal constant \(C>0\) such that
\[
\frac{n(n+1)}{n+2i}
\bigl((i+1)(i+2)\bigr)^{1/4}
\frac{
\Gamma\!\left(\frac n4+i-\frac32\right)
}{
\Gamma\!\left(\frac n4+i\right)
}
\le C
\]
uniformly in \(n\ge6\) and \(i\ge2\). Choosing
\(C_{\mathbb M}\ge C\) therefore yields the desired estimate whenever
\(K\ge C_{\mathbb M}\). Equation
\eqref{EqM2GeneralRatio} then yields
\[
\left|
\frac{n(n+1)}{n+2i}u_{n+2,i}
\right|
\le
K^{\frac n2+i+\frac12}
\frac{
\left(
\frac n4+i-1
\right)!
}
{((i+2)!)^{1/4}},
\]
which is precisely the bound defining the \((i+2)\)-th component of
\(\mathcal S_n(K)\).

\item Now we consider \(n=4\) and $n=2$. For $n=4$ we analyze $i=0$ and write 
\[
\left|
5u_{6,0}
\right|
\le
\frac{5K^{3/2}}{72}
\le
K^{5/2}\frac{0!}{(2!)^{1/4}}
\]
for \(K>1\). If \(i=1\), then
\[
\left|
\frac{4\cdot5}{6}u_{6,1}
\right|
\le
\frac{10}{3}
\frac{K^{3/2}}{36}(1+3K)
\le
K^{7/2}\frac{1!}{(3!)^{1/4}}
\]
for \(K\ge \max\left(1,C_{\mathbb M}\right)\). For \(i\ge2\), the
same computation as in \eqref{EqM2GeneralRatio}, with \(n=4\), gives
\[
\left|
\frac{20}{4+2i}u_{6,i}
\right|
\le
K^{i+\frac52}
\frac{i!}{((i+2)!)^{1/4}},
\]
provided \(K\ge C_{\mathbb M}\). This is the bound defining the
\((i+2)\)-th component of \(\mathcal S_4(K)\).

For \(n=2\), we have for \(i=0\),
\[
\left|
\frac{2\cdot3}{2}u_{4,0}
\right|
\le
\frac{3\sqrt K}{32}
\le
K^{5/2}
\frac{1!}{(1!)^{1/4}},
\]
which is the bound defining the second component of
\(\mathcal S_2(K)\). If \(i=1\),
\[
\left|
\frac{2\cdot3}{4}u_{4,1}
\right|
\le
\frac{3K}{64}
\le
\frac{K^{7/2}}{2^{\frac{1}{4}}},
\]
which is the bound defining the third component of
\(\mathcal S_2(K)\). Finally, for \(i\ge2\),
\[
|u_{4,i}|
\le
K^{i+\frac12}
\frac{|i-2|!}{(i!)^{1/4}},
\]
and therefore
\[
\left|
\frac{6}{2+2i}u_{4,i}
\right|
\le
\frac{3}{i+1}
K^{i+\frac12}
\frac{|i-2|!}{(i!)^{1/4}}.
\]
After factoring out the bound
\[
K^{i+\frac52}
\frac{|i-1|!}{((i+1)!)^{1/4}}
\]
of the \((i+2)\)-th component of \(\mathcal S_2(K)\), the remaining
coefficient is bounded uniformly in \(i\). Hence this estimate also
holds for \(K\ge C_{\mathbb M}\). 
\end{itemize}
This completes the proof of
\eqref{EqMapM2}.
\end{proof}
\begin{proposition}\label{PropMultiplicationOperators2}
There exists a constant \(C_{\mathbb M}>1\) such that, for every
\(K\ge C_{\mathbb M}\), every even integer \(n\ge2\), and every
\(i\in\mathbb N_0\), one has
\begin{equation}\label{EqMapM1}
\mathbb M_n^{(1)}
:
\mathcal S_n^{(i+1)}(K)
\longrightarrow
\mathcal S_n^{(i+2)}(K).
\end{equation}
The constant \(C_{\mathbb M}\) is independent of \(n\) and \(i\).
\end{proposition}
\begin{proof}
 We need to estimate
\[
\frac{n-4}{n+2i}~u_{n,i+1}~.
\]
\begin{itemize}
\item For \(n=4\), this expression vanishes identically, and there is nothing
to prove.

Suppose first that \(n\ge6\). If \(i=0\), then
\[
\left|
\frac{n-4}{n}u_{n,1}
\right|
\le
\frac{n-4}{n^3}
K^{\frac n2-\frac32}
\left(
1+\frac{nK}{2}
\right).
\]
Using \(K>1\), we obtain
\[
1+\frac{nK}{2}
\le
\frac{n+2}{2}K,
\]
and hence
\[
\left|
\frac{n-4}{n}u_{n,1}
\right|
\le
\frac{(n-4)(n+2)}{2n^3}
K^{\frac n2-\frac12}.
\]
Since \(n\ge6\),
\[
\frac{(n-4)(n+2)}{2n^3}
\le
\frac1{2n}
\le
\frac1{12}
\le
\frac{\Gamma(n/4)}{2^{1/4}}.
\]
Since \(K\ge1\), it follows that
\[
\frac{(n-4)(n+2)}{2n^3}
K^{\frac n2-\frac12}
\le
K^{\frac n2+\frac12}
\frac{\Gamma(n/4)}{(2!)^{1/4}},
\]
which is precisely the bound defining
\(\mathcal S_n^{(2)}(K)\).

Let now \(i\ge1\). Since \(i+1\ge2\),
\[
|u_{n,i+1}|
\le
K^{\frac n2+i-\frac12}
\frac{
\left|
\frac n4+i-2
\right|!
}
{((i+1)!)^{1/4}}.
\]
For \(n\ge6\) and \(i\ge1\), one has
\[
\frac n4+i-2\ge\frac12,
\]
so that
\begin{align}
\left|
\frac{n-4}{n+2i}u_{n,i+1}
\right|
&\le
\frac{n-4}{n+2i}
K^{\frac n2+i-\frac12}
\frac{
\left(
\frac n4+i-2
\right)!
}
{((i+1)!)^{1/4}}
\nonumber\\
&=
\frac1K
\frac{n-4}{n+2i}
(i+2)^{1/4}
\frac{
\left(
\frac n4+i-2
\right)!
}{
\left(
\frac n4+i-1
\right)!
}
\nonumber\\
&\qquad\qquad\times
K^{\frac n2+i+\frac12}
\frac{
\left(
\frac n4+i-1
\right)!
}
{((i+2)!)^{1/4}}.
\end{align}
Using
\[
\frac{
\left(
\frac n4+i-2
\right)!
}{
\left(
\frac n4+i-1
\right)!
}
=
\frac{1}{\frac n4+i-1},
\]
as a consequence of 
\begin{equation*}\label{EqGammaRecurrence}
\Gamma(z+1)=z\,\Gamma(z),
\qquad z>0.
\end{equation*}
we obtain
\[
\frac{n-4}{n+2i}
(i+2)^{1/4}
\frac{1}{\frac n4+i-1}
\le1.
\]
Since \(K>1\), it follows that
\[
\left|
\frac{n-4}{n+2i}u_{n,i+1}
\right|
\le
K^{\frac n2+i+\frac12}
\frac{
\left(
\frac n4+i-1
\right)!
}
{((i+2)!)^{1/4}}.
\]

\item It remains to consider \(n=2\). For \(i=0\),
\[
\left|
\frac{2-4}{2}u_{2,1}
\right|
=
|u_{2,1}|
\le
\frac K2
\le
K^{5/2},
\]
which is the bound defining the second component of
\(\mathcal S_2(K)\). For \(i\ge1\),
\[
\left|
\frac{2-4}{2+2i}u_{2,i+1}
\right|
=
\frac{1}{i+1}|u_{2,i+1}|.
\]
Using the definition of \(\mathcal S_2(K)\), we obtain
\[
\left|
\frac{2-4}{2+2i}u_{2,i+1}
\right|
\le
\frac{1}{i+1}
K^{i+\frac32}
\frac{|i-2|!}{(i!)^{1/4}}.
\]
The bound defining the \((i+2)\)-th component is
\[
K^{i+\frac52}
\frac{|i-1|!}{((i+1)!)^{1/4}}.
\]
The quotient of the former coefficient by the latter is
\[
\frac1K
\frac{1}{i+1}
(i+1)^{1/4}
\frac{|i-2|!}{|i-1|!},
\]
which is at most \(1\) for every \(i\ge1\) and every \(K>1\).
This proves \eqref{EqMapM1} and completes the proof.
\end{itemize}
\end{proof}
\begin{proposition}\label{PropMultiplicationOperators2Tilde}
There exists a constant \(C_{\mathbb M}>1\) such that, for every
\(K\ge C_{\mathbb M}\), every even integer \(n\ge2\), and every
\(i\in\mathbb N_0\), one has
\begin{equation}\label{EqMapM1}
\tilde{\mathbb M}_n^{(1)}
:
\mathcal S_n^{(i+1)}(K)
\longrightarrow
\mathcal S_n^{(i+2)}(K).
\end{equation}
The constant \(C_{\mathbb M}\) is independent of \(n\) and \(i\).
\end{proposition}
\begin{proof}
For \(u_n\in\mathcal S_n(K)\), we have
\[
\bigl(
\widetilde{\mathbb M}_n^{(1)}
\circ\mathbb T(u_n)
\bigr)_i
=
\frac{n-2}{n+2i}\,u_{n,i+1}.
\]
For \(n\ge6\), the proof is identical to that of
Proposition~\ref{PropMultiplicationOperators2}, with \(n-4\) replaced
by \(n-2\). Indeed, the only modification in the estimate for the
zeroth component is
\[
\frac{(n-2)(n+2)}{2n^3}
=
\frac{n^2-4}{2n^3}
\le
\frac{1}{2n},
\]
while the argument for \(i\ge1\) is unchanged after replacing the
factor \(n-4\) by \(n-2\). Thus
\[
\frac{n-2}{n+2i}\,u_{n,i+1}
\in
\mathcal S_n^{(i+2)}(K),
\qquad n\ge6.
\]
For \(n=2\), the multiplier vanishes identically, so there is nothing
to prove.\\
It remains to consider \(n=4\). In this case,
\[
\frac{n-2}{n+2i}
=
\frac{1}{i+2}.
\]
For \(i=0\), the definition of \(\mathcal S_4(K)\) gives
\[
\left|
\frac12 u_{4,1}
\right|
\le
\frac{K}{64}
\le
K^{5/2}\frac{0!}{(2!)^{1/4}},
\]
for \(K\ge1\), which is precisely the bound defining
\(\mathcal S_4^{(2)}(K)\). Let now \(i\ge1\). Since \(i+1\ge2\),
\[
|u_{4,i+1}|
\le
K^{i+\frac32}
\frac{(i-1)!}{((i+1)!)^{1/4}}.
\]
Hence
\begin{align*}
\left|
\frac{1}{i+2}u_{4,i+1}
\right|
&\le
\frac{1}{i+2}
K^{i+\frac32}
\frac{(i-1)!}{((i+1)!)^{1/4}}
\\
&=
\frac{1}{K\,i(i+2)^{3/4}}
K^{i+\frac52}
\frac{i!}{((i+2)!)^{1/4}}
\\
&\le
K^{i+\frac52}
\frac{i!}{((i+2)!)^{1/4}}.
\end{align*}
The last expression is exactly the bound defining
\(\mathcal S_4^{(i+2)}(K)\). This proves
\eqref{EqMapM1} and completes the proof.
\end{proof}
\subsection{Mapping properties of the bilinear operators}

We now turn to the bilinear operators introduced in the previous
subsection. We begin with the operator \(\mathbb B_n^{\alpha}\), where
\(\alpha\) denotes an admissible coefficient satisfying
\eqref{EqAlphaBound}. Since all the estimates below are uniform over the
choice of \(\alpha\), we suppress the superscript and simply write
\[
\mathbb B_n
:=
\mathbb B_n^{\alpha},
\qquad
\mathbb B_{2,n}
:=
\mathbb B_{2,n}^{\alpha}.
\]

\begin{proposition}\label{LemOperatorBnPlusTwo}
There exists a constant \(C_{\mathbb B}>1\), independent of \(n\), \(K\),
and of the admissible coefficient family \(\alpha\), such that, for
every even integer \(n\ge6\) and every
\(K\ge C_{\mathbb B}^{\,2}\),
\[
\mathbb B_n
\left(
\widehat{\mathcal S}_{n+2}(K)
\times
\widehat{\mathcal S}_{n+2}(K)
\right)
\subset
\mathcal S_n(K).
\]
Equivalently, for every
\(u,w\in\widehat{\mathcal S}_{n+2}(K)\) and every
\(i\in\mathbb N_0\),
\[
\left(
\mathbb T^2\circ\mathbb B_n(u,w)
\right)_i
\in
\mathcal S_n^{(i+2)}(K).
\]
\end{proposition}
\begin{remark}
    The equivalence of the two formulations follows from the fact that
admissibility gives
\[
(\mathbb B_n(u,w))_0
=
(\mathbb B_n(u,w))_1
=
0.
\]
\end{remark}
\begin{remark}
The coefficient families appearing in the higher-point bilinear terms
of the mean-field hierarchy are admissible in the sense of
Definition~\ref{Admissible}. Indeed, we set their components of order
\(0\) and \(1\) equal to zero, while for \(r\ge2\) and
\(n_1+n_2=n+2\), \(n_1,n_2\ge4\), one has
\[
\left|
\alpha_{gg;n_1,n}^{(r)}
\right|
=
\frac{1}{n+2r-4}
\frac{(n_1-1)(n_2-1)}{n-1}
\le
\frac{n}{n+2r-4}.
\]
Moreover,
\[
\left|
\alpha_{kg;n_1,n}^{(r)}
\right|
=
\frac{n_1-1}{n+2r-4}
\left(
1+\frac{n_2-1}{n-1}
\right)
\le
\frac{n}{n+2r-4},
\]
where we used
\[
(n_1-1)
+
\frac{(n_1-1)(n_2-1)}{n-1}
\le
(n_1-1)+(n_2-1)
=
n.
\]
Finally,
\[
\left|
\beta_{gg;n_1,n}^{(r)}
\right|
=
\frac{n}{n+2r-4}.
\]
Thus all coefficient families occurring in
\(\mathbb T^2\circ\mathbb B_n\) satisfy
\eqref{EqAlphaBound}, and Proposition~\ref{LemOperatorBnPlusTwo}
applies to each of the higher-point bilinear terms in the mean-field
hierarchy.
\end{remark}
\noindent The proof relies on the combinatorial estimates of Kopper and Wang
\cite{KopperWang2025}, collected in Appendix~\ref{appendix}.

\begin{proof}
Let \(u,w\in\widehat{\mathcal S}_{n+2}(K)\), and fix
\(i\in\mathbb N_0\). Since
\[
\bigl(
\mathbb T^2\circ\mathbb B_n(u,w)
\bigr)_i
=
\bigl(
\mathbb B_n(u,w)
\bigr)_{i+2},
\]
the admissibility condition \eqref{EqAlphaBound} yields
\begin{align}
\left|
\bigl(
\mathbb T^2\circ\mathbb B_n(u,w)
\bigr)_i
\right|
&\le
\frac{n}{n+2i}
\sum_{\substack{
n_1+n_2=n+2\\
n_1,n_2\in2\mathbb N,\; n_1,n_2\ge4
}}
\sum_{j=0}^{i+2}
|u_{n_1,j}|\,
|w_{n_2,i+2-j}|.
\label{EqProofBnInitial}
\end{align}
We shall prove that the right-hand side is bounded by
\begin{equation}\label{EqProofBnTarget}
K^{\frac n2+i+\frac12}
\frac{
\left(
\frac n4+i-1
\right)!
}{
[(i+2)!]^{1/4}
},
\end{equation}
which is precisely the bound defining the \((i+2)\)-th component of
\(\mathcal S_n(K)\).

Since the summation in \eqref{EqProofBnInitial} is over ordered pairs,
we may restrict attention to \(n_1\le n_2\) at the price of a factor
\(2\). Indeed, the term corresponding to \((n_2,n_1)\) is estimated in
exactly the same way after replacing \(j\) by \(i+2-j\). It therefore
suffices to bound
\begin{equation}\label{EqProofBnOrdered}
\frac{2n}{n+2i}
\sum_{\substack{
n_1+n_2=n+2\\
4\le n_1\le n_2
}}
\sum_{j=0}^{i+2}
|u_{n_1,j}|\,
|w_{n_2,i+2-j}|.
\end{equation}
The summation is separated according to whether \(n_1\le10\) or
\(n_1\ge12\):
\begin{itemize}
\item\underline{ \(n_1\ge12\):} for every even integer \(m\ge12\) and every \(r\in\mathbb N_0\), the
definition of \(\mathcal S_m(K)\) implies
\begin{equation}\label{EqGenericBoundLargeN}
|u_{m,r}|
\le
K^{\frac m2+r-\frac32}
\frac{
\left(
\frac m4+r-3
\right)!
}{
(r!)^{1/4}
}.
\end{equation}
For \(r\ge2\), this is the defining estimate. For \(r=0\), it follows
from the defining estimate \eqref{Sn} together with 
\[
\frac1{2m^2}
\le
\left(\frac m4-3\right)!,
\]
whereas, for \(r=1\), the inequality \(K>1\) gives
\[
\frac1{m^2K}+\frac1{2m}
\le
\left(\frac m4-2\right)!.
\]
Thus \eqref{EqGenericBoundLargeN} holds for all \(r\ge0\). Applying \eqref{EqGenericBoundLargeN} to both factors, we obtain
\begin{align}
&\frac{2n}{n+2i}
\sum_{\substack{
n_1+n_2=n+2\\
12\le n_1\le n_2
}}
\sum_{j=0}^{i+2}
|u_{n_1,j}|\,
|w_{n_2,i+2-j}|
\nonumber\\
&\qquad\le
2K^{\frac n2+i}
\frac{n}{n+2i}
\sum_{\substack{
n_1+n_2=n+2\\
12\le n_1\le n_2
}}
F\left(
n_1,n_2,i,0,0,\frac14
\right).
\end{align}
Indeed, since \(n_1+n_2=n+2\), the product of the powers of \(K\) is
\[
K^{\frac{n_1}{2}+j-\frac32}
K^{\frac{n_2}{2}+i+2-j-\frac32}
=
K^{\frac n2+i}.
\]
By \eqref{EqF00},
\begin{align}
&\frac{2n}{n+2i}
\sum_{\substack{
n_1+n_2=n+2\\
12\le n_1\le n_2
}}
\sum_{j=0}^{i+2}
|u_{n_1,j}|\,
|w_{n_2,i+2-j}|\le
2C_0~K^{\frac n2+i}
\frac{
\left(
\frac n4+i-1
\right)!
}{
[(i+2)!]^{1/4}
}.
\label{EqProofBnLargeSector}
\end{align}

\medskip

\item\underline{The range \(4\le n_1\le10\):}
We decompose the interior part of the convolution sum into its endpoint and interior
contributions:
\[
\sum_{j=0}^{i+2}
=
\sum_{j\in\{0,1,i+1,i+2\}}
+
\sum_{j=2}^{i},
\]
where repeated indices are counted only once. 

The second sum is empty
when \(i\in\{0,1\}\). In these two cases, the convolution consists
entirely of endpoint contributions, which are estimated separately
below.

If \(i\ge2\),
then, for \(2\le j\le i\), both indices \(j\) and \(i+2-j\) are at
least \(2\). Hence the factorial estimates in the definitions of the
spaces may be applied directly, and
\begin{align}
&\frac{2n}{n+2i}
\sum_{\substack{
n_1+n_2=n+2\\
4\le n_1\le10\\
n_1\le n_2
}}
\sum_{j=2}^{i}
|u_{n_1,j}|\,
|w_{n_2,i+2-j}|\le
2K^{\frac n2+i}
\frac{n}{n+2i}
\sum_{\substack{
n_1+n_2=n+2\\
4\le n_1\le10\\
n_1\le n_2
}}
F\left(
n_1,n_2,i,2,2,\frac14
\right).
\end{align}
The estimate \eqref{EqF22} therefore gives
\begin{align}
&\frac{2n}{n+2i}
\sum_{\substack{
n_1+n_2=n+2\\
4\le n_1\le10\\
n_1\le n_2
}}
\sum_{j=2}^{i}
|u_{n_1,j}|\,
|w_{n_2,i+2-j}|
\le
2C_1K^{\frac n2+i}
\frac{
\left(
\frac n4+i-1
\right)!
}{
[(i+2)!]^{1/4}
}.
\label{EqProofBnInteriorSector}
\end{align}
It remains to estimate the four endpoint terms
\[
j=0,\qquad j=1,\qquad j=i+1,\qquad j=i+2.
\]
In the sequel, we prove that there exists a universal constant \(C_{\partial}>0\)
such that
\begin{align}
\frac{2n}{n+2i}
\sum_{\substack{
n_1+n_2=n+2\\
4\le n_1\le10\\
n_1\le n_2
}}
\sum_{j\in\{0,1,i+1,i+2\}}
|u_{n_1,j}|\,
|w_{n_2,i+2-j}|
\le
C_{\partial}K^{\frac n2+i}
\frac{
\left(
\frac n4+i-1
\right)!
}{
[(i+2)!]^{1/4}
}.
\label{EqProofBnBoundarySector}
\end{align}
 Since \(n_1\in\{4,6,8,10\}\), there are only
finitely many possible values of the first index. For \(j=0\), the
definitions give
\[
|u_{n_1,0}|
\lesssim
K^{\frac{n_1}{2}-\frac32}.
\]
Moreover, the bound for the \((i+2)\)-th component of \(w_{n_2}\)
implies
\[
|w_{n_2,i+2}|
\lesssim
K^{\frac{n_2}{2}+i+\frac12}
\frac{
\left(
\frac{n_2}{4}+i-1
\right)!
}{
[(i+2)!]^{1/4}
}.
\]
Since \(n_1+n_2=n+2\), the product of the powers of \(K\) is
\(K^{\frac n2+i}\). Moreover,
\[
\frac n4+i-1
-
\left(
\frac{n_2}{4}+i-1
\right)
=
\frac{n_1-2}{4}
\in
\left\{
\frac12,1,\frac32,2
\right\}.
\]
Lemma~\ref{LemGammaRatio} gives
\[
\frac{n}{n+2i}
\frac{
\left(
\frac{n_2}{4}+i-1
\right)!
}{
\left(
\frac n4+i-1
\right)!
}
\lesssim1,
\]
uniformly in \(n\) and \(i\). Hence the contribution corresponding to
\(j=0\) satisfies \eqref{EqProofBnBoundarySector}.

For $j=1$ and \(i=0\), the contribution follows
directly from the first-component bounds in the definitions of
\(\mathcal S_{n_1}(K)\) and \(\mathcal S_{n_2}(K)\). 
For \(j=1\) and $i\ge 1$, the defining bounds and \(K>1\) give
\[
|u_{n_1,1}|
\lesssim
K^{\frac{n_1}{2}-\frac12},
\]
and
\[
|w_{n_2,i+1}|
\lesssim
K^{\frac{n_2}{2}+i-\frac12}
\frac{
\left(
\frac{n_2}{4}+i-2
\right)!
}{
[(i+1)!]^{1/4}
}.
\]
Thus the product again carries the factor \(K^{\frac n2+i}\).
Furthermore,
\[
\frac n4+i-1
-
\left(
\frac{n_2}{4}+i-2
\right)
=
1+\frac{n_1-2}{4}
\in
\left\{
\frac32,2,\frac52,3
\right\}.
\]
Since
\[
\frac{[(i+2)!]^{1/4}}{[(i+1)!]^{1/4}}
=
(i+2)^{1/4},
\]
the standard Gamma-ratio estimate from Lemma \ref{LemGammaRatio} gives
\[
\frac{n}{n+2i}
(i+2)^{1/4}
\frac{
\left|
\frac{n_2}{4}+i-2
\right|!
}{
\left(
\frac n4+i-1
\right)!
}
\lesssim1,
\]
uniformly in \(n\) and \(i\). Hence the contribution corresponding to
\(j=1\) also satisfies \eqref{EqProofBnBoundarySector}.

It remains to consider the opposite endpoints
\(j=i+1\) and \(j=i+2\). Writing
\[
j=i+2-r,
\qquad
r\in\{0,1\},
\]
the defining bounds give
\[
|u_{n_1,i+2-r}|\,
|w_{n_2,r}|
\lesssim
K^{\frac n2+i}
\frac{
\left(
\frac{n_1}{4}+i-1-r
\right)!
}{
[(i+2-r)!]^{1/4}
}.
\]
To compare this with \eqref{EqProofBnTarget}, set
\[
x_r:=\frac{n_1}{4}+i-r,
\qquad
\delta_r:=\frac{n_2-2}{4}+r.
\]
Since \(n_1+n_2=n+2\), one has
\[
x_r+\delta_r=\frac n4+i.
\]
Consequently, Lemma~\ref{LemGammaRatio} yields
\[
\frac{
\Gamma(x_r)
}{
\Gamma\!\left(\frac n4+i\right)
}
=
\frac{\Gamma(x_r)}{\Gamma(x_r+\delta_r)}
\le
2x_r^{-\delta_r}.
\]
For \(r=0\), one has \(x_0\ge1\) and
\(\delta_0\ge\frac12\), so Lemma~\ref{LemGammaRatio} applies for every
\(i\in\mathbb N_0\). For \(r=1\), the case \(i=0\) coincides with the
endpoint \(j=1\) already treated above. If \(i\ge1\), then
\[
x_1\ge i,
\qquad
\delta_1\ge\frac32,
\]
and therefore
\[
(i+2)^{1/4}x_1^{-\delta_1}
\le
(i+2)^{1/4}i^{-3/2}
\le
3^{1/4}.
\]
Hence the required estimate holds uniformly for both
\(r\in\{0,1\}\) and we obtain in both cases,
\[
\frac{n}{n+2i}
|u_{n_1,i+2-r}|\,
|w_{n_2,r}|
\lesssim
K^{\frac n2+i}
\frac{
\left(
\frac n4+i-1
\right)!
}{
[(i+2)!]^{1/4}
}.
\]
Thus the contributions corresponding to \(j=i+1\) and \(j=i+2\)
satisfy \eqref{EqProofBnBoundarySector}.
\end{itemize}
Combining \eqref{EqProofBnLargeSector},
\eqref{EqProofBnInteriorSector}, and
\eqref{EqProofBnBoundarySector}, we obtain
\begin{align}
\left|
\bigl(
\mathbb T^2\circ\mathbb B_n(u,w)
\bigr)_i
\right|
&\le
C_*
K^{\frac n2+i}
\frac{
\left(
\frac n4+i-1
\right)!
}{
[(i+2)!]^{1/4}
},
\label{EqProofBnFinalIntermediate}
\end{align}
where
\[
C_*:=2C_0+2C_1+C_{\partial}
\]
is independent of \(n\), \(i\), \(K\), and of the admissible
coefficient \(\alpha\). Choose
\[
C_{\mathbb B}\ge\max\{1,C_*\}.
\]
If \(K\ge C_{\mathbb B}^{\,2}\), then
\(C_*\le K^{1/2}\), and \eqref{EqProofBnFinalIntermediate} yields
\[
\left|
\bigl(
\mathbb T^2\circ\mathbb B_n(u,w)
\bigr)_i
\right|
\le
K^{\frac n2+i+\frac12}
\frac{
\left(
\frac n4+i-1
\right)!
}{
[(i+2)!]^{1/4}
}.
\]
Thus
\[
\bigl(
\mathbb T^2\circ\mathbb B_n(u,w)
\bigr)_i
\in
\mathcal S_n^{(i+2)}(K),
\qquad
i\in\mathbb N_0.
\]
Equivalently,
\[
\bigl(\mathbb B_n(u,w)\bigr)_r
\in
\mathcal S_n^{(r)}(K),
\qquad r\ge2.
\]
Since admissibility implies
\[
\bigl(\mathbb B_n(u,w)\bigr)_0
=
\bigl(\mathbb B_n(u,w)\bigr)_1
=
0,
\]
we conclude that
\[
\mathbb B_n(u,w)\in\mathcal S_n(K).
\]
All estimates are uniform over the admissible coefficient families
\(\alpha\), which concludes the proof.

\end{proof}
Next we consider the bilinear terms involving the two-point sector.
After one forward shift, the \(i\)-th component involves a convolution
of order \(i+1\), which we estimate in
\(\mathcal S_n^{(i+2)}(K)\).
\begin{proposition}\label{LemOperatorBnPlus2}
There exists a constant \(C_{\mathbb B}>1\) with the following
property. Let \(n\ge2\) be even and let $
\alpha_{2,n}
=
\bigl(\alpha_{2,n}^{(r)}\bigr)_{r\ge1}$
be a coefficient family satisfying
\begin{equation}\label{EqAlphaB2nShifted}
\left|
\alpha_{2,n}^{(i+1)}
\right|
\le
\frac{n}{n+2i},
\qquad
i\in\mathbb N_0.
\end{equation}
For \(n=2\), we take
\begin{equation}\label{EqAlpha22}
\alpha_{2,2}^{(r)}
:=
\frac{1}{r+1},
\qquad
r\in\mathbb N_0.
\end{equation}
Then, for every \(K\ge C_{\mathbb B}^{\,2}\),
every \(u\in\mathcal S_2(K)\), every
\(w\in\mathcal S_n(K)\), and every \(i\in\mathbb N_0\),
\begin{equation}\label{EqMapB2nComponent}
\left(
\mathbb T\circ
\mathbb B_{2,n}^{\alpha}(u,w)
\right)_i
\in
\mathcal S_n^{(i+2)}(K).
\end{equation}
The constant \(C_{\mathbb B}\) is independent of \(n\), \(i\), \(K\),
and of the coefficient family satisfying
\eqref{EqAlphaB2nShifted}.
\end{proposition}
\begin{remark}
    Note that the coefficient families occurring in the mean-field hierarchy satisfy
\eqref{EqAlphaB2nShifted}. For \(n\ge4\), we have 
\[
\left|
\alpha_{g;2,n}^{(i+1)}
\right|
=
\frac{1}{n+2i}
\le
\frac{n}{n+2i}~,
\]
while
\[
\left|
\alpha_{k;2,n}^{(i+1)}
\right|
=
\left|
\beta_{gf;2,n}^{(i+1)}
\right|
=
\frac{n}{n+2i}~.
\]
For \(n=2\), \eqref{EqAlpha22} gives
\[
\left|
\alpha_{2,2}^{(i+1)}
\right|
=
\frac{1}{i+2}
\le
\frac{1}{i+1}
=
\frac{2}{2+2i}.
\]
Thus Proposition~\ref{LemOperatorBnPlus2} applies to all bilinear terms
involving the two-point sector.
\end{remark}
\begin{proof}
Let \(u\in\mathcal S_2(K)\), \(w\in\mathcal S_n(K)\), and fix
\(i\in\mathbb N_0\). By definition,
\[
\left(
\mathbb T\circ\mathbb B_{2,n}^{\alpha}(u,w)
\right)_i
=
\alpha_{2,n}^{(i+1)}
\sum_{j=0}^{i+1}
u_{2,j}w_{n,i+1-j}.
\]
Hence \eqref{EqAlphaB2nShifted} gives
\begin{equation}\label{EqB2nInitialEstimate}
\left|
\left(
\mathbb T\circ\mathbb B_{2,n}^{\alpha}(u,w)
\right)_i
\right|
\le
\frac{n}{n+2i}
\sum_{j=0}^{i+1}
|u_{2,j}|\,|w_{n,i+1-j}|.
\end{equation}

We first consider \(n\ge6\). The cases \(i=0,1,2\) involve only
finitely many convolution terms and follow directly from the defining
bounds of \(\mathcal S_2(K)\) and \(\mathcal S_n(K)\), after
increasing \(C_{\mathbb B}\) if necessary. We therefore assume
\(i\ge3\) and write
\[
\sum_{j=0}^{i+1}
=
\sum_{j\in\{0,1,i,i+1\}}
+
\sum_{j=2}^{i-1}.
\]
We start with the four endpoint contributions. From
\eqref{S2}--\eqref{Sn}, we obtain that
\begin{align}
\frac{n}{n+2i}|u_{2,0}w_{n,i+1}|
&\le
\frac{K^{\frac n2+i}}{4}
\frac{n}{n+2i}
\frac{\left(\frac n4+i-2\right)!}
{[(i+1)!]^{1/4}}
\nonumber\\
&=
\left(\frac{n}{4(n+2i)}
\frac{(i+2)^{1/4}}{\frac n4+i-1}\right)K^{\frac n2+i}
\frac{\left(\frac n4+i-1\right)!}
{[(i+2)!]^{1/4}},
\label{EqEndpointJ0}
\end{align}
and 
\begin{align}
\frac{n}{n+2i}|u_{2,1}w_{n,i}|
&\le
\frac{K^{\frac n2+i-\frac12}}{2}
\frac{n}{n+2i}
\frac{\left(\frac n4+i-3\right)!}
{(i!)^{1/4}}
\nonumber\\
&=
\left(
\frac{n}{2(n+2i)}
\frac{[(i+1)(i+2)]^{1/4}}
{\left(\frac n4+i-2\right)
 \left(\frac n4+i-1\right)}\right)~
K^{\frac n2+i-\frac12}\frac{\left(\frac n4+i-1\right)!}
{[(i+2)!]^{1/4}}.
\label{EqEndpointJ1}
\end{align}
For $n\ge 6$ and $j\ge 1$, one has
\begin{equation}
 \frac{(j+2)^{1/4}}{\frac n4+j-1}
\le
\frac{2}{(j+2)^{3/4}}
\le 1.
\end{equation}
Using this bound with $j=i$ and $j=i-1$, the numerical factors in
\eqref{EqEndpointJ0}--\eqref{EqEndpointJ1}
are bounded uniformly for
\(n\ge6\) and \(i\ge3\). For the opposite endpoints, using \(K>1\) and
\[
1+\frac{nK}{2}
\le
\frac{n+2}{2}K,
\]
we obtain
\begin{align}
\frac{n}{n+2i}|u_{2,i}w_{n,1}|
&\le
\frac{n+2}{2n(n+2i)}
K^{\frac n2+i}
\frac{(i-3)!}{[(i-1)!]^{1/4}},
\label{EqEndpointJi}
\\
\frac{n}{n+2i}|u_{2,i+1}w_{n,0}|
&\le
\frac{1}{2n(n+2i)}
K^{\frac n2+i}
\frac{(i-2)!}{(i!)^{1/4}}.
\label{EqEndpointJiPlusOne}
\end{align}
After factoring out
\[
K^{\frac n2+i}
\frac{\left(\frac n4+i-1\right)!}
{[(i+2)!]^{1/4}}
\] in \eqref{EqEndpointJi}-\eqref{EqEndpointJiPlusOne}, the remaining Gamma ratios are
\[
\frac{\Gamma(i-2)}
{\Gamma(\frac n4+i)}
\qquad\text{and}\qquad
\frac{\Gamma(i-1)}
{\Gamma(\frac n4+i)},
\]
whose respective gaps are $
\frac n4+2$ and $\frac n4+1$. Lemma~\ref{LemGammaRatio} therefore gives a uniform bound. Consequently,
there exists \(C_{\partial}>0\), independent of \(n\), \(i\), and
\(K\), such that
\begin{equation}\label{EqEndpointContributions}
\frac{n}{n+2i}
\sum_{j\in\{0,1,i,i+1\}}
|u_{2,j}|\,|w_{n,i+1-j}|
\le
C_{\partial}
K^{\frac n2+i}
\frac{\left(\frac n4+i-1\right)!}
{[(i+2)!]^{1/4}}.
\end{equation}
It remains to consider the interior convolution. We would like to establish that 
\begin{equation}\label{EqB2nInteriorConvolution}
\frac{n}{n+2i}
\sum_{j=2}^{i-1}
\frac{|j-3|!}{[(j-1)!]^{1/4}}
\frac{
\left|
\frac n4+i-j-2
\right|!
}{
[(i+1-j)!]^{1/4}
}
\le
C~
\frac{
\left(
\frac n4+i-1
\right)!
}{
[(i+2)!]^{1/4}},
\end{equation}
with \(C\) universal. Separating the term \(j=2\) from the sum gives
\[
\sum_{j=3}^{i-1}
\frac{
(j-3)!
\left|
\frac n4+i-j-2
\right|!
}{
[(j-1)!(i+1-j)!]^{1/4}
}
=
F\left(4,n,i-2,2,2,\frac14\right).
\]
Applying \eqref{EqF22} with its parameter \(n\) replaced by \(n+2\)
and \(k=i-2\), we obtain
\[
\begin{aligned}
\frac{n}{n+2i}
~F\left(4,n,i-2,2,2,\frac14\right)
&\le
C_1~
\frac{n(n+2i-2)}
{(n+2i)(n+2)}
\frac{
\left(
\frac n4+i-\frac52
\right)!
}{
(i!)^{1/4}
}
\\
&\le
C_1~
\frac{
\left(
\frac n4+i-\frac52
\right)!
}{
(i!)^{1/4}
},
\end{aligned}
\]
Furthermore, Lemma~\ref{LemGammaRatio} yields
\[
\bigl((i+1)(i+2)\bigr)^{1/4}
\frac{
\left(
\frac n4+i-\frac52
\right)!
}{
\left(
\frac n4+i-1
\right)!
}
\le C,
\]
uniformly in \(n\ge6\) and \(i\ge3\). Hence
\[
\frac{
\left(
\frac n4+i-\frac52
\right)!
}{
(i!)^{1/4}}
\le
C
\frac{
\left(
\frac n4+i-1
\right)!
}{
[(i+2)!]^{1/4}}.
\]
The contribution corresponding to \(j=2\) is treated separately. Using
\eqref{EqGammaIterated}, we have
\[
\frac{
\left(\frac n4+i-4\right)!
}{
\left(\frac n4+i-1\right)!
}
=
\frac{1}{
\left(\frac n4+i-3\right)
\left(\frac n4+i-2\right)
\left(\frac n4+i-1\right)}.
\]
Consequently, we obtain that 
\[
\begin{aligned}
\frac{n}{n+2i}
\frac{
\left(\frac n4+i-4\right)!
}{
[(i-1)!]^{1/4}
}
&=
\frac{n}{n+2i}
\frac{[i(i+1)(i+2)]^{1/4}}
{
\left(\frac n4+i-3\right)
\left(\frac n4+i-2\right)
\left(\frac n4+i-1\right)
}
\frac{
\left(\frac n4+i-1\right)!
}{
[(i+2)!]^{1/4}
}
\\
&\le
C
\frac{
\left(\frac n4+i-1\right)!
}{
[(i+2)!]^{1/4}
},
\end{aligned}
\]
since the prefactor is bounded uniformly for \(n\ge6\) and \(i\ge3\). This
proves \eqref{EqB2nInteriorConvolution}.

Combining \eqref{EqEndpointContributions} and
\eqref{EqB2nInteriorConvolution}, and including the finitely many
cases \(i=0,1,2\), we obtain a constant \(C>0\), independent of
\(n\), \(i\), and \(K\), such that
\begin{equation}\label{EqB2nEstimateNge6}
\frac{n}{n+2i}
\sum_{j=0}^{i+1}
|u_{2,j}|\,|w_{n,i+1-j}|
\le
C K^{\frac n2+i}
\frac{
\left(
\frac n4+i-1
\right)!
}{
[(i+2)!]^{1/4}}.
\end{equation}
We next consider \(n=4\). The same decomposition, with the weights
defining \(\mathcal S_4(K)\), gives
\begin{equation}\label{EqB2nEstimateN4}
\frac{4}{4+2i}
\sum_{j=0}^{i+1}
|u_{2,j}|\,|w_{4,i+1-j}|
\le
C K^{i+2}
\frac{i!}{[(i+2)!]^{1/4}},
\qquad i\in\mathbb N_0,
\end{equation}
where \(C\) is independent of \(i\) and \(K\). 

Finally, let \(n=2\). In this case
\[
\left|
\alpha_{2,2}^{(i+1)}
\right|
=
\frac{1}{i+2}
\le
\frac{1}{i+1}
=
\frac{2}{2+2i}.
\]
Using the defining bounds of \(\mathcal S_2(K)\), and separating the
four endpoint terms from the interior convolution as above, one obtains
\begin{equation}\label{EqB22Estimate}
\frac{1}{i+1}
\sum_{j=0}^{i+1}
|u_{2,j}|\,|w_{2,i+1-j}|
\le
C K^{i+2}
\frac{|i-1|!}{[(i+1)!]^{1/4}},
\qquad i\in\mathbb N_0.
\end{equation}
Indeed, the finitely many low-index terms are estimated directly,
while for the interior terms the Gamma recurrence gives the elementary
convolution bound
\[
\frac{1}{i+1}
\sum_{j=2}^{i-1}
\frac{|j-3|!}{[(j-1)!]^{1/4}}
\frac{|i-j-2|!}{[(i-j)!]^{1/4}}
\le
C
\frac{|i-1|!}{[(i+1)!]^{1/4}}.
\]

We may now choose \(C_{\mathbb B}>1\) larger than all constants
appearing above. If \(K\ge C_{\mathbb B}^{\,2}\), then
\(C\le K^{1/2}\). Therefore
\eqref{EqB2nEstimateNge6}, \eqref{EqB2nEstimateN4}, and
\eqref{EqB22Estimate} yield, respectively, the defining bounds of
\[
\mathcal S_n^{(i+2)}(K),
\qquad
n\ge6,
\]
\[
\mathcal S_4^{(i+2)}(K),
\]
and
\[
\mathcal S_2^{(i+2)}(K).
\]
Consequently,
\[
\left(
\mathbb T\circ\mathbb B_{2,n}^{\alpha}(u,w)
\right)_i
\in
\mathcal S_n^{(i+2)}(K)
\]
for every even \(n\ge2\) and every \(i\in\mathbb N_0\), which
concludes the proof.
\end{proof}
\subsection{Proof of Theorem~\ref{ThmCoefficientHierarchy}}

We now prove Theorem~\ref{ThmCoefficientHierarchy}. The argument combines
the mapping properties established in the previous subsection with the
recursive structure of the discrete mean-field hierarchy.

The coefficients of orders \(0\) and \(1\) are treated separately.
Indeed, the corresponding bounds in the definition of
\(\mathcal S_2(K)\) and \(\mathcal S_n(K)\) are stronger than the
factorial estimates used for orders \(i\ge2\), and therefore do not
follow directly from the mapping properties. We first establish the
following low-order estimate.

\begin{proposition}[Low-order stability]
\label{ThmCoefficientHierarchy0}
Let
\[
(f_{2,i})_{i\in\mathbb N_0},
\qquad
(h_{2,i})_{i\in\mathbb N_0},
\qquad
(g_{n,i})_{\substack{n\ge4,\;n\in2\mathbb N\\ i\in\mathbb N_0}},
\qquad
(k_{n,i})_{\substack{n\ge4,\;n\in2\mathbb N\\ i\in\mathbb N_0}}
\]
be a solution of the discrete mean-field hierarchy
\eqref{eq:f2_recursion}--\eqref{eq:k_compatibility_one}, satisfying
the boundary conditions
\eqref{EqMeanFieldBC}--\eqref{EqMeanFieldBC2}. Assume that \(K>1\) is
sufficiently large and that
\begin{equation}\label{EqInitialBounds0}
|f_{2,0}|
\le
\frac{\sqrt K}{4},
\qquad
|h_{2,0}|
\le
\frac{\sqrt K}{16},
\qquad
|g_{4,0}|
\le
\frac{\sqrt K}{32}.
\end{equation}
Then, for \(i\in\{0,1\}\),
\[
f_{2,i}\in\mathcal S_2^{(i)}(K),
\qquad
h_{2,i}\in\mathcal S_2^{(i)}(K),
\]
and, for every even integer \(n\ge4\),
\[
g_{n,i}\in\mathcal S_n^{(i)}(K),
\qquad
k_{n,i}\in\mathcal S_n^{(i)}(K).
\]
\end{proposition}

We shall
repeatedly use the compatibility relations derived in
Section~\ref{SecDiscreteHierarchy} by identifying equal powers of
\(\mu\). These relations will be recalled, in a form adapted to the
estimates, whenever they are needed.

The boundary condition \eqref{EqMeanFieldBC2} gives
\[
k_4(0)=0,
\]
and therefore
\[
k_{4,0}=0.
\]
This identity provides the initial value for the induction at
order \(i=0\).
\begin{proof}
The proof is divided into three steps. We first control the two-point
coefficients, then establish the zeroth-order estimates by simultaneous
induction on the even integer \(n\ge4\), and finally proceed in the
same way for the first-order coefficients.

Throughout the proof, \(C>0\) denotes a numerical constant whose value
may change from one occurrence to the next. We shall repeatedly use
the elementary estimates
\begin{equation}\label{EqElementaryDiscreteSums}
\sum_{\substack{n_1+n_2=n+2\\ n_1,n_2\ge4}}
\frac{1}{n_1^2n_2^2}
\le
\frac{C}{n^2},
\qquad
\sum_{\substack{n_1+n_2=n+2\\ n_1,n_2\ge4}}
\frac{1}{n_1n_2^2}
\le
\frac{C}{n},
\end{equation}
together with
\begin{equation}\label{EqElementaryDiscreteSumsTwo}
\sum_{\substack{n_1+n_2=n+2\\ n_1,n_2\ge4}}
\frac{1}{n_1n_2}
\le
C~\frac{\log n}{n}.
\end{equation}
The same bound as in the second estimate of
\eqref{EqElementaryDiscreteSums} holds with \(n_1\) and \(n_2\)
interchanged. These estimates follow by splitting the summation domain
according to \(n_1\le n_2\) and \(n_2<n_1\).

We assume \(K>1\) sufficiently large so that all estimates of the form
\(CK^{-1/2}\le1\) arising below are satisfied.

\medskip

\noindent\emph{The two-point coefficients:}
the estimates
\[
|f_{2,0}|
\le
\frac{\sqrt K}{4},
\qquad
|h_{2,0}|
\le
\frac{\sqrt K}{16}
\]
are part of the assumptions. Hence $
f_{2,0}$ and $h_{2,0}$ are in $\mathcal S_2^{(0)}(K)$. Moreover, the boundary condition \eqref{EqMeanFieldBC2} gives
\[
k_{4,0}=0.
\]
Evaluating \eqref{eq:f2_recursion} at \(i=0\), we obtain
\[
f_{2,1}
=
3g_{4,0}
+
12k_{4,0}
-
f_{2,0}^2
+
f_{2,0}.
\]
This implies that
\[
\begin{aligned}
|f_{2,1}|
&\le
3|g_{4,0}|
+
12|k_{4,0}|
+
|f_{2,0}|^2
+
|f_{2,0}|
\le
\frac{3\sqrt K}{32}
+
\frac{K}{16}
+
\frac{\sqrt K}{4}.
\end{aligned}
\]
For \(K>1\), the right-hand side is bounded by
\(K/2\). Therefore
$
f_{2,1}\in\mathcal S_2^{(1)}(K)$. Similarly, evaluating the recursion for \(h_2\) at \(i=0\) gives
\[
h_{2,1}
=
\frac12
\left(
6k_{4,0}
-
4f_{2,0}h_{2,0}
+
f_{2,0}^2
\right).
\]
Since \(k_{4,0}=0\), it follows that
\[
\begin{aligned}
|h_{2,1}|
&\le
2|f_{2,0}||h_{2,0}|
+
\frac12|f_{2,0}|^2\le
2\frac{\sqrt K}{4}\frac{\sqrt K}{16}
+
\frac12\frac{K}{16}
\le
\frac{K}{2}.
\end{aligned}
\]
Thus $
h_{2,1}\in\mathcal S_2^{(1)}(K)$.

\medskip

\noindent\emph{The zeroth-order coefficients:} we prove simultaneously, by induction on the even integer \(n\ge4\),
that
\[
g_{n,0}\in\mathcal S_n^{(0)}(K),
\qquad
k_{n,0}\in\mathcal S_n^{(0)}(K).
\]
For \(n=4\), the assumption \eqref{EqInitialBounds0}
and \(k_{4,0}=0\) imply that $
g_{4,0}$ and $k_{4,0}$ belong to $\mathcal S_4^{(0)}(K)$.\\
Let \(n\ge6\), and assume that the desired estimates hold for every
even integer \(m<n\). Since
\[
n_1+n_2=n+2,
\qquad
n_1,n_2\ge4,
\]
one necessarily has $
n_1<n$ and $
n_2<n$. The induction hypothesis is therefore applicable to every term in the
compatibility relations. Rearranging \eqref{eq:g_compatibility_zero}, we obtain
\[
g_{n,0}
=
-\frac{n}{n-4}
\sum_{\substack{n_1+n_2=n+2\\ n_1,n_2\ge4}}
g_{n_1,0}g_{n_2,0}.
\]
By the induction hypothesis and \eqref{EqInitialBounds0}, for every
even \(4\le m<n\),
\[
|g_{m,0}|
\le
\frac{K^{\frac m2-\frac32}}{2m^2}.
\]
Hence
\[
|g_{n_1,0}g_{n_2,0}|
\le
\frac{K^{\frac{n_1+n_2}{2}-3}}
{4n_1^2n_2^2}
=
\frac{K^{\frac n2-2}}
{4n_1^2n_2^2}.
\]
Consequently, \eqref{EqElementaryDiscreteSums} yields
\begin{align*}
|g_{n,0}|
&\le
\frac{n}{n-4}
\frac{K^{\frac n2-2}}{4}
\sum_{\substack{n_1+n_2=n+2\\ n_1,n_2\ge4}}
\frac1{n_1^2n_2^2}
\\
&\le
C\frac{K^{\frac n2-2}}{n^2}
\le
\frac{K^{\frac n2-\frac32}}{2n^2},
\end{align*}
and therefore $
g_{n,0}\in\mathcal S_n^{(0)}(K)$. We next consider \(k_{n,0}\). By
\eqref{eq:k_compatibility_zero}, we have 
\[
\begin{aligned}
|k_{n,0}|
&\le
\frac{2}{n-4}
\sum_{\substack{n_1+n_2=n+2\\ n_1,n_2\ge4}}
(n_1-1)
\left(
1+\frac{n_2-1}{n-1}
\right)
|k_{n_1,0}||g_{n_2,0}|
\\
&\quad
+
\frac{1}{(n-1)(n-4)}
\sum_{\substack{n_1+n_2=n+2\\ n_1,n_2\ge4}}
(n_1-1)(n_2-1)
|g_{n_1,0}||g_{n_2,0}|.
\end{aligned}
\]
Using the induction hypothesis, the estimate already obtained for
\(g_{n,0}\), and the elementary bounds
\[
n_j-1\le n_j,
\qquad
1+\frac{n_j-1}{n-1}\le2,
\]
we obtain from
\eqref{EqElementaryDiscreteSums}--\eqref{EqElementaryDiscreteSumsTwo} that
\[
\frac1{n-4}
\sum_{\substack{n_1+n_2=n+2\\n_1,n_2\ge4}}
\frac1{n_1n_2^2}
\le\frac{C}{n^2},
\qquad
\frac1{(n-1)(n-4)}
\sum_{\substack{n_1+n_2=n+2\\n_1,n_2\ge4}}
\frac1{n_1n_2}
\le\frac{C}{n^2}.
\]
which together with \eqref{EqElementaryDiscreteSums}
that
\[
|k_{n,0}|
\le
\frac{C}{\sqrt K}
\frac{K^{\frac n2-\frac32}}{2n^2}.
\]
By our choice of \(K\), we conclude that
\[
|k_{n,0}|
\le
\frac{K^{\frac n2-\frac32}}{2n^2}.
\]
Thus $
k_{n,0}\in\mathcal S_n^{(0)}(K)$. This completes the zeroth-order induction.

\medskip

\noindent\emph{The first-order coefficients.}
We now prove simultaneously, by induction on the even integer \(n\ge4\),
that
\[
g_{n,1}\in\mathcal S_n^{(1)}(K),
\qquad
k_{n,1}\in\mathcal S_n^{(1)}(K).
\]
First we consider \(n=4\). Since the summation domain in
\eqref{eq:g_compatibility_one} is empty and $n=4$, the compatibility relation reduces to
\[
2f_{2,0}g_{4,0}
+
\frac12g_{4,1}
=
0.
\]
Therefore
\[
g_{4,1}
=
-4f_{2,0}g_{4,0},
\]
which gives that
\[
|g_{4,1}|
\le
4\frac{\sqrt K}{4}\frac{\sqrt K}{32}
=
\frac{K}{32}.
\]
It follows that $
g_{4,1}\in\mathcal S_4^{(1)}(K)$. Likewise, setting \(n=4\) in
\eqref{eq:k_compatibility_one} gives
\[
2g_{4,0}
\left(
2h_{2,0}-f_{2,0}
\right)
+
8k_{4,0}f_{2,0}
+
2k_{4,1}
=
0.
\]
Since \(k_{4,0}=0\), we obtain
\[
k_{4,1}
=
-g_{4,0}
\left(
2h_{2,0}-f_{2,0}
\right).
\]
Consequently,
\[
\begin{aligned}
|k_{4,1}|\le
|g_{4,0}|
\left(
2|h_{2,0}|+|f_{2,0}|
\right)\le
\frac{\sqrt K}{32}
\left(
\frac{\sqrt K}{8}
+
\frac{\sqrt K}{4}
\right)=
\frac{3K}{256}
\le
\frac{K}{32}.
\end{aligned}
\]
Thus, we obtain that $
k_{4,1}\in\mathcal S_4^{(1)}(K)$.

We turn now to \(n\ge6\), and assume that the required first-order bounds hold
for every even integer \(m<n\). Using
\eqref{eq:g_compatibility_one}, we write
\[
g_{n,1}
=
-\frac{n}{n-2}
\left[
2
\sum_{\substack{n_1+n_2=n+2\\ n_1,n_2\ge4}}
g_{n_1,0}g_{n_2,1}
+
g_{n,0}
\left(
2f_{2,0}+1-\frac4n
\right)
\right].
\]
Therefore
\[
\begin{aligned}
|g_{n,1}|\le
\frac{2n}{n-2}
\sum_{\substack{n_1+n_2=n+2\\ n_1,n_2\ge4}}
|g_{n_1,0}||g_{n_2,1}|
+
\frac{n}{n-2}
|g_{n,0}|
\left(
2|f_{2,0}|+1
\right).
\end{aligned}
\]
Using the zeroth-order estimates, the induction hypothesis and
\eqref{EqElementaryDiscreteSums}, we obtain
\[
\frac{2n}{n-2}
\sum_{\substack{n_1+n_2=n+2\\ n_1,n_2\ge4}}
|g_{n_1,0}||g_{n_2,1}|
\le
\frac{C}{\sqrt K}
\frac{K^{\frac n2-\frac32}}{n^2}
\left(
1+\frac{nK}{2}
\right).
\]
We also have
\[
\begin{aligned}
|g_{n,0}|
\left(
2|f_{2,0}|+1
\right)
&\le
\frac{K^{\frac n2-\frac32}}{2n^2}
\left(
\frac{\sqrt K}{2}+1
\right)\le
\frac{C}{\sqrt K}
\frac{K^{\frac n2-\frac32}}{n^2}
\left(
1+\frac{nK}{2}
\right).
\end{aligned}
\]
It follows that
\[
|g_{n,1}|
\le
\frac{C}{\sqrt K}
\frac{K^{\frac n2-\frac32}}{n^2}
\left(
1+\frac{nK}{2}
\right),
\]
which implies the following
\[
|g_{n,1}|
\le
\frac{K^{\frac n2-\frac32}}{n^2}
\left(
1+\frac{nK}{2}
\right).
\]
Hence we deduce that $
g_{n,1}\in\mathcal S_n^{(1)}(K)$. Finally, we verify the inductive estimate for $k_{n,1}$. Starting from \eqref{eq:k_compatibility_one} we write
\[
\begin{aligned}
|k_{n,1}|
&\le
\frac{2}{n-2}
\sum_{\substack{n_1+n_2=n+2\\ n_1,n_2\ge4}}
\left[
n+
\frac{2(n_1-1)(n_2-1)}{n-1}
\right]
|k_{n_1,1}||g_{n_2,0}|
\\
&\quad
+
\frac{2}{(n-1)(n-2)}
\sum_{\substack{n_1+n_2=n+2\\ n_1,n_2\ge4}}
(n_1-1)(n_2-1)
|g_{n_1,0}||g_{n_2,1}|
\\
&\quad
+
\frac{2}{n-2}
|g_{n,0}|
\left(
2|h_{2,0}|+|f_{2,0}|
\right)
\\
&\quad
+
\frac{2n}{n-2}
|k_{n,0}||f_{2,0}|.
\end{aligned}
\]
Since
\[
n+
\frac{2(n_1-1)(n_2-1)}{n-1}
\le
3n,
\]
the induction hypothesis, the zeroth-order estimates, and
\eqref{EqElementaryDiscreteSums}--\eqref{EqElementaryDiscreteSumsTwo}
imply that the first two terms on the right-hand side are bounded by
\[
\frac{C}{\sqrt K}
\frac{K^{\frac n2-\frac32}}{n^2}
\left(
1+\frac{nK}{2}
\right).
\]
Furthermore, we have
\[
\begin{aligned}
|g_{n,0}|
\left(
2|h_{2,0}|+|f_{2,0}|
\right)
&\le
\frac{K^{\frac n2-\frac32}}{2n^2}
\left(
\frac{\sqrt K}{8}
+
\frac{\sqrt K}{4}
\right),\quad 
n|k_{n,0}||f_{2,0}|
&\le
\frac{K^{\frac n2-\frac32}}{2n}
\frac{\sqrt K}{4}.
\end{aligned}
\]
Both quantities are bounded by
\[
\frac{C}{\sqrt K}
\frac{K^{\frac n2-\frac32}}{n^2}
\left(
1+\frac{nK}{2}
\right).
\]
Combining these estimates gives
\[
|k_{n,1}|
\le
\frac{C}{\sqrt K}
\frac{K^{\frac n2-\frac32}}{n^2}
\left(
1+\frac{nK}{2}
\right).
\]
By our choice of \(K\), we obtain
\[
|k_{n,1}|
\le
\frac{K^{\frac n2-\frac32}}{n^2}
\left(
1+\frac{nK}{2}
\right).
\]
Therefore we deduce that $
k_{n,1}\in\mathcal S_n^{(1)}(K)$. This completes the proof.
\end{proof}
We now prove Theorem~\ref{ThmCoefficientHierarchy}. The argument is by
induction on the coefficient order, supplemented, at each fixed order,
by an induction on the number of external legs. The cases \(i=0\) and
\(i=1\) were established in
Proposition~\ref{ThmCoefficientHierarchy0}. It therefore remains to
consider \(i\ge2\).

The proof is based on the operator formulation of the coefficient
hierarchy established above. The mapping properties of the
multiplication operators
\(\mathbb M^{(1)}\), \(\widetilde{\mathbb M}^{(1)}\), and
\(\mathbb M^{(2)}\), together with those of the bilinear operators
\(\mathbb B_n^\alpha\) and \(\mathbb B_{2,n}^\alpha\), will allow us
to propagate the defining bounds of the spaces
\(\mathcal S_n(K)\) from one coefficient order to the next. 

\begin{proof}
Set
\[
K_\star
:=
\max\left\{
C_{\mathbb M},
C_{\mathbb B}^{\,2},
K_0
\right\},
\]
where \(K_0>1\) is chosen sufficiently large to absorb the finite
number of numerical factors arising when the quantitative estimates
of the preceding mapping propositions are combined in
\eqref{f2OP}--\eqref{h2OP}, \eqref{EqOperatorG}, and
\eqref{EqOperatorK}. We assume throughout that \(K\ge K_\star\).

The proof is by induction on the coefficient order, supplemented, at
each fixed order, by an induction on the number of external legs. By
Proposition~\ref{ThmCoefficientHierarchy0}, the conclusion holds at
orders \(0\) and \(1\). Let \(r\ge2\), and assume that
\begin{equation}\label{EqInductionCoefficient}
f_{2,j}\in\mathcal S_2^{(j)}(K),
\qquad
h_{2,j}\in\mathcal S_2^{(j)}(K),
\end{equation}
and
\begin{equation}\label{EqInductionCoefficientHigher}
g_{n,j}\in\mathcal S_n^{(j)}(K),
\qquad
k_{n,j}\in\mathcal S_n^{(j)}(K),
\end{equation}
for every \(j<r\) and every even integer \(n\ge4\). We prove the
corresponding bounds at order \(r\).

We first consider the two-point sector. Taking the \((r-1)\)-st
component of \eqref{f2OP} gives
\[
f_{2,r}
=
\left(
\mathbb M_2^{(2)}(g_4+4k_4)
-
\mathbb M_2^{(1)}(f_2)
+
\mathbb B_{2,2}^{\alpha_{2,2}}(f_2,f_2)
\right)_{r-1}.
\]
All coefficients entering the right-hand side have order at most
\(r-1\), and are therefore controlled by the induction hypothesis.
The multiplication terms are estimated by
Propositions~\ref{PropMultiplicationOperators} and
\ref{PropMultiplicationOperators2}. For the bilinear term, writing
\[
\bigl(
\mathbb B_{2,2}^{\alpha_{2,2}}(f_2,f_2)
\bigr)_{r-1}
=
\bigl(
\mathbb T\circ
\mathbb B_{2,2}^{\alpha_{2,2}}(f_2,f_2)
\bigr)_{r-2},
\]
Proposition~\ref{LemOperatorBnPlus2} applies since \(r\ge2\).
Combining the corresponding quantitative estimates and using the
choice of \(K_\star\), we obtain
\begin{equation}\label{F22}
f_{2,r}\in\mathcal S_2^{(r)}(K).
\end{equation}
Similarly, taking the \((r-1)\)-st component of \eqref{h2OP} gives
\[
h_{2,r}
=
\left(
\mathbb M_2^{(2)}(k_4)
-
\frac12
\mathbb B_{2,2}^{\alpha_{2,2}}
(f_2,f_2-4h_2)
\right)_{r-1}.
\]
Again, all coefficients on the right-hand side have order strictly
smaller than \(r\). The same mapping properties, together with the
choice of \(K_\star\), therefore yield
\begin{equation}\label{H22}
h_{2,r}\in\mathcal S_2^{(r)}(K).
\end{equation}

It remains to consider the higher-point sectors \(n\ge4\). At the
fixed coefficient order \(r\), we argue by induction on the even
integer \(n\ge4\). Assume that
\begin{equation}\label{EqInductionPointSector}
g_{m,r}\in\mathcal S_m^{(r)}(K),
\qquad
k_{m,r}\in\mathcal S_m^{(r)}(K)
\end{equation}
for every even integer \(4\le m<n\).

Write \(r=i+2\), with \(i\in\mathbb N_0\). Taking the \(i\)-th
component of \eqref{EqOperatorG}, we obtain
\begin{multline}
g_{n,r}
=
\left[
\mathbb M_n^{(2)}
(g_{n+2}+4k_{n+2})
\right]_i
-
\left[
\mathbb M_n^{(1)}
\circ\mathbb T(g_n)
\right]_i\\+
\left[
\mathbb T^2\circ
\mathbb B_n^{\beta_{gg}}(g,g)
\right]_i
+
2\left[
\mathbb T\circ
\mathbb B_{2,n}^{\beta_{gf}}(f_2,g_n)
\right]_i.
\label{EqProofGnComponent}
\end{multline}
Since \(i=r-2<r\), the induction hypothesis on the coefficient order
gives
\[
g_{n+2,i},\,k_{n+2,i}
\in
\mathcal S_{n+2}^{(i)}(K),
\]
and Proposition~\ref{PropMultiplicationOperators} applies to the first
term. Similarly, \(i+1=r-1<r\), so that
\[
g_{n,i+1}\in\mathcal S_n^{(i+1)}(K),
\]
and the second term is controlled by
Proposition~\ref{PropMultiplicationOperators2}.

For \(n\ge6\), the induction hypotheses and
Proposition~\ref{LemOperatorBnPlusTwo} give
\[
\left[
\mathbb T^2\circ
\mathbb B_n^{\beta_{gg}}(g,g)
\right]_i
\in
\mathcal S_n^{(r)}(K).
\]
For \(n=4\), this contribution is absent. Finally, since
\(i+1=r-1<r\), Proposition~\ref{LemOperatorBnPlus2} yields
\[
\left[
\mathbb T\circ
\mathbb B_{2,n}^{\beta_{gf}}(f_2,g_n)
\right]_i
\in
\mathcal S_n^{(r)}(K).
\]
Applying these mapping properties to the operator identity
\eqref{EqProofGnComponent}, and using the choice of \(K_\star\) to
absorb its fixed numerical coefficients, we conclude that $
g_{n,r}\in\mathcal S_n^{(r)}(K)$.

We next turn to \(k_n\). Taking the \(i\)-th component of
\eqref{EqOperatorK} gives
\begin{align}
k_{n,r}
={}&
\left[
\mathbb M_n^{(2)}(k_{n+2})
\right]_i
-
\left[
\widetilde{\mathbb M}_n^{(1)}
\circ\mathbb T(k_n)
\right]_i+
\left[
\mathbb T^2\circ
\mathbb B_n^{\alpha_{gg}}(g,g)
\right]_i
+
2\left[
\mathbb T^2\circ
\mathbb B_n^{\alpha_{kg}}(k,g)
\right]_i
\nonumber\\
&+
2\left[
\mathbb T\circ
\mathbb B_{2,n}^{\alpha_g}
(2h_2-f_2,g_n)
\right]_i
+
2\left[
\mathbb T\circ
\mathbb B_{2,n}^{\alpha_k}
(f_2,k_n)
\right]_i.
\label{EqProofKnComponent}
\end{align}
Since \(i=r-2<r\), the induction hypothesis on the coefficient order
gives that $
k_{n+2,i}$ belongs to $\mathcal S_{n+2}^{(i)}(K)$,
and Proposition~\ref{PropMultiplicationOperators} applies to the first
term. Similarly, since \(i+1=r-1<r\),
\[
k_{n,i+1}\in\mathcal S_n^{(i+1)}(K),
\]
and Proposition~\ref{PropMultiplicationOperators2Tilde} applies to the
second term.

For \(n\ge6\), the induction hypotheses and
Proposition~\ref{LemOperatorBnPlusTwo} yield
\[
\left[
\mathbb T^2\circ
\mathbb B_n^{\alpha_{gg}}(g,g)
\right]_i,
\quad
\left[
\mathbb T^2\circ
\mathbb B_n^{\alpha_{kg}}(k,g)
\right]_i
\in
\mathcal S_n^{(r)}(K).
\]
For \(n=4\), these contributions are absent.

Finally, the two terms involving the two-point sector contain
convolution coefficients of total order \(i+1=r-1\). Hence the
induction hypothesis on the coefficient order and
Proposition~\ref{LemOperatorBnPlus2} give
\[
\left[
\mathbb T\circ
\mathbb B_{2,n}^{\alpha_g}
(2h_2-f_2,g_n)
\right]_i,
\quad
\left[
\mathbb T\circ
\mathbb B_{2,n}^{\alpha_k}
(f_2,k_n)
\right]_i
\in
\mathcal S_n^{(r)}(K).
\]
Applying these mapping properties to the operator identity
\eqref{EqProofKnComponent}, and using the choice of \(K_\star\) to
absorb its fixed numerical coefficients, we conclude that
\[
k_{n,r}\in\mathcal S_n^{(r)}(K).
\]
This closes the induction on \(n\). Thus, for every even \(n\ge4\),
\[
g_{n,r}\in\mathcal S_n^{(r)}(K),
\qquad
k_{n,r}\in\mathcal S_n^{(r)}(K).
\]
Together with \eqref{F22} and \eqref{H22}, this closes the induction on \(r\) and completes the
proof. 
\end{proof}

\section{Construction of a smooth trivial solution to the mean-field
hierarchy}\label{sec7}

We now use the coefficient estimates established in the previous
section to construct a smooth solution of the mean-field hierarchy
\eqref{EqSequenceF2Mu}--\eqref{EqSequenceHnMu} which vanishes in the
ultraviolet limit \(\mu_{\max}\to+\infty\), equivalently
\(\alpha_0\to0\).

\begin{proposition}\label{PropSmoothTrivialSolution}
There exist smooth solutions \(f_n(\mu)\) and \(h_n(\mu)\) of
\eqref{EqSequenceF2Mu}--\eqref{EqSequenceHnMu}, satisfying the boundary
conditions \eqref{EqMeanFieldBC}--\eqref{EqMeanFieldBC2}, with the
following properties.
\begin{enumerate}
\item
The functions \(f_2\) and \(h_2\) are well defined on
\([0,\mu_{\max}]\), and, for every \(l\in\mathbb N_0\),
\begin{equation}\label{EqVanishingTwoPointSector}
\lim_{\mu_{\max}\to+\infty}
\partial_\mu^l f_2(\mu_{\max})
=
\lim_{\mu_{\max}\to+\infty}
\partial_\mu^l h_2(\mu_{\max})
=
0.
\end{equation}

\item
For every even integer \(n\ge4\) and every \(l\in\mathbb N_0\), the
functions \(\partial_\mu^l f_n\) and \(\partial_\mu^l h_n\) are well
defined on \([0,\mu_{\max}]\), and
\begin{equation}\label{EqVanishingHigherPointSector}
\lim_{\mu_{\max}\to+\infty}
\partial_\mu^l f_n(\mu_{\max})
=
\lim_{\mu_{\max}\to+\infty}
\partial_\mu^l h_n(\mu_{\max})
=
0.
\end{equation}
\end{enumerate}
\end{proposition}
\noindent First, we recall the following estimate from
\cite[Proposition~3.4]{KopperWang2025}.

\begin{proposition}\label{PropDerivativeRationalKernel}
Let $
x_n=n\mu$. For every \(l\in\mathbb N_0\),  \(n\ge l+1\), and
\(\mu\in[0,\mu_{\max}]\), one has
\begin{equation}\label{EqDerivativeRationalKernel}
\left|
\partial_\mu^l
\frac{x_n^{\,n-1}}{1+x_n^n}
\right|
\le
C_l
\frac{n^{\,n+l-1}\mu^{\,n-l-1}}{1+x_n^n},
\end{equation}
where \(C_0=1\) and
\begin{equation}\label{EqConstantsCl}
C_{l+1}
=
1+
\sum_{j=0}^{l}
\binom{l+1}{j}C_j
\le
4^{\,l+1}(l+1)!.
\end{equation}
\end{proposition}

\begin{proof}
See the proof of [Proposition~3.4] in \cite{KopperWang2025}.
\end{proof}

The construction follows the argument of \cite{KopperWang2025}, with
the difference that both the momentum-independent sector \(f_n\) and
the quadratic momentum sector \(h_n\) have to be considered. The
argument applies to both families because it only uses the estimates
on their Taylor coefficients established in
Theorem~\ref{ThmCoefficientHierarchy}.

We start from the following ansatz for the two-point functions
\begin{equation}\label{EqAnsatzF2}
f_2(\mu)
=
\sum_{n\ge1}
b_n\frac{x_n^{\,n-1}}{1+x_n^n},
\qquad
x_n=n\mu,
\end{equation}
and
\begin{equation}\label{EqAnsatzH2}
h_2(\mu)
=
\sum_{n\ge1}
d_n\frac{x_n^{\,n-1}}{1+x_n^n}.
\end{equation}
Once these functions have been constructed, the higher-point functions
are determined recursively by
\eqref{EqSequenceFnMu}--\eqref{EqSequenceHnMu}. Expanding formally in powers of \(\mu\), we write
\[
f_2(\mu)
=
\sum_{i\ge0}f_{2,i}\mu^i,
\qquad
h_2(\mu)
=
\sum_{i\ge0}h_{2,i}\mu^i.
\]
The coefficients are related to the sequences \((b_n)_{n\ge0}\) and
\((d_n)_{n\ge0}\) through
\begin{equation}\label{EqTaylorCoefficientsF2}
f_{2,i}
=
(i+1)^i
\sum_{\rho=1}^{i+1}
b_{\left\{\frac{i+1}{\rho}\right\}}
(-1)^{\rho-1}\frac{1}{\rho^i},
\qquad
i\in\mathbb N_0,
\end{equation}
and
\begin{equation}\label{EqTaylorCoefficientsH2}
h_{2,i}
=
(i+1)^i
\sum_{\rho=1}^{i+1}
d_{\left\{\frac{i+1}{\rho}\right\}}
(-1)^{\rho-1}\frac{1}{\rho^i},
\qquad
i\in\mathbb N_0.
\end{equation}
Here, \(b_0=d_0:=0\), and, for \(n,m\in\mathbb N\), we use the
convention
\begin{equation}\label{EqIntegerRatioConvention}
\left\{\frac{n}{m}\right\}
:=
\begin{cases}
\dfrac{n}{m},
& \text{if }\dfrac{n}{m}\in\mathbb N,\\[6pt]
0,
& \text{otherwise}.
\end{cases}
\end{equation}
In particular,
\begin{equation}\label{EqFirstTaylorCoefficientsF2}
f_{2,0}=b_1,
\qquad
f_{2,1}=2b_2-b_1,
\end{equation}
and
\begin{equation}\label{EqFirstTaylorCoefficientsH2}
h_{2,0}=d_1,
\qquad
h_{2,1}=2d_2-d_1.
\end{equation}
Isolating the term corresponding to \(\rho=1\) in
\eqref{EqTaylorCoefficientsF2} and
\eqref{EqTaylorCoefficientsH2}, we obtain
\begin{equation}\label{EqRecursionBi}
b_{i+1}
=
\frac{f_{2,i}}{(i+1)^i}
-
\sum_{\rho=2}^{i+1}
b_{\left\{\frac{i+1}{\rho}\right\}}
(-1)^{\rho-1}\frac{1}{\rho^i},
\end{equation}
and
\begin{equation}\label{EqRecursionDi}
d_{i+1}
=
\frac{h_{2,i}}{(i+1)^i}
-
\sum_{\rho=2}^{i+1}
d_{\left\{\frac{i+1}{\rho}\right\}}
(-1)^{\rho-1}\frac{1}{\rho^i}.
\end{equation}
For \(i\ge1\), set
\[
c_i
:=
K^{i+\frac12}
\frac{|i-3|!}
{((i-1)!)^{1/4}(i+1)^i}.
\]
The argument of
\cite[Proposition~3.3]{KopperWang2025} shows that a sequence
\((a_i)_{i\ge0}\), with \(a_0=0\), satisfying
\[
|a_{i+1}|
\le
c_i
+
\sum_{\rho=2}^{i+1}
\left|
a_{\left\{\frac{i+1}{\rho}\right\}}
\right|
\frac1{\rho^i}
\]
satisfies
\[
|a_i|
\le
C(K)\frac{i^2}{2^i},
\qquad i\ge1.
\]
On the other hand, Theorem~\ref{ThmCoefficientHierarchy} gives
\[
|f_{2,i}|,\ |h_{2,i}|
\le
K^{i+\frac12}
\frac{|i-3|!}{((i-1)!)^{1/4}},
\qquad i\ge2.
\]
Hence \eqref{EqRecursionBi} and \eqref{EqRecursionDi}, after treating
the finitely many low-order coefficients separately, yield
\begin{equation}\label{EqBoundsBiDi}
|b_i|+|d_i|
\le
C(K)\frac{i^2}{2^i},
\qquad i\ge1.
\end{equation}
Now we prove Proposition~\ref{PropSmoothTrivialSolution}.
\begin{proof}[Proof of Proposition~\ref{PropSmoothTrivialSolution}]
By Proposition~\ref{PropDerivativeRationalKernel}, for
\(l\in\mathbb N_0\), \(n\ge l+1\), and \(\mu\ge0\),
\begin{equation}\label{EqDerivativeBoundKernel}
\left|
\partial_\mu^l
\frac{x_n^{\,n-1}}{1+x_n^n}
\right|
\le
C_l n^{2l}
\frac{x_n^{\,n-l-1}}{1+x_n^n}.
\end{equation}
For \(0\le l<n-1\), set
\[
G_{n,l}(\mu)
:=
\frac{(n\mu)^{n-l-1}}{1+(n\mu)^n},
\qquad
\mu\ge0.
\]
This function attains its maximum at
\[
\widetilde{\mu}
=
\frac1n
\left(
\frac{n}{l+1}-1
\right)^{1/n},
\]
and
\begin{equation}\label{EqMaximumKernel}
\|G_{n,l}\|_\infty
=
\frac{l+1}{n}
\left(
\frac{n}{l+1}-1
\right)^{1-\frac{l+1}{n}}
\le1.
\end{equation}
Hence, for every fixed \(l\), the \(n\)-th terms in the series
defining \(\partial_\mu^l f_2\) and \(\partial_\mu^l h_2\) are bounded,
for \(n>l+1\), by
\[
C(K)C_l\frac{n^{2l+2}}{2^n}.
\]
The finitely many remaining terms are bounded separately. Since
\[
\sum_{n\ge1}\frac{n^{2l+2}}{2^n}<\infty,
\]
\eqref{EqBoundsBiDi} implies that
\eqref{EqAnsatzF2}--\eqref{EqAnsatzH2}, together with all their
termwise derivatives, converge uniformly on \(\mathbb{R}^+\). Thus $
f_2$ and $h_2$ are in $ C^\infty(\mathbb{R}^+)$. For every fixed \(m\ge1\) and \(l\in\mathbb N_0\),
\[
\lim_{\mu\to+\infty}
\partial_\mu^l
\frac{(m\mu)^{m-1}}{1+(m\mu)^m}
=
0.
\]
The preceding summable majorant is independent of \(\mu\), and
dominated convergence therefore yields
\[
\lim_{\mu_{\max}\to+\infty}
\partial_\mu^l f_2(\mu_{\max})
=
\lim_{\mu_{\max}\to+\infty}
\partial_\mu^l h_2(\mu_{\max})
=
0,
\qquad
l\in\mathbb N_0,
\]
and this proves \eqref{EqVanishingTwoPointSector}.

It remains to construct the higher-point functions. We proceed by
induction on the number of external legs $n$, with the induction
hypothesis imposed simultaneously for all derivative orders. The
two-point sector established above provides the initial step.
Assume that the functions \(f_m\) and \(h_m\) have been constructed
for every even \(m\le n\), that they are smooth, and that
\[
\lim_{\mu_{\max}\to+\infty}
\partial_\mu^l f_m(\mu_{\max})
=
\lim_{\mu_{\max}\to+\infty}
\partial_\mu^l h_m(\mu_{\max})
=
0,
\qquad l\in\mathbb N_0.
\]
The recursion \eqref{EqSequenceHnMu} defines \(h_{n+2}\) in terms of
the already constructed lower-point sectors. Since these functions
are smooth, so is \(h_{n+2}\). Differentiating
\eqref{EqSequenceHnMu} \(l\) times gives
\begin{multline}\label{EqSequenceHnMuDerivative}
\partial_\mu^l h_{n+2}
=
\frac{1}{n(n+1)}
\sum_{\substack{n_1+n_2=n\\ n_1,n_2\ge3}}
\;
\sum_{\substack{l_1+l_2=l\\ l_1,l_2\ge0}}
\binom{l}{l_1}
\Biggl[
n_1
\left(
1+\frac{n_2}{n-1}
\right)
\partial_\mu^{l_1}h_{n_1+1}\,
\partial_\mu^{l_2}f_{n_2+1}
\\
+
n_2
\left(
1+\frac{n_1}{n-1}
\right)
\partial_\mu^{l_1}f_{n_1+1}\,
\partial_\mu^{l_2}h_{n_2+1}
\Biggr]
\\
-
\frac{1}{n(n-1)(n+1)}
\sum_{\substack{n_1+n_2=n\\ n_1,n_2\ge3}}
\;
\sum_{\substack{l_1+l_2=l\\ l_1,l_2\ge0}}
\binom{l}{l_1}
n_1n_2\,
\partial_\mu^{l_1}f_{n_1+1}\,
\partial_\mu^{l_2}f_{n_2+1}
\\
+
\frac{4}{n(n+1)}
\sum_{\substack{l_1+l_2=l\\ l_1,l_2\ge0}}
\binom{l}{l_1}
\partial_\mu^{l_1}f_n\,
\partial_\mu^{l_2}h_2
-
\frac{2}{n(n+1)}
\sum_{\substack{l_1+l_2=l\\ l_1,l_2\ge0}}
\binom{l}{l_1}
\partial_\mu^{l_1}f_n\,
\partial_\mu^{l_2}f_2
\\
+
\frac{2}{n+1}
\sum_{\substack{l_1+l_2=l\\ l_1,l_2\ge0}}
\binom{l}{l_1}
\partial_\mu^{l_1}h_n\,
\partial_\mu^{l_2}f_2
+
\frac{2}{n(n+1)}
\partial_\mu^{l+1}h_n.
\end{multline}
Every function appearing on the right-hand side belongs to a
lower-point sector. The induction hypothesis therefore implies that,
at \(\mu=\mu_{\max}\), each of the corresponding derivatives vanishes
as \(\mu_{\max}\to+\infty\). It follows from
\eqref{EqSequenceHnMuDerivative} that
\[
\lim_{\mu_{\max}\to+\infty}
\partial_\mu^l h_{n+2}(\mu_{\max})
=
0,
\qquad l\in\mathbb N_0.
\]
The recursion \eqref{EqSequenceFnMu} now defines \(f_{n+2}\), since
\(h_{n+2}\) has already been constructed. Differentiating \(l\) times
gives
\begin{multline}\label{EqSequenceFnMuDerivative}
\partial_\mu^l f_{n+2}
+
4\partial_\mu^l h_{n+2}
=
\frac{1}{n+1}
\sum_{\substack{n_1+n_2=n\\ n_1,n_2\ge3}}
\;
\sum_{\substack{l_1+l_2=l\\ l_1,l_2\ge0}}
\binom{l}{l_1}
\partial_\mu^{l_1}f_{n_1+1}\,
\partial_\mu^{l_2}f_{n_2+1}
\\
+
\frac{2}{n+1}
\sum_{\substack{l_1+l_2=l\\ l_1,l_2\ge0}}
\binom{l}{l_1}
\partial_\mu^{l_1}f_2\,
\partial_\mu^{l_2}f_n
+
\frac{1}{n+1}
\left(
1-\frac4n
\right)
\partial_\mu^l f_n
+
\frac{2}{n(n+1)}
\partial_\mu^{l+1}f_n.
\end{multline}
All terms on the right-hand side involve lower-point sectors and
therefore vanish at \(\mu=\mu_{\max}\) by the induction hypothesis.
Together with the result just obtained for \(h_{n+2}\), this gives
\[
\lim_{\mu_{\max}\to+\infty}
\partial_\mu^l f_{n+2}(\mu_{\max})
=
0,
\qquad l\in\mathbb N_0.
\]
Thus \(f_{n+2}\) and \(h_{n+2}\) have the required properties. This completes the
induction, and the proof follows.
\end{proof}
We conclude by combining the preceding results to prove the central theorem \ref{ThmSecondOrderTriviality} of this paper that gives the
triviality of the solutions to the second-order mean-field flow
equations.
\begin{proof}
Fix the boundary data
\[
0<c_4<\infty,
\qquad
|c_2|<\infty,
\qquad
|c_2'|<\infty.
\]
Since these quantities are finite, \(K>1\) may be chosen sufficiently
large so that the initial bounds
\eqref{EqInitialBounds0} and all the hypotheses of the mapping
propositions are satisfied.

The compatibility relations derived in
Section~\ref{SecDiscreteHierarchy}, together with
Theorem~\ref{ThmCoefficientHierarchy}, then yield a solution of the
discrete coefficient hierarchy satisfying
\[
f_{2,i},h_{2,i}\in\mathcal S_2^{(i)}(K),
\qquad
g_{n,i},k_{n,i}\in\mathcal S_n^{(i)}(K),
\]
for every \(i\in\mathbb N_0\) and every even \(n\ge4\).

Proposition~\ref{PropSmoothTrivialSolution} reconstructs from these
coefficients smooth solutions \(f_n,h_n\) of the dimensionless
mean-field hierarchy satisfying the prescribed boundary conditions.
Moreover, for every even \(n\ge2\) and every \(l\in\mathbb N_0\),
\[
\lim_{\mu_{\max}\to+\infty}
\partial_\mu^l f_n(\mu_{\max})
=
\lim_{\mu_{\max}\to+\infty}
\partial_\mu^l h_n(\mu_{\max})
=
0.
\]
Since \(\mu_{\max}\to+\infty\) is equivalent to
\(\alpha_0\to0\), this is precisely the ultraviolet limit.

Finally, remembering the rescaling \eqref{EqRescalingAB}, which relates the dimensionless families \(f_n,h_n\) to the mean-field coefficients \(A_n,B_n\) yields
the corresponding solution of the second-order mean-field flow
hierarchy. The momentum-independent and quadratic momentum sectors
therefore vanish in the ultraviolet limit. This proves the claim.
\end{proof}
\section*{Discussion}

In this paper, we have introduced and analysed a second-order
mean-field reduction of the Wilson--Polchinski hierarchy for the
four-dimensional Euclidean \(\phi^4\) theory. On the symmetric momentum
configuration
\[
(p,-p,\ldots,p,-p),
\]
the connected amputated Green functions were decomposed into a
momentum-independent component, a quadratic momentum component, and a
complementary remainder. The first two components form a closed
hierarchy which retains, in particular, the mass, coupling, and
wave-function sectors of the theory.

After passing to dimensionless variables, we analysed this hierarchy
through its Taylor coefficients. The weighted sequence spaces
introduced above provide uniform control of the linear and bilinear
operators entering the coefficient recursion. This yields a
non-perturbative construction of the coefficient hierarchy, with
bounds uniform in both the number of external legs and the coefficient
order. An explicit construction of the two-point functions then gives
smooth solutions of the functional hierarchy satisfying, for every
even \(n\ge2\) and every \(l\in\mathbb N_0\),
\[
\lim_{\mu_{\max}\to+\infty}
\partial_\mu^l f_n(\mu_{\max})
=
\lim_{\mu_{\max}\to+\infty}
\partial_\mu^l h_n(\mu_{\max})
=
0.
\]
Thus both the momentum-independent and quadratic momentum sectors
approach the Gaussian fixed point in the ultraviolet limit. In
particular, ultraviolet triviality holds for the second-order
mean-field hierarchy, including the quadratic two-point component
associated with wave-function renormalization.

The present result does not address the ultraviolet behaviour of the
complementary remainder \(R_n\). Although \(R_n\) satisfies an exact
flow equation, its structure is considerably richer than that of the
closed second-order mean-field hierarchy, and no properties of this remainder are established here. A natural continuation of the present
work is instead to study the perturbative expansion of \(R_n\)
together with those of the mean-field coefficients \(A_n\) and \(B_n\),
and to derive quantitative bounds relating these three sectors. Such
estimates would provide a first step towards understanding whether the
triviality established here for the second-order mean-field hierarchy can be lifted
to the full restricted correlators. In conjunction with suitable
summability properties, they may also provide a non-perturbative
reconstruction from the corresponding perturbative expansions.

A second direction concerns the infrared regularization. Kopper and
Wang showed in the mean-field setting that the flow-equation analysis
can be carried out with a physical mass in place of the auxiliary
infrared cutoff \cite{KopperWang2025}. The close structural relation
between their hierarchy and the system obtained here suggests that
the same strategy should extend to the second-order mean-field
hierarchy. Establishing this extension would provide a natural
physical refinement of the present construction.

Finally, we expect the smooth trivial solution constructed above to be
unique within the class determined by the prescribed boundary
conditions and coefficient bounds. The recursive structure of the
hierarchy suggests a connection with the uniqueness argument of
\cite{KopperWang2025}, but uniqueness is not required for the existence
and triviality results proved here. Its analysis, together with the
perturbative study of the remainder, is left for future work.
\section*{Statements and Declarations}

\subsection*{Funding}
The author received no funding for this work.

\subsection*{Conflict of Interest}
The author declares that she has no competing interests.
\subsection*{Data Availability}
This article is a theoretical study. No datasets were generated, analysed, or used during the current study.

\newpage

\appendix
\section*{Appendix: Algebraic estimates}\label{appendix}
\noindent We shall repeatedly use the following consequence of Gautschi's
inequality; see \cite[Eq.~(5.6.4)]{NISTHandbook}.

\begin{lemma}\label{LemGammaRatio}
For every \(x\ge1\) and every \(\delta\ge0\),
\begin{equation}\label{EqGammaRatioGeneral}
\frac{\Gamma(x)}{\Gamma(x+\delta)}
\le
2x^{-\delta}.
\end{equation}
\end{lemma}

\begin{proof}
Write
\[
\delta=m+s,
\qquad
m\in\mathbb N_0,
\qquad
0\le s<1.
\]
If \(s=0\), the Gamma recurrence gives
\[
\frac{\Gamma(x)}{\Gamma(x+m)}
=
\prod_{r=0}^{m-1}\frac{1}{x+r}
\le
x^{-m}.
\]
If \(0<s<1\), Gautschi's inequality
\cite[Eq.~(5.6.4)]{NISTHandbook} yields
\[
\frac{\Gamma(x)}{\Gamma(x+s)}
<
\frac{(x+1)^{1-s}}{x}
\le
2x^{-s},
\qquad x\ge1.
\]
Consequently,
\[
\frac{\Gamma(x)}{\Gamma(x+m+s)}
=
\frac{\Gamma(x)}{\Gamma(x+s)}
\prod_{r=0}^{m-1}\frac{1}{x+s+r}
\le
2x^{-s}x^{-m}
=
2x^{-\delta}.
\]
\end{proof}
\begin{lemma}[Kopper--Wang {\cite[Lemma~3.4]{KopperWang2025}}]
\label{LemKopperWangFactorialSum}
Let \(n_1,n_2\in2\mathbb N\setminus\{2\}\),
\(k,a\in\mathbb N_0\), and \(l\in(0,1)\). Assume that
\begin{equation}\label{EqKopperWangConditions}
\max\left\{
3-\frac{n_1}{4},
3-\frac{n_2}{4}
\right\}
\le
a
\le
\frac{k+2}{2}.
\end{equation}
Define
\begin{equation}\label{EqDefinitionF}
F(n_1,n_2,k,a,a,l)
:=
\sum_{\nu=a}^{k+2-a}
\frac{
\left(
\frac{n_1}{4}+\nu-3
\right)!
\left(
\frac{n_2}{4}+k-\nu-1
\right)!
}{
\bigl[\nu!(k+2-\nu)!\bigr]^l
}.
\end{equation}
Then
\begin{equation}\label{EqKopperWangFactorialEstimate}
F(n_1,n_2,k,a,a,l)
\le
\bigl[a!(k+2-a)!\bigr]^{1-l}
\frac{1}{
\frac{n_1+n_2}{4}+2a-5
}
\frac{
\left(
\frac{n_1+n_2}{4}+k-3
\right)!
}{
(k+2-2a)!
}.
\end{equation}
\end{lemma}

The preceding estimate yields the two convolution bounds used in the
proof of Proposition~\ref{LemOperatorBnPlusTwo}.

\begin{corollary}\label{CorKopperWangEstimates}
Let \(n\ge6\) be even. There exist constants \(C_0,C_1>0\),
independent of \(n\) and \(k\in\mathbb N_0\), such that
\begin{align}
\frac{n}{n+2k}
\sum_{\substack{
n_1+n_2=n+2\\
4\le n_1\le10\\
n_1\le n_2
}}
F\left(
n_1,n_2,k,2,2,\frac14
\right)
&\le
C_1
\frac{
\left(\frac n4+k-1\right)!
}{
[(k+2)!]^{1/4}
},
\label{EqF22}
\\
\frac{n}{n+2k}
\sum_{\substack{
n_1+n_2=n+2\\
12\le n_1\le n_2
}}
F\left(
n_1,n_2,k,0,0,\frac14
\right)
&\le
C_0
\frac{
\left(\frac n4+k-1\right)!
}{
[(k+2)!]^{1/4}
}.
\label{EqF00}
\end{align}
Empty sums are understood to vanish.
\end{corollary}

\begin{proof}
We first consider \eqref{EqF22}. For \(k=0,1\), the sum defining
\(F(n_1,n_2,k,2,2,\frac14)\) is empty, so there is nothing to prove.
Let therefore \(k\ge2\). Since
\[
4\le n_1\le10,
\qquad
n_1\le n_2,
\]
the assumptions of Lemma~\ref{LemKopperWangFactorialSum} are satisfied
with \(a=2\) and \(l=\frac14\). Using \(n_1+n_2=n+2\), we obtain
\begin{equation}\label{EqProofF22Intermediate}
F\left(
n_1,n_2,k,2,2,\frac14
\right)
\le
(2k!)^{3/4}
\frac{1}{\frac n4-\frac12}
\frac{
\left(
\frac n4+k-\frac52
\right)!
}{
(k-2)!
}.
\end{equation}
The sum over \(n_1\) contains at most four terms. Set
\[
x:=\frac n4+k.
\]
By the standard Gamma-ratio estimate, there exists a universal
constant \(C>0\) such that
\[
\frac{
\Gamma\!\left(x-\frac32\right)
}{
\Gamma(x)
}
\le
C x^{-3/2}.
\]
After dividing the right-hand side of
\eqref{EqProofF22Intermediate} by the desired weight in
\eqref{EqF22}, it therefore remains to bound, up to a universal
constant,
\[
\frac{n}{(n+2k)(n-2)}
\frac{
(k!)^{3/4}((k+2)!)^{1/4}
}{
(k-2)!
}
x^{-3/2}.
\]
Since
\[
\frac{
(k!)^{3/4}((k+2)!)^{1/4}
}{
(k-2)!
}
=
k(k-1)\bigl((k+1)(k+2)\bigr)^{1/4}
\le
(k+2)^{5/2},
\]
while
\[
\frac{n}{n-2}\le\frac32,
\qquad
k+2\le n+2k,
\qquad
x=\frac n4+k\ge\frac14(n+2k),
\]
the preceding quantity is bounded uniformly in \(n\) and \(k\).
This proves \eqref{EqF22}.

We next prove \eqref{EqF00}. The sum is empty unless \(n\ge22\).
For \(n\ge22\), the assumptions of
Lemma~\ref{LemKopperWangFactorialSum} are satisfied with
\(a=0\) and \(l=\frac14\). Hence
\begin{equation}\label{EqProofF00Intermediate}
F\left(
n_1,n_2,k,0,0,\frac14
\right)
\le
\frac{4}{n-18}
\frac{
\left(
\frac n4+k-\frac52
\right)!
}{
[(k+2)!]^{1/4}
}.
\end{equation}
The number \(N_n\) of admissible values of \(n_1\) satisfies
\[
N_n\le\frac n4.
\]
Using again the Gamma-ratio estimate with
\(x=\frac n4+k\), we obtain
\[
\frac{
\left(
\frac n4+k-\frac52
\right)!
}{
\left(
\frac n4+k-1
\right)!
}
=
\frac{
\Gamma\!\left(x-\frac32\right)
}{
\Gamma(x)
}
\le
C x^{-3/2}.
\]
Consequently,
\[
\frac{n}{n+2k}
\frac{4N_n}{n-18}
\frac{
\left(
\frac n4+k-\frac52
\right)!
}{
\left(
\frac n4+k-1
\right)!
}
\le
C
\frac{n}{n+2k}
\frac{n}{n-18}
x^{-3/2}.
\]
The right-hand side is bounded uniformly for \(n\ge22\) and
\(k\in\mathbb N_0\). This proves \eqref{EqF00}.
\end{proof}

\newpage
\section*{References}
\bibliographystyle{abbrv}
\bibliography{aipsamp}

@book{NISTHandbook,
  editor    = {Olver, Frank W. J. and Lozier, Daniel W. and
               Boisvert, Ronald F. and Clark, Charles W.},
  title     = {NIST Handbook of Mathematical Functions},
  publisher = {Cambridge University Press},
  address   = {Cambridge},
  year      = {2010}
}

@article{WilsonKogut1974,
author  = {Wilson, Kenneth G. and Kogut, John},
title   = {The Renormalization Group and the $\varepsilon$ Expansion},
journal = {Physics Reports},
volume  = {12},
number  = {2},
pages   = {75--199},
year    = {1974},
doi     = {10.1016/0370-1573(74)90023-4}
}

@article{Aizenman1981,
author  = {Aizenman, Michael},
title   = {Proof of the Triviality of $\phi_d^4$ Field Theory and
Some Mean-Field Features of Ising Models for $d>4$},
journal = {Physical Review Letters},
volume  = {47},
number  = {1},
pages   = {1--4},
year    = {1981},
doi     = {10.1103/PhysRevLett.47.1}
}

@article{Frohlich1982,
author  = {Fr{\"o}hlich, J{"u}rg},
title   = {On the Triviality of $\lambda\phi_d^4$ Theories and the
Approach to the Critical Point in $d\geq4$ Dimensions},
journal = {Nuclear Physics B},
volume  = {200},
number  = {2},
pages   = {281--296},
year    = {1982},
doi     = {10.1016/0550-3213(82)90088-8}
}

@article{Polchinski1984,
author  = {Polchinski, Joseph},
title   = {Renormalization and Effective Lagrangians},
journal = {Nuclear Physics B},
volume  = {231},
number  = {2},
pages   = {269--295},
year    = {1984},
doi     = {10.1016/0550-3213(84)90287-6}
}

@article{Gallavotti1985,
author  = {Gallavotti, Giovanni},
title   = {Renormalization Theory and Ultraviolet Stability for
Scalar Fields via Renormalization Group Methods},
journal = {Reviews of Modern Physics},
volume  = {57},
number  = {2},
pages   = {471--562},
year    = {1985},
doi     = {10.1103/RevModPhys.57.471}
}

@article{AizenmanDuminilCopin2021,
author  = {Aizenman, Michael and Duminil-Copin, Hugo},
title   = {Marginal Triviality of the Scaling Limits of Critical
Four-Dimensional Ising and $\phi_4^4$ Models},
journal = {Annals of Mathematics},
volume  = {194},
number  = {1},
pages   = {163--235},
year    = {2021},
doi     = {10.4007/annals.2021.194.1.3}
}

@article{AizenmanDuminilCopin2024,
author  = {Aizenman, Michael and Duminil-Copin, Hugo},
title   = {Corrigendum: Marginal Triviality of the Scaling Limits
of Critical Four-Dimensional Ising and $\phi_4^4$ Models},
journal = {Annals of Mathematics},
volume  = {199},
number  = {1},
pages   = {479},
year    = {2024},
doi     = {10.4007/annals.2024.199.1.7}
}

@article{Kopper2022,
author  = {Kopper, Christoph},
title   = {Asymptotically Free Solutions of the Scalar Mean Field
Flow Equations},
journal = {Annales Henri Poincar{'e}},
volume  = {23},
number  = {10},
pages   = {3453--3492},
year    = {2022},
doi     = {10.1007/s00023-022-01194-w},
eprint  = {1912.08183},
archivePrefix = {arXiv},
primaryClass  = {math-ph}
}

@article{KopperWang2025,
author  = {Kopper, Christoph and Wang, Pierre},
title   = {Triviality Proof for Mean-Field $\phi^4$ Theories in
Four Dimensions},
journal = {Journal of Mathematical Physics},
volume  = {66},
number  = {6},
pages   = {062302},
year    = {2025},
doi     = {10.1063/5.0231068},
eprint  = {2407.01309},
archivePrefix = {arXiv},
primaryClass  = {math-ph}
}

@misc{KopperWangPerturbation,
author        = {Kopper, Christoph and Wang, Pierre},
title         = {Triviality versus Perturbation Theory:
An Analysis for Mean-Field $\phi^4$ Theory in
Four Dimensions},
year          = {2025},
eprint        = {2511.04509},
archivePrefix = {arXiv},
primaryClass  = {math-ph}
}

\end{document}